\documentclass[11pt]{article}

\usepackage{tcolorbox}

\usepackage[margin=1in]{geometry}
\usepackage{amsmath,array,bm}
\usepackage{amsmath,amsfonts}
\usepackage{amsthm}
\usepackage{amssymb,latexsym}
\usepackage{algorithm}
\usepackage{subcaption}
\usepackage{algorithmicx}
\usepackage{marginnote}
\usepackage[]{algpseudocode}
\usepackage{float}
\usepackage{thmtools}
\usepackage{mathtools}
\usepackage{thm-restate}
\usepackage{bussproofs}
\usepackage[square,sort,comma,numbers]{natbib}
\usepackage{tikz}
\usepackage{xspace}
\usepackage{xcolor}
\usepackage{tipa}
\usepackage{mdframed}
\usepackage[export]{adjustbox}[2011/08/13]
\usepackage{multirow}
\usepackage{hhline}
\usepackage{todonotes}
\usepackage{tcolorbox} 
\usepackage{dirtytalk}
\usepackage{yhmath}
\usepackage{enumitem}
\usepackage[pdftex, plainpages = false, pdfpagelabels, 
                 bookmarks=false,
                 bookmarksopen = true,
                 bookmarksnumbered = true,
                 breaklinks = true,
                 linktocpage,
                 pagebackref,
                 colorlinks = true,  
                 linkcolor = blue,
                 urlcolor  = blue,
                 citecolor = red,
                 anchorcolor = green,
                 hyperindex = true,
                 hyperfigures
                 ]{hyperref}

\newtheorem{lemma}{Lemma}
\newtheorem{claim}{Claim}
\newtheorem{observation}{Observation}
\newtheorem{theorem}{Theorem}
\newtheorem{corollary}{Corollary}

\newtheorem{proposition}{Proposition}

\algtext*{EndWhile}
\algtext*{EndIf}
\algtext*{EndFor}

\title{\bf Polynomial-Time Constant-Approximation for Fair Sum-of-Radii Clustering}

\author{
Sina Bagheri Nezhad\thanks{Portland State University, USA. Email: \texttt{sina5@pdx.edu}.}\ \ \ \ \ 
Sayan Bandyapadhyay\thanks{Portland State University, USA. Email: \texttt{sayanb@pdx.edu}.}\ \ \ \ \ 
Tianzhi Chen\thanks{Portland State University, USA. Email: \texttt{chench@pdx.edu}.}
} 

\date{} 

\theoremstyle{definition}

\begin{document}
\maketitle

\begin{abstract}
In a seminal work, Chierichetti et al.~\cite{chierichetti2017fair} introduced the $(t,k)$-fair clustering problem: Given a set of red points and a set of blue points in a metric space, a clustering is called fair if the number of red points in each cluster is at most $t$ times and at least $1/t$ times the number of blue points in that cluster. The goal is to compute a fair clustering with at most $k$ clusters that optimizes certain objective function. Considering this problem, they designed a polynomial-time $O(1)$- and $O(t)$-approximation for the $k$-center and the $k$-median objective, respectively. 
Recently, Carta et al.~\cite{carta2024fpt} studied this problem with the sum-of-radii objective and obtained a $(6+\epsilon)$-approximation with running time $O((k\log_{1+\epsilon}(k/\epsilon))^kn^{O(1)})$, i.e., fixed-parameter tractable in $k$. Here $n$ is the input size. 
In this work, we design the first polynomial-time $O(1)$-approximation for $(t,k)$-fair clustering with the sum-of-radii objective, improving the result of Carta et al. Our result places sum-of-radii in the same group of objectives as $k$-center, that admit polynomial-time $O(1)$-approximations. 
This result also implies a polynomial-time $O(1)$-approximation for the Euclidean version of the problem, for which an $f(k)\cdot n^{O(1)}$-time $(1+\epsilon)$-approximation was known due to Drexler et al.~\cite{drexler2023approximating}. Here  $f$ is an exponential function of $k$. We are also able to extend our result to any arbitrary $\ell\ge 2$ number of colors when $t=1$. This matches known results for the $k$-center and $k$-median objectives in this case. The significant disparity of sum-of-radii compared to $k$-center and $k$-median presents several complex challenges, all of which we successfully overcome in our work. Our main contribution is a novel cluster-merging-based analysis technique for sum-of-radii that helps us achieve the constant-approximation bounds. 
\end{abstract} 


\section{Introduction}
Given a set of points $P$ in a metric space $(\Omega,d)$ and an integer $k > 0$, the task of clustering is to find a partition $X_1,\ldots,X_k$ of $P$ into $k$ groups or clusters such that each group has similar points. The similarity of the clusters is typically modeled using an objective function which is to be minimized. In this work, we focus on the \textit{sum-of-radii} objective, which is defined as the sum of the radii of $k$ balls that contain the points of the respective $k$ clusters. The sum-of-radii objective, while also center-based, has a different flavor from objectives such as $k$-center, $k$-median, and $k$-means, as it directly sums the radii of the clusters rather than measuring distances from each point to its assigned center. In these objectives, $k$ representative points (or cluster centers) are chosen, and the corresponding clusters are formed by assigning the points of $P$ to their nearest centers. Such a partition is popularly known as the \textit{Voronoi} partition. It is not hard to see that an optimal sum-of-radii clustering is not necessarily a Voronoi partition. The study of sum-of-radii was motivated by the idea that it could potentially reduce the so-called \textit{dissection effect} that is observed in $k$-center type objectives. In $k$-center, the goal is to minimize the maximum distance between points and their cluster centers. Equivalently, we would like to compute $k$ balls of the minimum possible same radius, that contain all the points. Consequently, the points in a ground truth cluster might be assigned to different clusters to minimize the $k$-center objective. Such an effect can be reduced by using the sum of radii objective instead, as here one can use clusters of varying radii. 

Sum-of-radii clustering is known to be NP-hard even in planar metrics and metrics of constant doubling dimension \cite{DBLP:journals/algorithmica/GibsonKKPV10}. Consequently, it has received substantial attention from the approximation algorithms community. Charikar and Panigrahy~\cite{CharikarP04} designed a Primal-Dual and Lagrangian relaxation-based 3.504-approximation algorithm that runs in polynomial time (poly-time). Recently, using similar techniques, Friggstad and Jamshidian~\cite{friggstad_et_al:LIPIcs.ESA.2022.56} improved the approximation factor to 3.389. The best-known approximation factor for sum-of-radii in polynomial time is $3+\epsilon$ for any $\epsilon > 0$, due to Buchem et al.~\cite{buchem20243+}. In stark contrast to other well-studied center-based objectives such as $k$-center and $k$-median, the sum-of-radii objective admits QPTASes \cite{DBLP:journals/algorithmica/GibsonKKPV10,banerjee2024novel}, which are based on randomized metric partitioning schemes. Additionally, the problem can be solved exactly in polynomial time in the Euclidean metric of constant dimension \cite{DBLP:journals/siamcomp/GibsonKKPV12}, a unique trait that has not been observed in any other popular clustering problem. The algorithm is based on a separator theorem that guarantees the existence of a balanced separator that intersects at most a constant number of optimal balls. 
The problem also admits polynomial time exact algorithms in other restricted settings, such as when singleton clusters are not allowed \cite{DBLP:journals/algorithmica/BehsazS15} and the metric is unweighted \cite{DBLP:journals/dm/HeggernesL06}. We note that poly-time $O(1)$-approximations are known for all three center-based objectives \cite{gonzalez1985clustering, AryaGKMMP-SIAMJ04,KanungoMNPSW04}.       

In recent years, sum-of-radii clustering has also been studied with additional constraints. One such popular constraint is the capacity constraint, which puts restriction on the number of points that each cluster can contain. In a series of articles \cite{DBLP:conf/esa/0002V20,DBLP:conf/compgeom/BandyapadhyayL023a,DBLP:conf/innovations/Jaiswal0Y24,filtser2024fpt}, $O(1)$-approximation algorithms have been designed for capacitated sum-of-radii with running time fixed-parameter tractable (FPT) in $k$ (i.e., $f(k)\cdot n^{O(1)}$ for a function $f$ of $k$), culminating in an approximation factor of 3. Inamdar and Varadarajan~\cite{DBLP:conf/esa/0002V20} studied sum-of-radii with a matroid constraint where the set of centers of the balls must be an independent set of a matroid. They obtain an FPT 9-approximation for this problem. The approximation factor has recently been improved to 3 by Chen et al.~\cite{chen2024parameterized}. Obtaining a polynomial-time $O(1)$-approximation for any of these constrained versions is an interesting open question. However, poly-time $O(1)$-approximations are known for sum-of-radii with lower bounds and with outliers \cite{ahmadian_et_al:LIPIcs.ICALP.2016.69,buchem20243+}. 

Sum-of-radii has also been studied with fairness constraints, which is the main focus of our work. Clustering with fairness constraints or fair clustering stems from the idea that protected groups (defined based on a sensitive feature, e.g.,  gender) must be well-represented in each cluster. In recent years, fair clustering has received significant attention from researchers across several areas of computer science. In a seminal work, Chierichetti et al.~\cite{chierichetti2017fair} introduced the \textit{$(t,k)$-fair clustering} problem. In this problem, we are given a set $P_1$ of red points, a set $P_{2}$ of blue points, that together contain $n$ points, and an integer balance parameter $t \ge 1$. A clustering is called \emph{$(t,k)$-fair} if, for any cluster $X$, the number of red points in $X$ is at least $1/t$ times and at most $t$ times the number of blue points in $X$. 
We say that each cluster in a $(t,k)$-fair clustering is \textit{$t$-balanced}.

Chierichetti et al. studied $(t,k)$-fair clustering with $k$-center and $k$-median objectives, and obtained poly-time 4- and $O(t)$-approximation, respectively. Since then obtaining a poly-time $O(1)$-approximation for $(t,k)$-fair median or means remained an intriguing open question. The main challenge in this case is that the optimal clusterings are no longer Voronoi partitions, as they also need to be $(t,k)$-fair. Chierichetti et al. devised a scheme called \textit{fairlet decomposition} to partition the input points into units called fairlets such that each fairlet contains either one red and at most $t$ blue points or one blue and at most $t$ red points. The most important observation is that any $t$-balanced cluster can be decomposed (or partitioned) into a set of fairlets. Consequently, such fairlet decomposition can be computed using a min-cost network flow-based algorithm. Unfortunately, the cost of such a flow can be as large as $t$ times the optimal $(t,k)$-fair median cost, leading to the $O(t)$-approximation.  

Subsequently, $(t,k)$-fair median/means has been studied in a plethora of works. The only setting where it is known to obtain a poly-time $O(1)$-approximation is when $t=1$ \cite{bohm2020fair}, that is for $(1,k)$-fair median/means. Schmidt~{et al.}~\cite{schmidt2019fair} obtained an $n^{O(k/\epsilon)}$ time $(1+\epsilon)$-approximation for the Euclidean version of $(t,k)$-fair median/means. For the same version, Backurs et al.~\cite{backurs2019scalable} gave a near-linear
time $O(\Tilde{d}\cdot\log n)$-approximation, where $\Tilde{d}$ is the dimension. 



The $(t,k)$-fair median/means problem has also been studied with an arbitrary $\ell$ number of groups. The algorithm of B{\"o}hm et al.~\cite{bohm2020fair} for $t=1$ also yields a poly-time $O(1)$-approximation in this case. Note that for $t=1$, a cluster contains the same number of points from all groups. Bandyapadhyay et al.~\cite{bandyapadhyay2024polynomial} obtained a poly-time approximation for $(t,k)$-fair median with a factor that depends on $t$, $\ell$, and $k$. 
Bercea et al.~\cite{bercea2019cost} and Bera et al.~\cite{bera2019fair} independently defined a generalization of $(t,k)$-fair clustering. There we are given balance parameters $\alpha_i,\beta_i \in [0,1]$ for each group $1\le i\le \ell$. A clustering is called \emph{fair representational} if the fraction of points from group $i$ in every cluster is at least $\alpha_i$ and at most $\beta_i$ for all $1\le i\le \ell$. They show that it is possible to obtain poly-time bi-criteria type $O(1)$-approximations where we are allowed to violate the fairness constraints by an additive small constant factor. Subsequently, Dai et al.~\cite{DaiMV22} designed a DP-based poly-time $O(\log k)$-approximation for this problem. For $\ell$ groups, their running time is $n^{O(l)}$. 

Carta et al.~\cite{carta2024fpt} studied fair versions of sum-of-radii. In particular, they study a more general class of \textit{mergeable} constraints. A clustering constraint is called mergeable if the union of two clusters satisfying the constraint also satisfies the constraint. They show that the fairness constraints defined in $(t,k)$-fair clustering and fair representational clustering are mergeable. In their work, they obtained a $(6+\epsilon)$-approximation for sum-of-radii with mergeable constraints. In particular, for the above two fairness constraints, their run time is $O((k\log_{1+\epsilon}(k/\epsilon))^kn^{O(1)})$, so FPT in $k$. The algorithm iteratively guesses the next cluster based on a \textit{$k$-center completion problem} leading to the FPT run time. Their approximation factor improves to $3+\epsilon$ when $t=1$. Drexler et al.~\cite{drexler2023approximating} obtained an FPT $(1+\epsilon)$-approximation for Euclidean sum-of-radii with mergeable constraints. For this version, Banerjee et al.~\cite{banerjee2024novel} obtained a probabilistic exact algorithm with runtime $n^{O(d^2\log d)}$ and a probabilistic $(1+\epsilon)$-approximation with runtime FPT in $k+d$. Chen et al.~\cite{chen2024parameterized} considered a fair sum-of-radii problem, where the metric space $X$ is divided into demographic groups $X_1,\ldots,X_m$ and we are also given integers $k_1,\ldots,k_m$. A sum-of-radii clustering is called \textit{fair} if for the set $C$ of centers of the corresponding balls, $|C\cap X_i|=k_i$ for all $1\le i\le m$. This is indeed a restricted version of the matroid sum-of-radii as defined previously, and thus it also yields an FPT 3-approximation for this problem. For the Euclidean version of this problem, Banerjee et al.~\cite{banerjee2025improved} designed FPT $(1+\epsilon)$-approximation with $k+d$ being the parameter. A summary of the results on fair clustering under various objectives is provided in Table \ref{tab:summary}. 

As mentioned before, for fair representational models, only bi-criteria type $O(1)$-approximations are known for $k$-center/median/means, even with two groups. As we focus on our theoretical quest of designing $O(1)$-approximations fully satisfying the fairness constraints, we study $(t,k)$-fair sum-of-radii. In light of the above discussion, we state the following two questions. 

\begin{tcolorbox}
	\begin{description}
	\setlength{\itemsep}{-2pt}
	\item[Question $1$:] Does $(t,k)$-fair sum-of-radii (with two groups) admit 
	a poly-time constant-approximation algorithm?  
	\end{description}
\end{tcolorbox}

\begin{tcolorbox}
	\begin{description}
	\setlength{\itemsep}{-2pt}
	\item[Question $2$:] Does $(1,k)$-fair sum-of-radii with an arbitrary $\ell\ge 2$ number of groups admit 
	a poly-time constant-approximation algorithm?  
	\end{description}
\end{tcolorbox}

\begin{table}[t]
\centering
\caption{Summary of approximation results for fair clustering under various objectives. ``Poly'' denotes polynomial time; ``FPT'' denotes fixed-parameter tractable in $k$; $\ell$ is the number of groups.}
\label{tab:summary}
\begin{tabular}{|l|c|c|c|c|}
\hline
\textbf{Objective} & \textbf{Fairness Type} & \textbf{Approximation} & \textbf{Time} & \textbf{Reference} \\
\hline
$k$-Center & $(t,k)$ (2 groups) & 4 & Poly & \cite{chierichetti2017fair} \\
$k$-Median & $(t,k)$ (2 groups) & $O(t)$ & Poly & \cite{chierichetti2017fair} \\
$k$-Median & $(1,k)$ ($\ell$ groups) & $O(1)$ & Poly & \cite{bohm2020fair} \\
$k$-Median & $(t,k)$ ($\ell$ groups) & $f(t,\ell,k)$ & Poly & \cite{bandyapadhyay2024polynomial} \\
$k$-Median / Center & Representational (bi-criteria) & $O(1)$ & Poly & \cite{bera2019fair,bercea2019cost} \\
\hline
Sum-of-Radii & Unconstrained & $3+\varepsilon$ & Poly & \cite{buchem20243+} \\
Sum-of-Radii & Capacitated & 3 & FPT & \cite{filtser2024fpt} \\
Sum-of-Radii & Matroid constraint & 3 & FPT & \cite{chen2024parameterized} \\
Sum-of-Radii & $(t,k)$ (2 groups) & $6+\varepsilon$ & FPT & \cite{carta2024fpt} \\
Sum-of-Radii & $(1,k)$ ($\ell$ groups) & $3+\varepsilon$ & FPT & \cite{carta2024fpt} \\
Sum-of-Radii & $(t,k)$ (2 groups) & $\mathbf{144+\varepsilon}$ & \textbf{Poly} & \textbf{This work} \\
Sum-of-Radii & $(1,k)$ ($\ell$ groups) & $\mathbf{180+\varepsilon}$ & \textbf{Poly} & \textbf{This work} \\
\hline
\end{tabular}
\end{table}

\subsection{Our Contributions and Techniques} 
In our work, we prove two theorems resolving Questions 1 and 2 in the affirmative. First, we prove the following theorem. 

\begin{theorem}\label{thm:t-balanced}
    There is a polynomial-time $(144+\epsilon)$-approximation algorithm for $(t,k)$-fair sum-of-radii  (with two groups). 
\end{theorem}

Our result complements the FPT approximation result of Carta et al.~\cite{carta2024fpt} by achieving the first $O(1)$-approximation for the problem in polynomial time. The result also implies the first poly-time $O(1)$-approximation for Euclidean $(t,k)$-fair sum-of-radii. We note that our result should also be compared with that of $(t,k)$-fair $k$-median for which only $O(t)$-approximation is known in polynomial time. In particular, our result places sum-of-radii in the same group of objectives as $k$-center that admits polynomial-time $O(1)$-approximations. Moreover, our result shows that $(t,k)$-fair sum-of-radii is in contrast to most of the constrained versions of sum-of-radii, including capacitated clustering, for which only FPT $O(1)$-approximations are known. 

Next, we give an overview of our approach. Our approximation algorithm is motivated by the algorithms for $(t,k)$-fair center and $(t,k)$-fair median \cite{chierichetti2017fair}. These algorithms have two major steps. In the first step, a \textit{fairlet decomposition} of the points in $X = P_1\cup P_2$ is computed, i.e., a partition $\mathcal{Y}=\{Y_1,\ldots,Y_m\}$ such that for each fairlet $Y_i$, it either has 1 red point and at most $t$ blue points or 1 blue point and at most $t$ red points. Let $\beta: P_1\cup P_2 \rightarrow [m]$ be the function that maps each point $x$ to the index of the fairlet that contains $x$. From each $Y_i$, an arbitrary point $y_i$ is designated as its \textit{representative}. In the second step, a clustering of these $m$ representatives is computed with the respective cost function. Also, for each $Y_i$, all of its points are assigned to the cluster that contains $y_i$. The new clustering is obviously $(t,k)$-fair, as each cluster is a merger of fairlets. For the analysis of the cost of the computed clustering, they define a {fairlet decomposition} cost, which is used to bound the assignment cost of the points in the second step. For $k$-center, this cost is $\max_{x\in X} d(x,y_{\beta(x)})$, and for $k$-median, it is $\sum_{x\in X} d(x,y_{\beta(x)})$. Indeed, both of these costs when optimal are comparable to the optimal $(t,k)$-fair clustering cost. For $k$-center, it is within a constant factor, and for $k$-median it is within an $O(t)$ factor. Then, it is sufficient to compute a fairlet decomposition in the first step whose cost is within a small constant-factor of the optimal fairlet decomposition cost. 

Coming back to $(t,k)$-fair sum-of-radii, it is not clear how to define a suitable fairlet decomposition cost that can be compared to the optimal $(t,k)$-fair sum-of-radii cost. In particular, such a cost needs to be defined independent of the number of clusters $k$. However, for sum-of-radii, the objective is the sum of radii of $k$ clusters. For example, a natural candidate, the cost for $k$-median, i.e., $\sum_{x\in X} d(x,y_{\beta(x)})$, is likely to be much larger than the optimal sum-of-radii cost. In the absence of such a suitable fairlet decomposition cost, it is difficult to argue the increase in the assignment cost, when actual points of $Y_i$ are assigned instead of just the representative $y_i$. 

\medskip
\noindent
\textbf{Our approach.} Our algorithm is surprisingly simple to state.
We first compute a complete bi-partite graph $G$ with $P_1$ and $P_2$ being the two parts. The weight of each edge is set to be the distance between the two corresponding endpoints. Subsequently, a degree-constrained, spanning subgraph of this graph is computed where each vertex has a degree in range $[1,t]$, and the sum of the weights of the edges is minimized. Such an optimal subgraph can be computed in polynomial time using the algorithm of Gabow \cite{gabow1983efficient}. Moreover, one can show that such a subgraph is a collection of stars each having at most $t$ edges. Thus, our algorithm up to this point is in a similar spirit to that of $k$-median. As we argued before, the total weight of such a subgraph can be very large compared to the optimal sum-of-radii cost. Our main contribution is to prove that there is a sum-of-radii clustering of the stars (or representatives of them) computed in this way whose cost is at most a constant times the optimal $(t,k)$-fair sum-of-radii cost (\textbf{Lemma \ref{lem:approx-cluster-of-star-points-32-factor}}). Then, one can compute an approximate sum-of-radii clustering of these stars and return the corresponding clustering of the points in $P_1\cup P_2$. The obtained clustering is $(t,k)$-fair, as the clusters are disjoint union of the vertices of stars, each having at most $t$ edges. In the following, we outline the proof of the existence of a clustering of the computed stars whose cost is nicely bounded. This proof is based on a novel analysis technique that merges a set of optimal clusters to obtain \textit{superclusters}. We believe this technique would be of independent interest.  

\smallskip
\noindent
\textbf{Proof of Lemma \ref{lem:approx-cluster-of-star-points-32-factor}.} Let $H$ be the degree-constrained subgraph computed with the minimum weight possible. Also, let $\mathcal{C}^*=\{C_1^*, C_2^*, \ldots, C_k^*\}$ be a fixed optimal $(t,k)$-fair sum-of-radii clustering. 
We repetitively merge pairs of these clusters if there are edges in $H$ across them. Let $\hat{\mathcal{C}}=\{\hat{C}_1, \hat{C}_2, \ldots, \hat{C}_\kappa\}$ be the resulting clustering. By our construction, each star of $H$ is fully contained in one of these merged clusters or \textit{superclusters}. Thus, it is sufficient to show that the cost of $\hat{\mathcal{C}}$ is at most $O(1)$ times the cost of $\mathcal{C}^*$ (\textbf{Lemma \ref{lem:optimal-cluster-of-stars-8-factor}}).

\smallskip
\noindent
\textbf{Proof of {Lemma} \ref{lem:optimal-cluster-of-stars-8-factor}.} Note that it is sufficient to show that the radius of each supercluster $\hat{C}_i$ is at most $O(1)$ times the sum of the radii of the \textit{associated} optimal clusters whose merger is $\hat{C}_i$. To show this, we construct a new graph $G^*$ by contracting the associated clusters into vertices. Then, these contracted cluster vertices along with the edges of $H$ form a connected component. Note that it is sufficient to bound the (weighted) diameter of this component graph $G^*$, as the interpoint distances within an optimal cluster that was contracted are nicely bounded by the diameter of the cluster. To bound the diameter of $G^*$, we introduce a notion of \textit{minimum-switch} paths between pairs of cluster vertices. Intuitively, $G^*$ has two directed edges corresponding to each (bi-chromatic) edge in $H$ across pairs of clusters -- a $0$-edge, which represents a connection from the red point to the blue point, and a $1$-edge, which represents a connection from the blue point to the red point (see Figure \ref{fig:demo-min-switch}). Then, a \textit{minimum-switch} path between two fixed cluster vertices is a directed path in $G^*$ that has the minimum number of switches between 0- and $1$-edges (i.e., the minimum number of 0-to-1 and 1-to-0 switches) (see Figure \ref{fig:demo-min-switch}). See Section \ref{sec:switch} for formal definitions (Page 11). These paths play a central role in our analysis. We prove that it is possible to bound the (weighted) length of any such path in terms of the radii of the associated optimal clusters. Then the diameter of $G^*$ can also be bounded likewise, as any two cluster vertices are connected by a minimum-switch path. The important distinction is that the length of any arbitrary path might not be bounded in such a nice way. Consider any such minimum-switch path $\pi^*$. In the following, we describe the idea to bound its length (\textbf{Lemma \ref{lem:cost-of-pi}}). 

\begin{figure}
    \centering
    \includegraphics[width=0.7\linewidth]{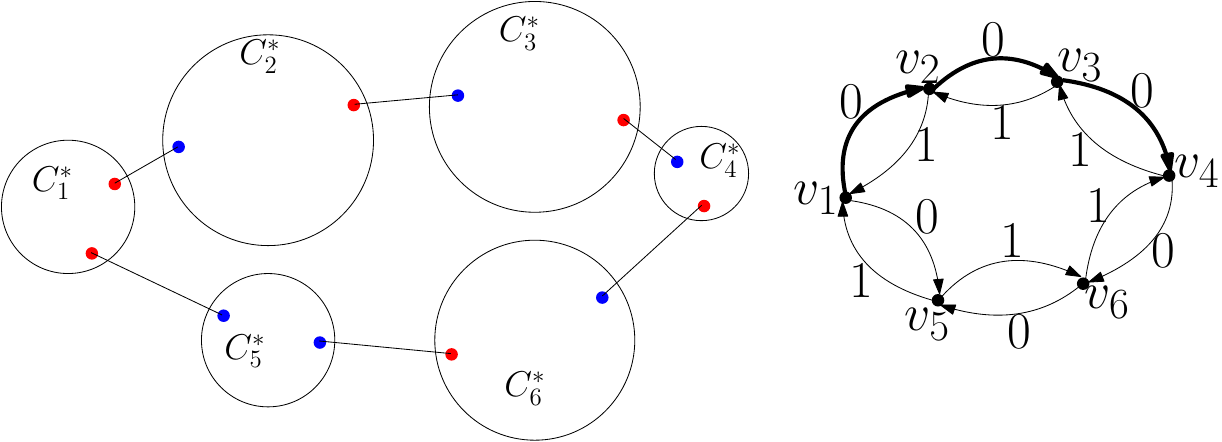}
    \caption{(Left) Optimal clusters and (bi-chromatic) edges of $H$ across them. (Right) The graph $G^*$, where the minimum-switch path from $v_1$ to $v_4$ is $v_1$-$v_2$-$v_3$-$v_4$, and has no switches.}
    \label{fig:demo-min-switch}
\end{figure}

\smallskip
\noindent
\textbf{Proof of {Lemma} \ref{lem:cost-of-pi}.} The overall idea is to show the existence of a set of edges $E_2'$ in the complete bipartite graph $G$, such that the deletion of $\pi^*$ from $H$ and the addition of $E_2'$ form a valid degree-constrained subgraph of $G$ on the set of vertices $P_1\cup P_2$. Additionally, we need the total weight of the edges of $E_2'$ to be small. Then by using the edge set difference between $H$ and the new degree-constrained subgraph, we can show that the weight of $\pi^*$ is also small, as $H$ is a minimum-weight degree-constrained subgraph of $G$. \textit{This is where we finally use the optimality of $H$.} However, it might not be possible to remove only the edges of $\pi^*$ from $H$ to show the existence of such a set $E_2'$. We show that there is a subset $E_1'$ that contains the edges of $\pi^*$ and can be removed to obtain such a valid degree-constrained subgraph. The construction of such $E_1'$ and $E_2'$ is fairly involved and is one of the main contributions of our work. 

\begin{figure}
    \centering
    \includegraphics[width=0.7\linewidth]{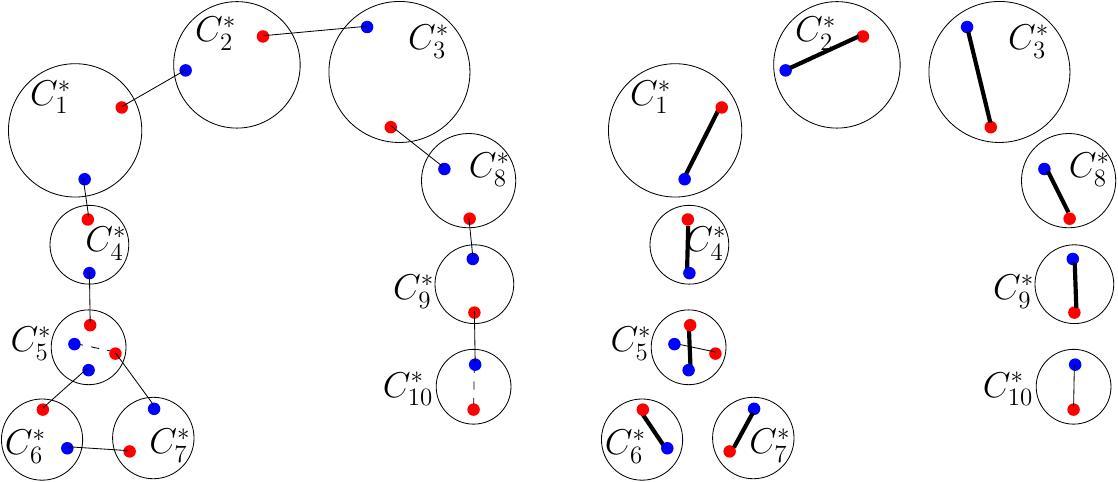}
    \caption{(Left) $\pi^*$ (1-2-3) is the path corresponding to the clusters 1, 2, 3 and has no switches. Clusters 1, 4, 5, 6, 7 form a hanging cycle and 3, 8, 9, 10 form a path with $0$-edges. The edges of $E_1'$ are shown using solid lines which are removed. (Right) The new degree-constrained subgraph. The edges of $E_2'$ are shown in bold which were added.}
    \label{fig:demo-edge-replace}
\end{figure}

\smallskip
\noindent
\textbf{Construction of $E_1'$ and $E_2'$.} 

\smallskip
\noindent
    \textbf{-- First, assume that $\pi^*$ does not have a switch and all the edges on $\pi^*$ are $0$-edges.} Initially, let $E_1'$ be the edges of $H$ corresponding to $\pi^*$. If we remove the edges in $E_1'$ from $H$, it is not guaranteed that their endpoints have at least 1 degree. So, we need to add a set of edges $E_2'$ to make it a valid degree-constrained subgraph. Now, consider the chain of clusters corresponding to $\pi^*$ connected by $0$-edges or red-to-blue edges. Then, the first cluster contains one red point corresponding to the first edge of $\pi^*$ and the last cluster contains one blue point corresponding to the last edge of $\pi^*$ ($C_1^*$ and $C_3^*$ in Figure \ref{fig:demo-edge-replace}). Moreover, any intermediate cluster in the chain corresponding to $\pi^*$ contains one blue and one red point ($C_2^*$ in Figure \ref{fig:demo-edge-replace}). We can add an edge to $E_2'$ to connect these two points and the length of this edge is nicely bounded by the diameter of a unique optimal cluster. We are left to fix the degrees of the red point in the first cluster and the blue point in the last cluster. Now, if both of their degrees in $H$ were at least 2, then the removal of the edges in $E_1'$ does not violate their degrees. So, we already have the desired $E_1'$ and the small cost set $E_2'$. Otherwise, if the degree of at least one of them is 1, after removal of the edges of $E_1'$ it becomes disconnected. So, we need to add edges to $E_2'$ to connect it with other points. In this case, we prove that the existence of a subgraph structure in $G^*$ called \textit{hanging cycle} can be exploited to fix the degree (see Figure \ref{fig:demo-edge-replace}). See Section \ref{sec:construction-E1'-E2'} for a formal definition (Page 15). Otherwise, if there is no hanging cycle in $G^*$, we prove that there is a directed path consisting of only $0$-edges or $1$-edges that can be used to fix the degree (see Figure \ref{fig:demo-edge-replace}). Intuitively, the idea is to do a graph search in $G^*$ (and respectively in $H$), until we find a blue (resp. red) point in a cluster that has a degree of at least 2 or a red (resp. blue) point that has degree at most $t-1$. The first case is again good, and in the second case, we can safely connect a disconnected blue point with the red point without violating their degrees. To prove the small cost of the edges of $E_2'$, we argue that each edge in $E_2'$ has both endpoints in the same optimal cluster (see Figure \ref{fig:demo-edge-replace}) and each optimal cluster contains endpoints of at most two edges of $E_2'$. This proves that the length of $\pi^*$ is at most 4 times the radii of the associated optimal clusters.

\smallskip
\noindent
\textbf{-- Now, it can very well be the case that $\pi^*$ has 0-1 edge switches.} Our proof up to this point is reasonably clean and simple. However, this is the most complicated case. Consider a vertex on $\pi^*$ where an edge switch happens. Then, the cluster corresponding to this vertex is intermediate, but now it has two red (or blue) points that become disconnected by the removal of the edges in $E_1'$. So, it is not possible to locally fix their degrees in contrast to the previous case. In this case, we employ two graph searches to connect these two points. Specifically, we prove the existence of two hanging cycles or paths that can be used to fix the degrees. The details are involved, but we managed to represent all the scenarios by four lemmas. The existential proofs of them (Lemmas \ref{lem:b-anchor-path-exists-u_l}, \ref{lem:b-anchor-path-exists-u_1}, \ref{lem:$1$-path-1-hanging-cycle-exist}, and \ref{lem:2-paths-no-hanging-cycle-exists}) are fairly involved and the main contributions of our work. We also make technical contributions by giving general proofs in Lemmas \ref{lem:$1$-path-1-hanging-cycle-exist} and \ref{lem:2-paths-no-hanging-cycle-exists}, which combine multiple subcases and handle them within general frameworks. Now, we have a collection of paths and hanging cycles corresponding to the first, last, and switching vertices of $\pi^*$. The edges of $H$ corresponding to these structures are added to $E_1'$ and respective edges are added to $E_2'$ to guarantee the degree constraints. However, we need to show that each vertex of $G^*$ appears in at most a constant number of these structures. Otherwise, it is not guaranteed that each optimal cluster contains endpoints of a constant number of edges of $E_2'$. \textit{Here we use the minimum-switch property of $\pi^*$.} This property ensures limited intersection between the structures. In particular, we show that two structures corresponding to two different switches are vertex-disjoint. The last property helps us prove that each optimal cluster may contain endpoints of at most three edges of $E_2'$. Consequently, the length of $\pi^*$ is at most 6 times the radii of the associated optimal clusters.

Next, we prove the following theorem concerning Question 2.

\begin{theorem}\label{thm:ell-groups}
    There is a polynomial-time $(180+\epsilon)$-approximation algorithm for $(1,k)$-fair sum-of-radii with $\ell\ge 2$ groups of points.  
\end{theorem}

Again our result directly improves the FPT approximation result of Carta et al.~\cite{carta2024fpt} and extends to more than 2 groups. The result matches the known constant-approximation bound for $k$-center/median/means in this case. The proof of the above theorem is similar to the proof of Theorem \ref{thm:t-balanced}, and so employs the same supercluster-based analysis framework. However, here we need to handle $\ell$ colors. The main challenge again boils down to bounding the diameter of a certain multi-partite graph $G_1^*$ with $\cup_{i=1}^\ell P_i$ being the set of vertices. Intuitively, by the analysis for two groups, the diameter of the graphs induced by only $P_1\cup P_i$ is nicely bounded. Specifically, we prove that for any two vertices in such a graph, there is a path without any 0-1 edge switch. However, we still need to bound the diameter of $G_1^*$. Consequently, we introduce an additional notion of \textit{minimum-color-switch} paths. We prove that the lengths of these paths can also be bounded nicely, exploiting their special properties.    

\subsection{Related Work} 
Because of its widespread popularity, sum-of-radii clustering has also been studied in a dynamic setting, where points can be added and removed \cite{DBLP:journals/algorithmica/HenzingerLM20}. Analogous to sum-of-radii, related objectives such as the sum of diameters \cite{DBLP:journals/njc/DoddiMRTW00,friggstad_et_al:LIPIcs.ESA.2022.56} and the sum of $\alpha$-th powers of radii (where $\alpha > 1$) \cite{DBLP:conf/isaac/BandyapadhyayV16} have also been investigated in literature. 

Following the seminal work of Chierichetti et al.~\cite{chierichetti2017fair}, several other notions of fairness have been considered in the context of clustering problems. The following is a sample of these works grouped by the fairness notions: individual fairness~\cite{jung2019center,negahbani2021better,vakilian2022improved,brubach2021fairness,aamand2023constant}, proportional fairness~\cite{chen2019proportionally,micha2020proportionally}, fair center representation~\cite{chen2016matroid,krishnaswamy2018constant,kleindessner2019fair,chiplunkar2020solve,hotegni2023approximation}, colorful~\cite{bandyapadhyay2019constant,jia2020fair,anegg2020technique}, and min-max fairness~\cite{abbasi2020fair,GhadiriSV21,makarychev2021approximation,chlamtavc2022approximating,ghadiri2022constant,gupta2022lp}.  

\paragraph{Organization.} First, we define some notation in Section \ref{sec:prelims}. The algorithm for $(t,k)$-fair sum-of-radii appears in Section \ref{sec:t-balanced}. Section \ref{sec:balanced} contains the algorithm for $(1,k)$-fair sum-of-radii. Lastly, we conclude with some open questions in Section \ref{sec:conclusion}. 

\section{Preliminaries}
\label{sec:prelims}

In \textit{sum-of-radii clustering}, we are given a set $P$ of $n$ points in a metric space with distance $d$ and an integer $k > 0$. We would like to find: (i) a subset $C$ of $P$ containing $k$ points and a non-negative integer $r_q$ (called radius) for each $q\in C$, and (ii) a function $\phi$ assigning each point $p\in P$ to a center $q\in C$ such that $d(p,q)\le r_q$. The subset $X_q=\phi^{-1}(q)$ for each $q\in C$ is called the \textit{cluster} corresponding to $q$ having \textit{radius} $r_q$. The goal is to find a clustering $\{X_q\mid q\in C\}$ that minimizes the sum of the radii $\sum_{q\in C} r_q$. 

\medskip
In \textit{$(t,k)$-fair sum-of-radii clustering}, we are given two disjoint groups $P_1$ (red) and $P_{2}$ (blue) having $n$ points in total in a metric space $(\Omega = P_1\cup P_2,d)$ and an integer balance parameter $t \ge 1$. A clustering is called \emph{$(t,k)$-fair} if, for each cluster $X$, the number of points from $P_1$ in $X$ is at least $1/t$ times the number of points from $P_2$ in $X$ and at most $t$ times the number of points from $P_2$ in $X$. The goal is to compute a $(t,k)$-fair clustering that minimizes the sum of the radii of the clusters.   
We say that each cluster in a $(t,k)$-fair clustering is \textit{$t$-balanced}.

\medskip
In \textit{Balanced sum-of-radii clustering}, we are given $\ell\ge 2$ disjoint groups $P_1,P_{2},\ldots,P_\ell$ having $n$ points in total in a metric space $(\Omega = \cup_{i=1}^\ell P_i,d)$ such that $|P_1|=|P_2|=\ldots=|P_\ell|$. A clustering is called \emph{balanced} if, for each cluster $X$, it holds that $|X\cap P_1|=|X\cap P_2|=\ldots=|X\cap P_\ell|$. The goal is to compute a balanced clustering that minimizes the sum of the radii of the clusters. We say that each cluster in a balanced clustering is \textit{$1$-balanced}. 

\medskip  
Consider any metric space $(\Omega_1,d_1)$ and a subset $S_1\subseteq \Omega_1$. For any cluster $Q$ and a point $p$, $d_1(p,Q)=\max_{q\in Q} d_1(p,q)$. The center of $Q$ in $S_1$ is the point, $\arg \min_{p\in S_1} d_1(p,Q)$. The radius of $Q$ w.r.t. $S_1$ and $d_1$, denoted by $r_{(S_1,d_1)}(Q)$, is the distance between $Q$ and its center in $S_1$, i.e., $r_{(S_1,d_1)}(Q)=\min_{p\in S_1} d_1(p,Q)$. We refer to the sum of the radii, w.r.t. $S_1$ and $d_1$, of the clusters in any clustering $\mathcal{C}$ as the cost of $\mathcal{C}$ w.r.t. $S_1$ and $d_1$ and denote it by cost$_{(S_1,d_1)}(\mathcal{C})$. 

We note that the term ``$(t, k)$-fairness'' refers specifically to the two-color case, where each cluster must maintain a red-to-blue ratio within $[1/t, t]$. This notion does not naturally extend to more than two colors. In contrast, in the multi-color setting with $\ell \ge 2$ groups, we adopt the term ``balanced clustering'' (or ``$(1, k)$-fair clustering with $\ell$ groups'') to describe the setting where each cluster must contain an equal number of points from each group. While we use similar notation for consistency, these two notions are structurally different and should be interpreted accordingly.

\usetikzlibrary{calc}
\section{The Algorithm for $(t,k)$-Fair Sum-of-Radii Clustering}
\label{sec:t-balanced}

In this section, we prove Theorem \ref{thm:t-balanced}. To set up the stage, we define the following problem. 

\medskip
\noindent
\begin{tcolorbox}
{\bf Min-cost Degree Constrained Subgraph (Min-cost DCS).} A \textit{Degree Constrained Subgraph} (DCS) $H=(V,E')$ of a graph $G=(V,E)$ is a 
subgraph such that the degree of each vertex $v$ in $H$ is in the range $[l(v),u(v)]$ for given integers $l(v)$ and $u(v)$. Suppose we are also given a weight function $w:E\rightarrow \mathbb{R}^+\cup \{0\}$. A min-cost DCS $H=(V,E')$ of $G$ is a DCS that minimizes the sum of the weights of the edges in $E'$ over all DCS.  
\end{tcolorbox}

The next proposition follows from the work of Gabow \cite{gabow1983efficient}. 

\begin{proposition}[\cite{gabow1983efficient}]
    \emph{Min-cost DCS} can be solved in $O(|V|^4)$ time. 
\end{proposition}

The proposition essentially follows from Theorem 5.2 \cite{gabow1983efficient}. There the stated time complexity is $O((\sum_{i\in V} u_i)$ $\min \{|E| \log |V|, |V|^2\})$, which is $O(|V|^4)$, as each upper-bound $u_i$ can be assumed to be at most the degree of the $i$-th vertex. One technicality is that they study the maximization version (with real weights) and we the minimization, but the minimization version can be solved by the standard method of negating the edge weights in min-cost DCS. Also see \cite{rajabi2023computing} that contains a similar discussion and an $O(|V|^6)$ time algorithm for min-cost DCS, which they call \textit{minimum-cost many-to-many matching with demands and capacities}. 

\begin{observation}
    A min-cost DCS with $l(v) = 1$ for all $v\in V$ does not contain a path of length three, and thus it is a disjoint union of star graphs. 
\end{observation}

\begin{proof}
    Since every vertex has lower‑bound \(1\), if there were a path \(a\!-\!b\!-\!c\!-\!d\) in our DCS then both \(b\) and \(c\) would have degree at least 2. Removing the middle edge \(b\!-\!c\) leaves all degrees \(\ge1\), contradicting minimality of the solution.
\end{proof}

Our algorithm is as follows. 
\medskip
\noindent
\begin{tcolorbox}
\noindent
{\bf The Algorithm.}\\
\textbf{1.} Construct a graph $G=(V,E)$ where $V=P_1\cup P_2$ and $E=\{\{p,q\}\mid p\in P_1, q\in P_2\}$. Define the weight function $w$ such that for each edge $e=\{p,q\}$, $w(e)=d(p,q)$. Compute a min-cost DCS $H=(V,E')$ of $G$ with $l(v)=1$ and $u(v)=t$ for all $v\in V$.  

\medskip
\noindent
\textbf{2.} 
Construct an edge-weighted graph $G'$ in the following way: For each $p\in \Omega$, add a vertex to $G'$; For each star $S$ in $H$, add a vertex corresponding to $S$ to $G'$, which we also call by $S$; For each $p,q \in \Omega$, add the edge $\{p,q\}$ to $G'$ with weight $d(p,q)$; For all $p\in \Omega$ and $S$ in $H$, add the edge $\{p,S\}$ to $G'$ with weight $\max_{q\in S} d(p,q)$. Let $d'$ be the shortest path metric in $G'$. Construct the metric space $(\Omega',d')$ where $\Omega'$ is the subset of vertices in $G'$ corresponding to the stars in $H$. 



\smallskip
\noindent
\textbf{3.} Compute a sum of radii clustering $X=\{X_1,\ldots,X_k\}$ of the points in $\Omega'$ using the Algorithm of Buchem et al.~\cite{buchem20243+} (with $\Omega'$ also being the candidate set of centers). 


\smallskip
\noindent
\textbf{4.} Compute a clustering $X'$ of the points in $P_1\cup P_2$ using $X$ in the following way. For each cluster $X_i\in X$, add the cluster  $\cup_{p\in S\mid S\in X_i} \{p\}$ to $X'$. Return $X'$.  
\end{tcolorbox}

Next, we analyze the algorithm. First, we have the following observation. 

\begin{observation}
    $X'$ is a $(t,k)$-fair clustering of $P_1\cup P_2$.  
\end{observation}

\begin{proof}
Consider any cluster $C$ in $X'$. Note that $C$ is a union of the points of a set of disjoint stars of $H$. Due to the degree bounds in $H$, the ratio of points from $P_1$ and $P_2$ in any star, and consequently in $C$, fall within the range $[1/t, t]$. This satisfies the $(t,k)$-fairness constraint for $X'$.
    
\end{proof}

Next, we analyze the approximation factor. 
Let $\mathcal{C}^*=\{C_1^*, C_2^*, \ldots, C_k^*\}$ be a fixed optimal $(t,k)$-fair  clustering. We will prove the following lemma. 

\begin{lemma}\label{lem:approx-cluster-of-star-points-32-factor}
    Consider the clustering $X$ of $\Omega'$ constructed in Step 3 of the algorithm. Then cost$_{(\Omega',d')}(X)\le \mathbf{(48+\epsilon)}\cdot \sum_{i=1}^k r_{(\Omega,d)}(C_i^*)$. 
\end{lemma}

\begin{corollary}
     Consider the clustering $X'$ of $P_1\cup P_2$ constructed in Step 4 of the algorithm. Then cost$_{(\Omega,d)}(X')\le \mathbf{(144+\epsilon)}\cdot \sum_{i=1}^k r_{(\Omega,d)}(C_i^*)$. Thus, our algorithm is a $\mathbf{(144+\epsilon)}$-approximation algorithm. 
\end{corollary}

\begin{proof}
    We claim that cost$_{(\Omega,d)}(X')\le 3\cdot $cost$_{(\Omega',d')}(X)$. Then the corollary follows by Lemma \ref{lem:approx-cluster-of-star-points-32-factor}. Consider any cluster $X_i$ of $X$ and the cluster $X_i'$ in $X'$ constructed from it. Let $S$ in $\Omega'$ be the center of $X_i$. Now, for any $S'\in X_i$, $d'(S,S')$ is the weight of a shortest path in $G'$ between $S$ and $S'$. Let $p'\in \Omega$ be the successor of $S$ on such a shortest path. So, $d'(S,S')\ge d'(S,p')=\max_{y\in S} d(y,p')\ge d(c,p')$, where $c$ is the central vertex of the star $S$, which is in $\Omega$. The equality follows by the definition of $d'$ and the fact that $d$ is a metric. Then, $d'(c,S')\le d'(c,p')+d'(p',S)+d'(S,S')\le 3\cdot d'(S,S')$. The first inequality is due to triangle inequality. Now, similarly, $d'(c,S')=\max_{q'\in S'} d(c,q')$. It follows that the ball in $(\Omega,d)$ centered at $c\in \Omega$ and having radius $3\cdot r_{(\Omega',d')}(X_i)$ contains all the points in $X_i'$. Hence, $r_{(\Omega,d)}(X_i')\le 3\cdot r_{(\Omega',d')}(X_i)$ and the claim follows. 
\end{proof}

\subsection{Proof of Lemma \ref{lem:approx-cluster-of-star-points-32-factor}}
In the following, we are going to prove Lemma \ref{lem:approx-cluster-of-star-points-32-factor}. Consider the min-cost DCS $H=(V,E')$ computed in Step 1. Also, consider the optimal clusters in $\mathcal{C}^*$. We construct a new clustering $\hat{\mathcal{C}}=\{\hat{C}_1, \hat{C}_2, \ldots, \hat{C}_\kappa\}$ by merging clusters in $\mathcal{C}^*$ in the following way, where $1\le \kappa\le k$. Initially, we set $\hat{\mathcal{C}}$ to $\mathcal{C}^*$. For each edge $\{p,q\}$ of $E'$ such that $p\in \hat{C}_i, q\in \hat{C}_j$ and $i\ne j$, replace $\hat{C}_i, \hat{C}_j$ in $\hat{\mathcal{C}}$ by their union and denote it by $\hat{C}_i$ as well. 

When the above merging procedure ends, by renaming the indexes, let $\hat{\mathcal{C}}=\{\hat{C}_1, \hat{C}_2, \ldots, \hat{C}_\kappa\}$ be the new clustering. Then, we have the following observations. 

\begin{observation}\label{obs:S-is-in-some-C-i}
    Consider any star $S$ in $H$. Then, for some $1\le i\le \kappa$, all the points of $S$ are contained in a $\hat{C}_i$.  
\end{observation}

\begin{observation}\label{obs:star-points-in-cluster}
    Consider any star $S$ in $H$ and the cluster $\hat{C}_i (\supseteq S)$ with center $c\in \Omega$. Then, for any $p \in S$, $d(p,c)\le r_{(\Omega,d)}(\hat{C}_i)$. 
\end{observation}

\begin{proof}
Because the points of $S$ are in $\hat{C}_i$, the farthest any point in $S$ can be from $c$ is not more than the farthest any point in $\hat{C}_i$ is from $c$. So, for any point $p$ in $S$, the distance between $p$ and $c$ is less than or equal to the maximum distance between any point in $\hat{C}_i$ and $c$, which we denote as $r_{(\Omega,d)}(\hat{C}_i)$.   
\end{proof}
    
\begin{observation}\label{obs:d'-p-to-c-is-at-most-radius}
    Consider the point $p$ in $\Omega'$ corresponding to a star $S$ in $H$ and the cluster $\hat{C}_i \supseteq S$ with center $c\in \Omega$. Then, $d'(p,c)\le r_{(\Omega,d)}(\hat{C}_i)$. 
\end{observation}

\begin{proof}
    As $d$ is a metric, $d'(p,c) = \max_{q \in S} d(q,c)$. By Observation \ref{obs:star-points-in-cluster}, $d(q,c)\le r_{(\Omega,d)}(\hat{C}_i)$. It follows that $d'(p,c) \le r_{(\Omega,d)}(\hat{C}_i)$. 
\end{proof}

Consider the clustering $\mathcal{C}'=\{C_1',\ldots,C'_\kappa\}$ of $\Omega'$ defined in the following way. For each star $S$ in $H$, identify the cluster $\hat{C}_i$ in $\hat{\mathcal{C}}$ that contains all the points in $S$. By Observation \ref{obs:S-is-in-some-C-i}, such an index $i$ exists. Assign the point $p$ in $\Omega'$ corresponding to $S$ to $C_i'$.  

\begin{lemma}\label{lem:costofC'-2-color}
    cost$_{(\Omega',d')}(\mathcal{C}')\le 2\cdot $ cost$_{(\Omega,d)}(\hat{\mathcal{C}})$. 
\end{lemma}

\begin{proof}
    First, we claim that $r_{(\Omega,d')}(C_i')\le r_{(\Omega,d)}(\hat{C}_i)$ for all $1\le i\le \kappa$. Let $c$ in $\Omega$ be the center of $\hat{C}_i$. Consider any star $S$ such that its corresponding point $p$ in $\Omega'$ is in $C_i'$. Then, by Observation \ref{obs:d'-p-to-c-is-at-most-radius}, $d'(p,c)\le r_{(\Omega,d)}(\hat{C}_i)$. As $c$ is in $\Omega$, it follows that, $r_{(\Omega,d')}(C_i')$ is at most $r_{(\Omega,d)}(\hat{C}_i)$. 

    Now, consider any two $S,S' \in C_i'$. By the above claim, $d'(S,S')\le 2\cdot r_{(\Omega,d)}(\hat{C}_i)$. Thus, for each such cluster $C_i'$, we can set a point $S\in C_i'$ as the center. As $S\in \Omega'$, $r_{(\Omega',d')}\le 2\cdot r_{(\Omega,d)}(\hat{C}_i)$. Summing over all clusters $C_i'$, we obtain the lemma.   
\end{proof}

We will prove the following lemma. 

\begin{lemma}\label{lem:optimal-cluster-of-stars-8-factor}
    cost$_{(\Omega,d)}(\hat{\mathcal{C}})\le 8\cdot \sum_{i=1}^k r_{(\Omega,d)}(C_i^*)$. 
\end{lemma}

Then, Lemma \ref{lem:approx-cluster-of-star-points-32-factor} follows by Lemma \ref{lem:optimal-cluster-of-stars-8-factor} and \ref{lem:costofC'-2-color} noting that the Algorithm of Buchem et al.~\cite{buchem20243+} yields a $(3+\epsilon)$-factor approximation to the optimal clustering (along with an appropriate scaling of $\epsilon$). In the rest of this section, we prove Lemma \ref{lem:optimal-cluster-of-stars-8-factor}. 

\subsection{Proof of Lemma \ref{lem:optimal-cluster-of-stars-8-factor}}\label{sec:switch}
For simplicity of notation, we drop $(\Omega,d)$ from $r_{(\Omega,d)}(.)$, as henceforth centers are always assumed to be in $\Omega$ and the metric to be $d$. 
Let us consider any fixed $\hat{C}_i$, and suppose it is constructed by merging the clusters $C_{i_1}^*,C_{i_2}^*,\ldots,C_{i_\tau}^*$. 
It is sufficient to prove that $r(\hat{C}_i)\le 8\cdot \sum_{j=1}^\tau r(C_{i_j}^*)$. For simplicity of notation, we rename $\hat{C}_i$ by $\hat{C}$, and $C_{i_1}^*,C_{i_2}^*,\ldots,C_{i_\tau}^*$ by $C_{1}^*,C_{2}^*,\ldots,C_{\tau}^*$. 

Let $H_1=(V_1,E_1)$ be the induced subgraph of $H$ such that the vertices of $V_1$ are in $\hat{C}$. We refer to a point of $P_1$ as a red point and a point of $P_2$ as a blue point. Note that the edges of $H$ are across red and blue points. In the following, we construct an edge-weighted, directed multi-graph $G^*=(V^*,E^*)$ in the following manner. $G^*$ has a vertex $v_j$ corresponding to each cluster $C_j^*$, where $1\le j\le \tau$. There is an edge $e=(v_i,v_j)$ from $v_i$ to $v_j$ for each $p\in P_1\cap C_i^*$ and $q\in P_2\cap C_j^*$ such that $\{p,q\}$ is in $E_1$. We refer to such an edge as a $0$-edge, i.e., its parity is 0. The weight $\omega_e$ of the edge $e$ is $d(p,q)$. Similarly, there is a $1$-edge (or parity 1 edge) $e=(v_i,v_j)$ from $v_i$ to $v_j$ for each $p\in P_2\cap C_i^*$ and $q\in P_1\cap C_j^*$ such that $\{p,q\}$ is in $E_1$. The weight $\omega_e$ of the edge $e$ is $d(p,q)$. For each edge $e_i\in E^*$, we denote the corresponding edge in $E_1$ by $\{r_i,b_i\}$, where $r_i$ is the red point and $b_i$ is the blue point. For simplicity of exposition, we are going to make heavy use of this correspondence.   

\begin{observation}\label{obs:$0$-$1$-edges-existence}
    Suppose there is a $0$-edge (resp. $1$-edge) $(v_i,v_j)$ in $E^*$. Then there is also a $1$-edge (resp. $0$-edge) $(v_j,v_i)$ in $E^*$. 
\end{observation}
\begin{proof}
    Suppose there exists a $0$-edge $(v_i, v_j)$ in $E^*$. Thus, there exist $p\in P_1\cap C_i^*$ and $q\in P_2\cap C_j^*$ such that $\{p,q\}$ is in $E_1$. Hence, by definition, there is an edge $(v_j,v_i)$ in $E^*$, which is a $1$-edge. Similarly, one can prove the statement if $(v_i, v_j)$ is a $1$-edge. 
\end{proof}

A directed path (or simply a path) $\pi =\{u_1,\ldots,u_l\}$ from $u_1$ to $u_l$ in $G^*$ is a sequence of distinct vertices such that $(u_i,u_{i+1})$ is in $G^*$ for all $1\le i\le l-1$. We say that $\pi$ contains the edges $(u_i,u_{i+1})$. If $\pi$ contains all $0$-edges (resp. $1$-edges), it is called a $0$-path (resp. $1$-path). Two consecutive edges $e_1=(u_i,u_{i+1}),e_2=(u_{i+1},u_{i+2})$ on $\pi$ are said to form a \textit{switch} if they have different parity. We say that the switch happens at $u_{i+1}$ and it is the corresponding switching vertex. The switch is called a $b$-switch if the parity of $e_1$ is $b$ for $b\in \{0,1\}$. 
A directed cycle is formed from $\pi$ by adding the edge $(u_l,u_1)$ (if any) with it. The \textit{reverse} path of $\pi$ is the path $\{u_{l},\ldots,u_1\}$ that contains the edges $(u_{i+1},u_i)$ for all $1\le i\le l-1$. Such edges exist according to Observation \ref{obs:$0$-$1$-edges-existence}. A $0$-path (resp. $1$-path) in a subgraph of $G^*$ starting at $v_i$ and ending at $v_j$ is called \textit{maximal} if $v_j$ does not have any outgoing $0$-edges (resp. $1$-edges) in the subgraph. 

\begin{observation}
    For any two vertices $v_i,v_j\in V^*$, there is a directed path from $v_i$ to $v_j$ in $G^*$. 
\end{observation}
\begin{proof}
    Note that $\hat{C}$ is obtained by merging $C_{1}^*,C_{2}^*,\ldots,C_{\tau}^*$ in an iterative fashion. Moreover, we merged two clusters, only if there were two endpoints of an edge of $E_1$ in those two clusters. Thus, there exists a sequence of distinct clusters, $C_{k^1}^*=C_i^*,C_{k^2}^*,\ldots,C_{k^\psi}^*=C_j^*$ such that for any $1\le \iota\le \psi-1$, $C_{k^\iota}^*$ and $C_{k^{\iota+1}}^*$ respectively contains one endpoint of an edge of $E_1$. It follows that the directed path $v_{k^1}=v_i,v_{k^2},\ldots,v_{k^\psi}=v_j$ exists in $G^*$.  
\end{proof}

Consider any two vertices $v_\alpha$ and $v_\beta$ of $G^*$. Let $\pi^*=\{v_\alpha=u_1,\ldots,u_l=v_\beta\}$ be a directed path from $v_\alpha$ to $v_\beta$ having the minimum number of switches, i.e., a \textit{minimum-switch path} from $v_\alpha$ to $v_\beta$. 
We prove the following lemma. 

\begin{lemma}\label{lem:cost-of-pi}
    $\sum_{e\in \pi^*} \omega_{e}\le 6\cdot \sum_{i=1}^\tau r(C_i^*)$. Moreover, if $\pi^*$ does not have a switch, $\sum_{e\in \pi^*} \omega_{e}\le 4\cdot \sum_{i=1}^\tau r(C_i^*)$. 
\end{lemma}

Before proving this lemma, we show how to prove Lemma \ref{lem:optimal-cluster-of-stars-8-factor}. Consider any point $p$ in $\hat{C}_j$. Let $p\in C_{g}^*$. Now, consider any point $q$ in $\hat{C}_j$ that is the farthest point from $p$.  Let $q\in C_{h}^*$. By Lemma \ref{lem:cost-of-pi} it follows that, there is a path, say $\pi'$, from $v_g$ to $v_h$ whose sum of the edge weights is at most $6\cdot \sum_{i=1}^\tau r(C_i^*)$. Then, 

\begin{align*}
    r(\hat{C}_j)& \le d(p,q)\le \sum_{e\in \pi'} \omega_{e}+\sum_{\text{vertex }v_i\in \pi'} 2\cdot r(C_i^*)\\&\le 8\cdot \sum_{i=1}^\tau r(C_i^*).   
\end{align*}

Summing over all clusters $\hat{C}_j$ in $\hat{\mathcal{C}}$, we obtain Lemma \ref{lem:optimal-cluster-of-stars-8-factor}. 

\subsection{Proof of Lemma \ref{lem:cost-of-pi}}

The overall idea is to show the existence of a subset of edges $E_2'\subset E$, such that the set of edges $(E'\setminus \pi^*)\cup E_2'$ forms a valid degree-constrained subgraph of $G$ on the set of vertices $P_1\cup P_2$. Additionally, we need that the total weight of the edges of $E_2'$ is small. Then we can show that the weight of $\pi^*$ is also small, as $H=(V,E')$ is a min-cost DCS of $G$. However, it might not be possible to remove only the edges of $\pi^*$ from $E'$ to show the existence of such a set $E_2'$. We show that there is a subset $E_1'\subseteq E'$ that contains the edges of $\pi^*$ and can be removed to obtain such a valid degree-constrained subgraph. In the following, we prove that obtaining two such sets $E_1'$ and $E_2'$ is sufficient to prove Lemma \ref{lem:cost-of-pi}. For a set of edges $S\subseteq E$, let $w(S)=\sum_{e\in S} w(e)$.   

\begin{lemma}\label{lem:cost-of-E_1'}
    Suppose there are $E_1'\subseteq E', E_2'\subset E$, such that the set of edges $(E'\setminus E_1')\cup E_2'$ forms a valid degree-constrained subgraph of $G$ and $w(E_2')\le 6\cdot \sum_{i=1}^\tau r(C_i^*)$. Then, $w(E_1')\le 6\cdot \sum_{i=1}^\tau r(C_i^*)$. 
\end{lemma}

\begin{proof}
    Note that $H$ is a min-cost DCS of $G$. Consider the graph $H'$ induced by the set of edges $(E'\setminus E_1')\cup E_2'$. By our assumption, $H'$ is a valid DCS of $G$. It follows that, 
    \begin{align*}
        & w(E')\le w((E'\setminus E_1')\cup E_2') \\
        \text{or, } &  w(E_1')\le w(E_2')\le 6\cdot \sum_{i=1}^\tau r(C_i^*).
    \end{align*}
    The last inequality follows from our assumption. 
\end{proof}

Assuming that the conditions of the above lemma are true, we finish the proof of Lemma \ref{lem:cost-of-pi}. 

\begin{align*}
    \sum_{e\in \pi^*} \omega_{e} &= \sum_{(v_i,v_j)\in \pi^*} \omega_{e}\\&= \sum_{\{p,q\}\text{ corresponding to } (v_i,v_j)\in \pi^*\mid p\in C_i^*,q \in C_j^*} w(\{p,q\})\le w(E_1') \le 6\cdot \sum_{i=1}^\tau r(C_i^*).
\end{align*}

If $\pi^*$ does not have a switch, then we will show that $w(E_2')\le 4\cdot \sum_{i=1}^\tau r(C_i^*)$. Hence, the moreover part in Lemma \ref{lem:cost-of-pi} also follows. It is left to show the existence of such $E_1'$ and $E_2'$, which we consider next. 

\subsection{Construction of $E_1'$ and $E_2'$}\label{sec:construction-E1'-E2'}

Two subraphs $G_1$ and $G_2$ of $G^*$ are called \textit{$0$-$1$-edge-disjoint} if for any edge $e_{\eta^1}$ of $G_1$ and $e_{\eta^2}$ of $G_2$ the corresponding edges in $E'$ are distinct. Thus, if $G_1$ contains a $0$-edge $(v_i,v_j)$, and $G_1$ and $G_2$ are $0$-$1$-edge-disjoint, then $G_2$ cannot contain the $0$-edge $(v_i,v_j)$ and the $1$-edge $(v_j,v_i)$. Similarly, if $G_1$ contains a $1$-edge $(v_i,v_j)$, and $G_1$ and $G_2$ are $0$-$1$-edge-disjoint, then $G_2$ cannot contain the $1$-edge $(v_i,v_j)$ and the $0$-edge $(v_j,v_i)$. Let $j^1 < j^2 <\ldots < j^\lambda$ be the indexes of the vertices on $\pi^*=\{u_1,\ldots,u_l\}$ where the switches occur. Note that $j^1 > 1, j^\lambda < l$. Denote the switch that occurs at $u_{j^h}$ by $b^h$ for all $1\le h\le \lambda$ (i.e., $b^h$ is the parity of $(u_{j^h-1},u_{j^h})$). Let $b^0$ be the parity of $(u_1,u_2)$ and $b^{\lambda+1}$ be the parity of $(u_{l-1},u_l)$. 

\begin{observation}\label{obs:range-of-b^h-path}
    For any $1\le h\le \lambda$, the only vertices of $\pi^*$ a $b^h$-path starting at $u_{j^h}$ can contain are $u_i,u_{i+1},\ldots,u_{i'}$ where $i=2$ if $h=1$ and $i=j^{h-1}+1$ otherwise and $i'=l-1$ if $h=\lambda$ and $i'=j^{h+1}-1$ otherwise. Also, the only vertices of $\pi^*$ a $b^{\lambda+1}$-path starting at $u_l$ can contain are $u_{z},\ldots,u_{l}$ where $z=j^\lambda+1$ if $\pi^*$ has at least one switch, and $z=1$ otherwise. Moreover, the only vertices of $\pi^*$ a $(1-b^{0})$-path starting at $u_1$ can contain are $u_{1},\ldots,u_{z'}$ where $z'=j^1-1$ if $\pi^*$ has at least one switch, and $z=l$ otherwise. 
\end{observation}

\begin{proof}
    Consider the vertex $u_{j^h}$. If there is a $b^h$-path $\pi_1$ starting at $u_{j^h}$ that contains one vertex $v_j$ among $u_1,\ldots,u_{i-1}$, then one can construct a new path $\pi_1^*$ by joining the part of $\pi^*$ from $u_1$ to $v_j$, the reverse of $\pi_1$ (from $v_j$ to $u_{j^h}$), and the part of $\pi^*$ from $u_{j^h}$ to $u_l$. But, the number of switches in $\pi_1^*$ is strictly less than that of $\pi^*$, which is a contradiction (see Figure \ref{fig:obs9}). Now, if there is a $b^h$-path $\pi_1'$ starting at $u_{j^h}$ that contains one vertex $v_j$ among $u_{i'+1},\ldots,u_{l}$, then one can construct a new path by joining the part of $\pi^*$ from $u_1$ to $u_{j^h}$, $\pi_1'$, and the part of $\pi^*$ from $v_j$ to $u_l$. But, the number of switches in this new path is strictly less than that of $\pi^*$, which is a contradiction.  

    Next, consider the vertex $u_l$. If $\pi^*$ does not have a switch, then a $b^{\lambda+1}$-path starting at $u_l$ can contain any vertex on $\pi^*$. So, suppose $\pi^*$ has a switch. If there is a $b^{\lambda+1}$-path $\pi_2$ starting at $u_l$ that contains one vertex $v_j$ among $u_1,\ldots,u_{j^\lambda}$, then one can construct a new path $\pi_2^*$ by joining the part of $\pi^*$ from $u_1$ to $v_j$ and the reverse of $\pi_2$. But, the number of switches in $\pi_2^*$ is strictly less than that of $\pi^*$, which is a contradiction.  

    Lastly, consider the vertex $u_1$. If $\pi^*$ does not have a switch, then a $(1-b^{0})$-path starting at $u_1$ can contain any vertex on $\pi^*$. So, suppose $\pi^*$ has a switch. If there is a $(1-b^{0})$-path $\pi_3$ starting at $u_1$ that contains one vertex $v_j$ among $u_{j^1},\ldots,u_{l}$, then one can construct a new path $\pi_3^*$ by joining $\pi_3$ and the part of $\pi^*$ from $v_j$ to $u_l$. But, the number of switches in $\pi_3^*$ is strictly less than that of $\pi^*$, which is a contradiction.    
\end{proof}
\begin{figure}
    \centering
    \includegraphics[width=0.6\linewidth]{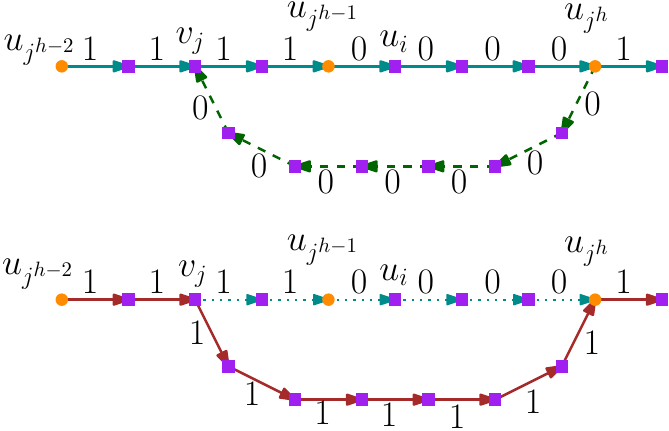}
    \caption{(Top) $\pi^*$ is shown using the bold edges and $\pi_1$ is shown using the dashed edges. (Bottom) $\pi_1^*$ is shown using the bold edges.}
    \label{fig:obs9}
\end{figure}

\textbf{First, we consider the simple case when the parity of $(u_{l-1},u_l)$ is 0 (resp. 1) and there is a $0$-path (resp. $1$-path) from $u_l$ to $u_1$ in $G^*$ that is $0$-$1$-edge-disjoint from $\pi^*$.} Let us denote the latter path by $\pi(l)$. Note that $\pi^*$ is a path having the minimum number of switches and the existence of $\pi(l)$ ensures that $\pi^*$ does not have a switch. Let $U_0\subseteq E^*$ be the subset of edges that lie on the paths in $\{\pi^*\}\cup \{\pi(l)\}$. Next, we define a subset $E_1'\subseteq E_1$ that has a one-to-one mapping with $U_0$. In particular, consider any edge $(v_i,v_j)$ in $U_0$. Note that if it is a $0$-edge, it was added due to an edge  $\{p,q\}$ in $E_1$ such that $p\in P_1\cap C_i^*$ and $q\in P_2\cap C_j^*$. We add the edge $\{p,q\}$ to $E_1'$. Otherwise, if $(v_i,v_j)$ is a $1$-edge, it was added due to an edge $\{p,q\}$ in $E_1$ such that $p\in P_2\cap C_i^*$ and $q\in P_1\cap C_j^*$. In this case, we add the edge $\{p,q\}$ to $E_1'$. 

Next, we show the construction of $E_2'$. Wlog, let us assume that $\pi^*$ is a $0$-path. The other case is symmetric. Note that then $\pi(l)$ is also a $0$-path as per our assumption. First, we describe the process of adding the replacement edges for the path $\pi^*$. Consider any intermediate vertex (if any) $v_{j'}$ on this path. Then, there are exactly two points in $C_{j'}^*$ corresponding to the edges on $\pi^*$, which are of opposite colors. We add an edge between these two points in $E_2'$ (see Figure \ref{fig:E_2'fornoswitch}). Removal of the edges of $E_1'$ corresponding to $\pi^*$ and the addition of this edge do not change the degree of the two points in $C_{j'}^*$. Similarly, we add edges to $E_2'$ corresponding to the intermediate vertices of $\pi(l)$. Next, consider the vertex $u_l=v_i$. There is an incoming $0$-edge on $\pi^*$ and an outgoing $0$-edge on $\pi(l)$ that are incident on $v_i$. Thus, there are exactly two points in $C_{i}^*$ corresponding to these two edges, which are of opposite colors. We add an edge between these two points in $E_2'$. Removal of the edges of $E_1'$ corresponding to those two edges, and the addition of this edge does not change the degree of the two points in $C_i^*$. Similarly,  consider the vertex $u_1=v_i'$. There is an outgoing $0$-edge on $\pi^*$ and an incoming $0$-edge on $\pi(l)$ that are incident on $v_i'$. Thus, there are exactly two points in $C_{i'}^*$ corresponding to these two edges, which are of opposite colors. We add an edge between these two points in $E_2'$. Again, the removal of the edges of $E_1'$ corresponding to those two edges, and the addition of this edge does not change the degree of the two points in $C_{i'}^*$. See Figure \ref{fig:E_2'fornoswitch} for an illustration.  

\begin{figure}[H]
    \centering
    \includegraphics[width=0.35\linewidth]{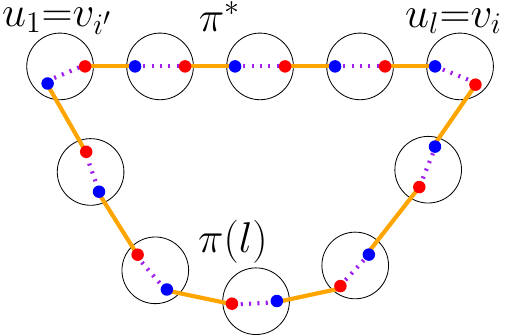}
    \caption{Figure illustrating the construction of $E_2'$ for $\{\pi^*\}\cup \{\pi(l)\}$. The bold (orange) edges are in $E_1'$ and the dashed (purple) edges are in $E_2'$.}
    \label{fig:E_2'fornoswitch}
\end{figure}

By our construction, the set of edges $(E'\setminus E_1')\cup E_2'$ form a valid degree-constrained subgraph of $G$ on the set of vertices $P_1\cup P_2$. The way we add the edges to $E_2'$, both endpoints of each edge lie in a cluster $C_j^*$ such that the vertex $v_j$ corresponding to the cluster lies on a path in $\{\pi^*\}\cup \{\pi(l)\}$. Now, $v_j$ can lie either on one such path or on two paths. Thus, we add at most two edges to $E_2'$ corresponding to $v_j$. The sum of the weights of these two edges is at most 2 times the diameter of $C_j^*$. Hence, by Lemma \ref{lem:cost-of-E_1'}, we obtain $$ \sum_{e\in \pi^*} \omega_{e}\le 4\cdot \sum_{i=1}^\tau r(C_i^*).$$ 

\textbf{Next, we consider the remaining case when the parity of $(u_{l-1},u_l)$ is 0 (resp. 1) and there is no $0$-path (resp. $1$-path) from $u_l$ to $u_1$ in $G^*$ that is $0$-$1$-edge-disjoint from $\pi^*$.} Thus, there is no $b^{\lambda+1}$-path from $u_l$ to $u_1$ in $G^*$ that is $0$-$1$-edge-disjoint from $\pi^*$. 

Consider a path $\pi$ and let $v_i$ denote its start vertex. Also, consider a cycle $O$ such that $\pi$ and $O$ have exactly one vertex $v_j$ in common. Note that $\pi$ might not have an edge, in which case $v_j=v_i$. Let $D$ be the graph formed by the union of $\pi$ and $O$, i.e., by gluing them together at $v_j$. We refer to such a graph $D$ as a \textit{hanging cycle} for $v_i$ with $v_j$ being the \textit{join} vertex. $D$ is called a $b$-hanging cycle if all the edges of $\pi$ and $O$ are $b$-edges. Let $p$ be the point in the cluster $C_j^*$ corresponding to the edge of the cycle $O$ incoming to $v_j$. Additionally, we say that $D$ is \textit{special} if the degree of $p$ in $H_1$ is at least 2, and $p$ is called the \textit{special} point of $D$ (see Figure \ref{fig:hangingcycle}).  

\begin{figure}[H]
    \centering
    \includegraphics[width=0.4\linewidth]{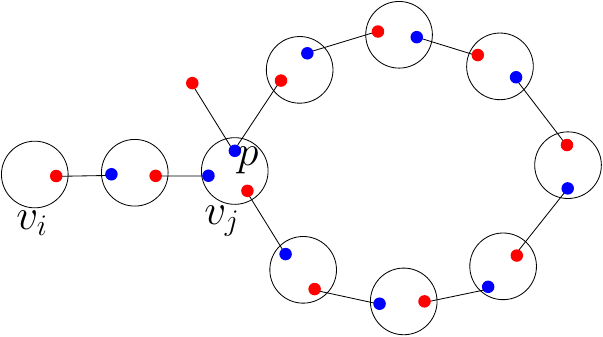}
    \caption{A special hanging cycle with the special point $p$.}
    \label{fig:hangingcycle}
\end{figure}

    

In the current case, we need the following four lemmas whose proofs will be given later. \\

\begin{restatable}{lemma}{anchorPathUl}\label{lem:b-anchor-path-exists-u_l}
    Suppose the parity of $(u_{l-1},u_l)$ is 0 (resp. 1), and there is no $0$-path (resp. $1$-path) from $u_l$ to $u_1$ in $G^*$ that is $0$-$1$-edge-disjoint from $\pi^*$. Moreover, suppose there is no special 0-hanging cycle (resp. 1-hanging cycle) in $G^*$ for $u_l$ that is $0$-$1$-edge-disjoint from $\pi^*$. Then, there exists a $0$-path (resp. $1$-path) $\pi_1$ in $G^*$ from $u_l$ to a vertex $v_j$, such that $\pi_1$ is $0$-$1$-edge-disjoint from $\pi^*$ and one of the following is true: (i) the degree of $b_\eta$ (resp. $r_\eta$) in $H_1$ is at least 2, where $e_\eta$ is the last edge on $\pi_1$ if it has an edge or $(u_{l-1},u_l)$ otherwise; or (ii) $C_j^*$ has a red (resp. blue) point whose degree in $H_1$ is at most $t-1$. 
\end{restatable}

\begin{restatable}{lemma}{anchorPathUone}\label{lem:b-anchor-path-exists-u_1}
    Suppose the parity of $(u_{1},u_2)$ is 1 (resp. 0), and there is no $1$-path (resp. $0$-path) from $u_l$ to $u_1$ in $G^*$ that is $0$-$1$-edge-disjoint from $\pi^*$. Moreover, suppose there is no special 0-hanging cycle (resp. 1-hanging cycle) in $G^*$ for $u_1$ that is $0$-$1$-edge-disjoint from $\pi^*$. Then, there exists a $0$-path (resp. $1$-path) $\pi_1$ from $u_1$ to a vertex $v_j$, such that $\pi_1$ is $0$-$1$-edge-disjoint from $\pi^*$ and one of the following is true: (i) the degree of $b_\eta$ (resp. $r_\eta$) in $H_1$ is at least 2, where $e_\eta$ is the last edge on $\pi_1$ if it has an edge or $(u_{2},u_1)$ otherwise, and (ii) $C_j^*$ has a red (resp. blue) point whose degree in $H_1$ is at most $t-1$. 
\end{restatable}

    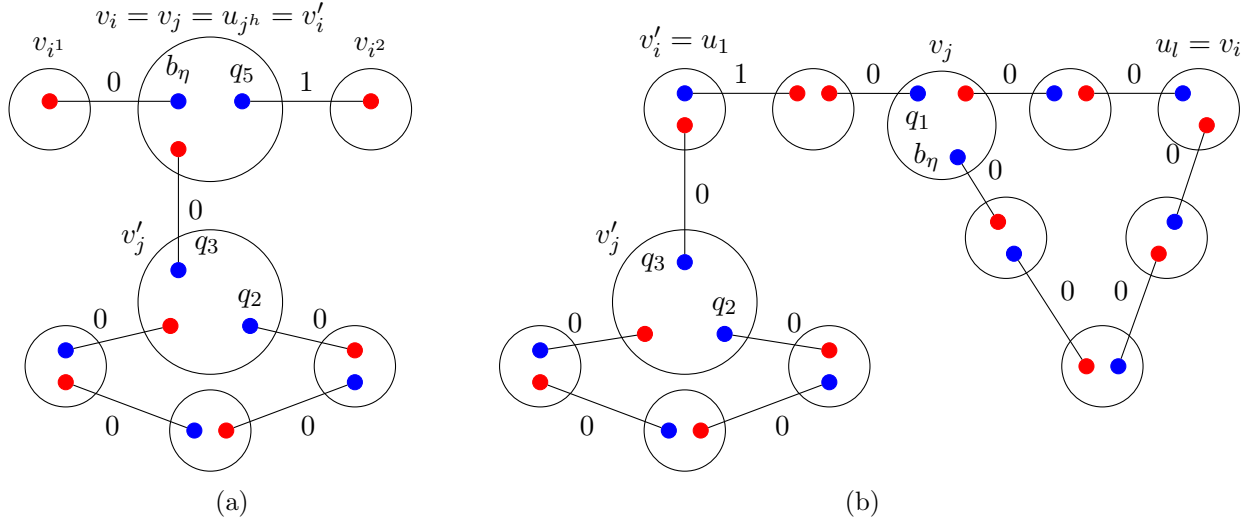
\begin{figure}[h]
        \begin{subfigure}{.375\textwidth}
        \begin{tikzpicture}[main/.style = {draw, circle, fill, inner sep=2pt}, node distance=50, scale=0.85]
            \coordinate (c1) at (0, 0);
            \draw (c1) circle[radius=18pt, label=$a$];
            \node[above, yshift=16pt] at (c1) {$v_{i^1}$};
            
            \node[main, red, xshift=0, yshift=3] at (c1) (c1-1) {};

            \coordinate (c2) at (2.5, 0);
            \draw (c2) circle[radius=32pt];
            \node[above, yshift=24pt] at (c2) {$v_i=v_j=u_{j^h}=v_i'$};
            
            \node[main, blue, xshift=-12, yshift=3, label=above:$b_\eta$] at (c2) (c2-1) {};
            \node[main, blue, xshift=12, yshift=3, label=above:$q_5$] at (c2) (c2-2) {};
            \node[main, red, xshift=-12, yshift=-15] at (c2) (c2-3) {};

            \coordinate (c3) at (5, 0);
            \draw (c3) circle[radius=18pt];
            \node[above, yshift=16pt] at (c3) {$v_{i^2}$};
            \node[main, red, xshift=0, yshift=3] at (c3) (c3-1) {};

            \coordinate (c4) at (2.5, -3);
            \draw (c4) circle[radius=32pt];
            \node[left, xshift=-20pt, yshift=24pt] at (c4) {$v_j'$};
            \node[main, blue, xshift=-12, yshift=12, label=above right:$q_3$] at (c4) (c4-1) {};
            \node[main, red, xshift=-15, yshift=-9] at (c4) (c4-2) {};
            \node[main, blue, xshift=15, yshift=-9, label=above:$q_2$] at (c4) (c4-3) {};

            \coordinate (c5) at (0.25, -4);
            \draw (c5) circle[radius=18pt];
            \node[left, xshift=-18pt, yshift=24pt] at (c5) {};
            \node[main, blue, xshift=0, yshift=6] at (c5) (c5-1) {};
            \node[main, red, xshift=0, yshift=-6] at (c5) (c5-2) {};
            
            \coordinate (c6) at (2.5, -5);
            \draw (c6) circle[radius=18pt];
            \node[left, xshift=-18pt, yshift=24pt] at (c6) {};
            \node[main, blue, xshift=-6] at (c6) (c6-1) {};
            \node[main, red, xshift=6] at (c6) (c6-2) {};
            
            \coordinate (c7) at (4.75, -4);
            \draw (c7) circle[radius=18pt];
            \node[left, xshift=-18pt, yshift=24pt] at (c7) {};
            \node[main, blue, xshift=0, yshift=-6] at (c7) (c7-1) {};
            \node[main, red, xshift=0, yshift=6] at (c7) (c7-2) {};
            
            \path [-] (c1-1) edge node[midway, above] {0} (c2-1);
            \path [-] (c2-2) edge node[midway, above] {1} (c3-1);
            
            \path [-] (c2-3) edge node[midway, right] {0} (c4-1);
            
            \path [-] (c4-2) edge node[midway, above left] {0} (c5-1);
            \path [-] (c5-2) edge node[midway, below left] {0} (c6-1);
            \path [-] (c6-2) edge node[midway, below right] {0} (c7-1);
            \path [-] (c7-2) edge node[midway, above right] {0} (c4-3);
        \end{tikzpicture}
        \caption{}
        \end{subfigure}
        \begin{subfigure}{.625\textwidth}
        \begin{tikzpicture}[main/.style = {draw, circle, fill, inner sep=2pt}, node distance=50, scale=0.85]
            \coordinate (c1) at (0, 0);
            \draw (c1) circle[radius=18pt, label=$a$];
            \node[above, yshift=16pt] at (c1) {$v_i' = u_1$};
            
            \node[main, blue, xshift=0, yshift=6] at (c1) (c1-1) {};
            \node[main, red, xshift=0, yshift=-6] at (c1) (c1-2) {};

            \coordinate (c2) at (2, 0);
            \draw (c2) circle[radius=18pt];
            
            \node[main, red, xshift=-6, yshift=6] at (c2) (c2-1) {};
            \node[main, red, xshift=6, yshift=6] at (c2) (c2-2) {};

            \coordinate (c3) at (4, -0.25);
            \draw (c3) circle[radius=24pt];
            \node[above, yshift=20pt] at (c3) {$v_j$};
            
            \node[main, blue, xshift=-9, yshift=12, label=below:$q_1$] at (c3) (c3-1) {};
            \node[main, red, xshift=9, yshift=12] at (c3) (c3-2) {};
            \node[main, blue, xshift=6, yshift=-12, label=left:$b_\eta$] at (c3) (c3-3) {};

            \coordinate (c8) at (6, 0);
            \draw (c8) circle[radius=18pt];
            \node[above, yshift=24pt] at (c8) {};
            \node[main, blue, xshift=-6, yshift=6] at (c8) (c8-1) {};
            \node[main, red, xshift=6, yshift=6] at (c8) (c8-2) {};

            \coordinate (c9) at (8, 0);
            \draw (c9) circle[radius=18pt];
            \node[above, yshift=16pt] at (c9) {$u_l = v_i$};
            \node[main, blue, xshift=-6, yshift=6] at (c9) (c9-1) {};
            \node[main, red, xshift=3, yshift=-6] at (c9) (c9-2) {};

            \coordinate (c10) at (7.5, -2);
            \draw (c10) circle[radius=18pt];
            \node[above, yshift=36pt] at (c10) {};
            \node[main, blue, xshift=3, yshift=6] at (c10) (c10-1) {};
            \node[main, red, xshift=-3, yshift=-6] at (c10) (c10-2) {};

            \coordinate (c11) at (6.5, -4);
            \draw (c11) circle[radius=18pt];
            \node[above, yshift=36pt] at (c11) {};
            \node[main, blue, xshift=6, yshift=0] at (c11) (c11-1) {};
            \node[main, red, xshift=-6, yshift=0] at (c11) (c11-2) {};

            \coordinate (c12) at (5, -2);
            \draw (c12) circle[radius=18pt];
            \node[above, yshift=36pt] at (c12) {};
            \node[main, blue, xshift=3, yshift=-6] at (c12) (c12-1) {};
            \node[main, red, xshift=-3, yshift=6] at (c12) (c12-2) {};

            \coordinate (c4) at (0, -3);
            \draw (c4) circle[radius=32pt];
            \node[left, xshift=-20pt, yshift=24pt] at (c4) {$v_j'$};
            \node[main, blue, xshift=0, yshift=15, label=left:$q_3$] at (c4) (c4-1) {};
            \node[main, red, xshift=-15, yshift=-12] at (c4) (c4-2) {};
            \node[main, blue, xshift=15, yshift=-12, label=above:$q_2$] at (c4) (c4-3) {};

            \coordinate (c5) at (-2.25, -4);
            \draw (c5) circle[radius=18pt];
            \node[left, xshift=-18pt, yshift=24pt] at (c5) {};
            \node[main, blue, xshift=0, yshift=6] at (c5) (c5-1) {};
            \node[main, red, xshift=0, yshift=-6] at (c5) (c5-2) {};
            
            \coordinate (c6) at (0, -5);
            \draw (c6) circle[radius=18pt];
            \node[left, xshift=-18pt, yshift=18pt] at (c6) {};
            \node[main, blue, xshift=-6] at (c6) (c6-1) {};
            \node[main, red, xshift=6] at (c6) (c6-2) {};
            
            \coordinate (c7) at (2.25, -4);
            \draw (c7) circle[radius=18pt];
            \node[left, xshift=-18pt, yshift=24pt] at (c7) {};
            \node[main, blue, xshift=0, yshift=-6] at (c7) (c7-1) {};
            \node[main, red, xshift=0, yshift=6] at (c7) (c7-2) {};
            
            \path [-] (c1-1) edge node[midway, above] {1} (c2-1);
            \path [-] (c2-2) edge node[midway, above] {0} (c3-1);
            
            \path [-] (c1-2) edge node[midway, right] {0} (c4-1);
            
            \path [-] (c4-2) edge node[midway, above left] {0} (c5-1);
            \path [-] (c5-2) edge node[midway, below left] {0} (c6-1);
            \path [-] (c6-2) edge node[midway, below right] {0} (c7-1);
            \path [-] (c7-2) edge node[midway, above right] {0} (c4-3);
            
            \path [-] (c3-2) edge node[midway, above] {0} (c8-1);
            \path [-] (c8-2) edge node[midway, above] {0} (c9-1);
            \path [-] (c9-2) edge node[midway, above left] {0} (c10-1);
            \path [-] (c10-2) edge node[midway, above left] {0} (c11-1);
            \path [-] (c11-2) edge node[midway, above right] {0} (c12-1);
            \path [-] (c12-2) edge node[midway, above right] {0} (c3-3);
        \end{tikzpicture}
        \caption{}
        \end{subfigure}
        \caption{Figure illustrating the two cases (a) $v_i=v_{i'}=u_{j^h}$ and (b) $v_i=u_l,v_{i'}=u_1$ of Lemma \ref{lem:$1$-path-1-hanging-cycle-exist}.}
        \label{fig:figure-1}
    \end{figure}

\begin{restatable}{lemma}{onePathOneHangingCycle}\label{lem:$1$-path-1-hanging-cycle-exist}
   Suppose $\pi^*$ has at least one switch. Consider any $v_i,v_{i'}\in V^*$ such that either $v_i=v_{i'}=u_{j^h}$ for some $1\le h\le \lambda$ or $v_i=u_l,v_{i'}=u_1$, the edge $(v_{i^1},v_i)$ on $\pi^*$, which is a $0$-edge (resp. $1$-edge), and the edge $(v_{i'},v_{i^2})$ on $\pi^*$, which is a $1$-edge (resp. $0$-edge). Also, suppose there is a special 0-hanging cycle (resp. 1-hanging cycle) $O$ for $v_{i'}$ in $G^*$ that is $0$-$1$-edge-disjoint from $\pi^*$. 
   Moreover, there is no special 0-hanging cycle (resp. 1-hanging cycle) in $G^*$ for $v_i$ that is $0$-$1$-edge-disjoint from $\pi^*$ and $O$, and either its special point is distinct from that of $O$ or the special point is the same as that of $O$ and has degree in $H_1$ at least 3.  

Then, there exists a $0$-path (resp. $1$-path) $\pi_1$ from $v_i$ to a vertex $v_j$ in $G^*$, such that $O$, $\pi^*$, and $\pi_1$ are $0$-$1$-edge-disjoint and one of the properties (i)-(xvii) below holds.  
    
    Denote by $e_\eta$ the last edge on $\pi_1$ if it has an edge or $(v_{i^1},v_i)$ otherwise. The edge in $E_1$ corresponding to $e_{\eta}$ is $\{r_{\eta},b_{\eta}\}$, where $r_{\eta}$ is the red point and $b_{\eta}$ is the blue point.  
    In case $v_{j}\ne u_1$ is on $\pi^*$, let $q_1$ be the point in $C_{j}^*$ corresponding to the incoming edge of $v_j$ on $\pi^*$. Also, if $v_i=u_l$, then $z=j^\lambda+1, z'=l$, and if $v_i=u_{j^h}$, then $z=2$ if $h=1$ and $z=j^{h-1}+1$ otherwise, and $z'=l-1$ if $h=\lambda$ and $z'=j^{h+1}-1$ otherwise.  Let $v_{j'}$ be the join vertex of $O$. Also, let $q_2$ be the point in $C_{j'}^*$ corresponding to the incoming edge to $v_{j'}$ that lies on the cycle of $O$, i.e., $q_2$ is the special point of $O$. Let $q_3$ be the point in $C_{j'}^*$ corresponding to the other incoming edge of $v_{j'}$ on $O$ if $v_{j'}\ne v_{i'}$. Moreover, let $q_5$ be the point in $C_{i'}^*$ corresponding to the edge $(v_{i'},v_{i^2})$. 
    \begin{enumerate}[label=(\roman*)]
        \item  $b_{\eta}= q_2$ (resp. $r_{\eta}= q_2$) and the degree of $q_2$ in $H_1$ is at least 3
\item  $b_{\eta}= q_2$ (resp. $r_{\eta}= q_2$), the degree of $q_2$ in $H_1$ is 2, $v_{j}\ne v_i$, and $q_2 = q_3$
\item  $b_{\eta}= q_2$  (resp. $r_{\eta}= q_2$), the degree of $q_2$ in $H_1$ is 2, $v_{j}\ne v_i$, $q_2\ne q_3$, and the degree of $q_3$  in $H_1$ is at least 2
 \item $b_{\eta}= q_2$  (resp. $r_{\eta}= q_2$), the degree of $q_2$ in $H_1$ is 2, $v_{j}$ is one of $u_z,u_{z+1},\ldots,u_{z'}$, and $q_2 = q_1$
\item $b_{\eta}= q_2$  (resp. $r_{\eta}= q_2$), the degree of $q_2$ in $H_1$ is 2, $v_{j}$ is one of $u_z,u_{z+1},\ldots,u_{z'}$, $q_2\ne q_1$, and the degree of $q_1$ in $H_1$ is at least 2
\item $b_{\eta}= q_2$  (resp. $r_{\eta}= q_2$), the degree of $q_2$ in $H_1$ is 2, $v_{j}$ is one of $u_z,u_{z+1},\ldots,u_{z'}$, $q_2\ne q_1$, $q_2\ne q_3$,  the degree of $q_1$ and $q_3$ in $H_1$  are 1, and there is a red (resp. blue) point in $C_{j}^*$ whose degree in $H_1$ is at most $t-1$
\item $b_{\eta}= q_2$  (resp. $r_{\eta}= q_2$), the degree of $q_2$ in $H_1$ is 2, $v_{j}$ is not one of $u_z,u_{z+1},\ldots,u_{z'+1}$, $q_2\ne q_3$,  the degree of $q_3$ in $H_1$ is 1, and there is a red (resp. blue) point in $C_{j}^*$ whose degree in $H_1$ is at most $t-1$
\item $b_{\eta}= q_2$  (resp. $r_{\eta}= q_2$), the degree of $q_2$ in $H_1$ is 2, $v_{j}=v_i\ne u_l$, and the degree of $q_5$ in $H_1$ is at least 2
 \item $b_{\eta}= q_2$  (resp. $r_{\eta}= q_2$), the degree of $q_2$ in $H_1$ is 2, $v_{j}=v_i\ne u_l$,  the degree of $q_5$ in $H_1$ is 1, and there is a red (resp. blue) point in $C_{j}^*$ whose degree in $H_1$ is at most $t-1$
\item  $b_{\eta}\ne q_2$ (resp. $r_{\eta}\ne q_2$) and the degree of $b_{\eta}$ (resp. $r_{\eta}$) in $H_1$ is at least 2
\item The degree of $b_{\eta}$ (resp. $r_{\eta}$) in $H_1$ is 1, $v_{j}$ is one of $u_z,u_{z+1},\ldots,u_{z'}$, and the degree of $q_1$ in $H_1$  is at least 2
 \item The degree of $b_{\eta}$ (resp. $r_{\eta}$) in $H_1$ is 1, $v_{j}$ is one of $u_z,u_{z+1},\ldots,u_{z'}$, and the degree of $q_3$ in $H_1$  is at least 2 and it is in $C_j^*$
\item The degree of $b_{\eta}$ (resp. $r_{\eta}$) in $H_1$ is 1, $v_{j}$ is one of $u_z,u_{z+1},\ldots,u_{z'}$, the degree of $q_1$ in $H_1$ is 1,  the degree of $q_3$ in $H_1$ is 1 or it is not in $C_j^*$ or $q_3$ doesn't exist, and there is a red (resp. blue) point in $C_{j}^*$ whose degree in $H_1$ is at most $t-1$
 \item The degree of $b_{\eta}$ (resp. $r_{\eta}$) in $H_1$ is 1, $v_{j}$ is not one of $u_z,u_{z+1},\ldots,u_{z'+1}$, and the degree of $q_3$ in $H_1$ is at least 2 and it is in $C_j^*$
 \item The degree of $b_{\eta}$ (resp. $r_{\eta}$) in $H_1$ is 1, $v_{j}$ is not one of $u_z,u_{z+1},\ldots,u_{z'+1}$, the degree of $q_3$ in $H_1$ is 1 or it is not in $C_j^*$ or $q_3$ doesn't exist, and there is a red (resp. blue) point in $C_{j}^*$ whose degree in $H_1$ is at most $t-1$
\item The degree of $b_{\eta}$ (resp. $r_{\eta}$) in $H_1$ is 1, $v_{j}=v_i\ne u_l$, and the degree of $q_5$ in $H_1$ is at least 2
 \item The degree of $b_{\eta}$ (resp. $r_{\eta}$) in $H_1$ is 1, $v_{j}=v_i$,  the degree of $q_5$ in $H_1$ is 1 if $v_i\ne u_l$, and there is a red (resp. blue) point in $C_{j}^*$ whose degree in $H_1$ is at most $t-1$

    \end{enumerate}

\end{restatable}

\begin{restatable}{lemma}{TwoPathsNoHangingCycle}\label{lem:2-paths-no-hanging-cycle-exists}
    Suppose $\pi^*$ has at least one switch. Consider any $v_i,v_{i'}\in V^*$ such that either $v_i=v_{i'}=u_{j^h}$ for some $1\le h\le \lambda$ or $v_i=u_l,v_{i'}=u_1$, the edge $(v_{i^1},v_i)$ on $\pi^*$, which is a $0$-edge (resp. $1$-edge), and the edge $(v_{i'},v_{i^2})$ on $\pi^*$, which is a $1$-edge (resp. $0$-edge). Moreover, suppose there is no special 0-hanging cycle (resp. 1-hanging cycle) in $G^*$ for $v_i$ or $v_{i'}$ that is $0$-$1$-edge-disjoint from $\pi^*$. Then, there exist two $0$-$1$-edge-disjoint $0$-paths (resp. $1$-paths) in $G^*$, $\pi_1$ from $v_i$ to a vertex $v_{j^1}$, and $\pi_2$ from $v_{i'}$ to $v_{j^2}$, such that they are also $0$-$1$-edge-disjoint from $\pi^*$ and one of the following properties holds. 
    
    Denote by $e_{\eta^1}$ the last edge on $\pi_1$ if it has an edge or $(v_{i^1},v_i)$ otherwise. The edge in $E_1$ corresponding to $e_{\eta^1}$ is $\{r_{\eta^1},b_{\eta^1}\}$, where $r_{\eta^1}$ is the red point and $b_{\eta^1}$ is the blue point. Also, denote by $e_{\eta^2}$ the last edge on $\pi_2$ if it has an edge or $(v_{i^2},v_{i'})$ otherwise. The edge in $E_1$ corresponding to $e_{\eta^2}$ is $\{r_{\eta^2},b_{\eta^2}\}$, where $r_{\eta^2}$ is the red point and $b_{\eta^2}$ is the blue point. Moreover, if $v_{i'}=u_1$, then $y=1, y'=j^1-1$, and if $v_{i'}=u_{j^h}$, then $y=2$ if $h=1$ and $y=j^{h-1}+1$ otherwise, and $y'=l-1$ if $h=\lambda$ and $y'=j^{h+1}-1$ otherwise.  Also, if $v_i=u_l$, then $z=j^\lambda+1, z'=l$, and if $v_i=u_{j^h}$, then $z=2$ if $h=1$ and $z=j^{h-1}+1$ otherwise, and $z'=l-1$ if $h=\lambda$ and $z'=j^{h+1}-1$ otherwise. 

    \begin{enumerate}[label=(\roman*)]
        \item The degree of both $b_{\eta^1}$ (resp. $r_{\eta^1}$) and $b_{\eta^2}$ (resp. $r_{\eta^2}$) in $H_1$ is at least 2  and $b_{\eta^1}\ne b_{\eta^2}$ (resp. $r_{\eta^1}\ne r_{\eta^2}$)
        
        \item The degree of both $b_{\eta^1}$ (resp. $r_{\eta^1}$) and $b_{\eta^2}$ (resp. $r_{\eta^2}$) in $H_1$ is at least 2, $b_{\eta^1}= b_{\eta^2}$ (resp. $r_{\eta^1}= r_{\eta^2}$), $v_{j^1}$ is one of $u_z,u_{z+1},\ldots,u_{z'}$, and the degree in $H_1$ of the blue (resp. red) point in $C_{j^1}^*$ corresponding to $\pi^*$ is at least 2

        \item The degree of both $b_{\eta^1}$ (resp. $r_{\eta^1}$) and $b_{\eta^2}$ (resp. $r_{\eta^2}$) in $H_1$ are at least 2, $b_{\eta^1}= b_{\eta^2}$ (resp. $r_{\eta^1}= r_{\eta^2}$), $v_{j^1}$ is one of $u_z,u_{z+1},\ldots,u_{z'}$, the degree in $H_1$ of the blue (resp. red) point in $C_{j^1}^*$ corresponding to $\pi^*$ is 1, and there is a red (resp. blue)  point in $C_{j^1}^*$ whose degree in $H_1$ is at most $t-1$

        \item The degree of both $b_{\eta^1}$ (resp. $r_{\eta^1}$) and $b_{\eta^2}$ (resp. $r_{\eta^2}$) in $H_1$ are at least 2, $b_{\eta^1}= b_{\eta^2}$ (resp. $r_{\eta^1}= r_{\eta^2}$), $v_{j^1}$ is not one of $u_z,u_{z+1},\ldots,u_{z'}$, and there is a red (resp. blue)  point in $C_{j^1}^*$ whose degree in $H_1$ is at most $t-1$

        \item The degree of $b_{\eta^1}$ (resp. $r_{\eta^1}$) in $H_1$ is at least 2, $v_{j^2}$ is not one of $u_y,u_{y+1},\ldots,u_{y'-1}$, the degree of $b_{\eta^2}$ (resp. $r_{\eta^2}$) in $H_1$ is 1, and there is a red (resp. blue)  point in $C_{j^2}^*$ whose degree in $H_1$ is at most $t-1$

        \item The degree of $b_{\eta^1}$ (resp. $r_{\eta^1}$) in $H_1$ is at least 2, $v_{j^2}$ is one of $u_y,u_{y+1},\ldots,u_{y'-1}$, the degree of $b_{\eta^2}$ (resp. $r_{\eta^2}$) in $H_1$ is 1, and the degree of the blue (resp. red) point of $C_{j^2}^*$ corresponding to $\pi^*$ is at least 2 in $H_1$   

        \item The degree of $b_{\eta^1}$ (resp. $r_{\eta^1}$) in $H_1$ is at least 2, $v_{j^2}$ is one of $u_y,u_{y+1},\ldots,u_{y'-1}$, the degree of $b_{\eta^2}$ (resp. $r_{\eta^2}$) in $H_1$ is 1, the degree of the blue (resp. red) point in $C_{j^2}^*$ corresponding to $\pi^*$ is 1 in $H_1$, and there is a red (resp. blue)  point in $C_{j^2}^*$ whose degree in $H_1$ is at most $t-1$

        \item The degree of $b_{\eta^1}$ (resp. $r_{\eta^1}$) in $H_1$ is 1 and the degree of $b_{\eta^2}$ (resp. $r_{\eta^2}$) in $H_1$ is at least 2, and there is a red (resp. blue)  point in $C_{j^1}^*$ whose degree in $H_1$ is at most $t-1$

        \item $j^1\ne j^2$, the degree of both $b_{\eta^1}$ (resp. $r_{\eta^1}$) and $b_{\eta^2}$ (resp. $r_{\eta^2}$) are 1 in $H_1$, $v_{j^2}$ is not one of $u_y,u_{y+1},\ldots,u_{y'-1}$, there is a red (resp. blue)  point in $C_{j^1}^*$ whose degree in $H_1$ is at most $t-1$, and there is a red (resp. blue)  point in $C_{j^2}^*$ whose degree in $H_1$ is at most $t-1$

        \item $j^1\ne j^2$, the degree of both $b_{\eta^1}$ (resp. $r_{\eta^1}$) and $b_{\eta^2}$ (resp. $r_{\eta^2}$) are 1 in $H_1$, $v_{j^2}$ is one of $u_y,u_{y+1},\ldots,u_{y'-1}$, the degree of the blue point in $C_{j^2}^*$ corresponding to $\pi^*$ is at least 2 in $H_1$, and there is a red (resp. blue)  point in $C_{j^1}^*$ whose degree in $H_1$ is at most $t-1$

        \item $j^1\ne j^2$, the degree of both $b_{\eta^1}$ (resp. $r_{\eta^1}$) and $b_{\eta^2}$ (resp. $r_{\eta^2}$) are 1 in $H_1$, $v_{j^2}$ is one of $u_y,u_{y+1},\ldots,u_{y'-1}$, the degree of the blue point in $C_{j^2}^*$ corresponding to $\pi^*$ is 1, there is a red (resp. blue)  point in $C_{j^1}^*$ whose degree in $H_1$ is at most $t-1$, and there is a red (resp. blue)  point in $C_{j^2}^*$ whose degree in $H_1$ is at most $t-1$ 
        
        \item $j^1= j^2$, the degree of both $b_{\eta^1}$ (resp. $r_{\eta^1}$) and $b_{\eta^2}$ (resp. $r_{\eta^2}$) are 1 in $H_1$, $C_{j^1}^*=C_{j^2}^*$ has a red (resp. blue) point whose degree in $H_1$ is at most $t-2$  or two red (resp. blue) points whose degree in $H_1$ are at most $t-1$. 
    \end{enumerate}
\end{restatable}

Next, we apply the above lemmas to show the construction of $E_1'$ and $E_2'$. 
\textbf{Recall that in this case there is no $b^{\lambda+1}$-path from $u_l$ to $u_1$ in $G^*$ that is $0$-$1$-edge-disjoint from $\pi^*$.} If $\pi^*$ has a switch, then there is no $0$-path or $1$-path from $u_l$ to $u_1$ in $G^*$ that is $0$-$1$-edge-disjoint from $\pi^*$. Otherwise, it must be that $b^0=b^{\lambda+1}$, and hence by our assumption, there is no $b^0$-path from $u_l$ to $u_1$ in $G^*$ that is $0$-$1$-edge-disjoint from $\pi^*$. We conclude that in this case, there is no $b^0$-path from $u_l$ to $u_1$ in $G^*$ that is $0$-$1$-edge-disjoint from $\pi^*$. 
Then, by Lemma \ref{lem:b-anchor-path-exists-u_1} it follows that, either there is a $(1-b^0)$-hanging cycle for $u_1$ in $G^*$ $0$-$1$-edge-disjoint from $\pi^*$, or a $(1-b^0)$-path starting from $u_1$ in $G^*$ with special properties. This is true, as the parity of $(u_1,u_2)$ is $b^0$. We denote this structure by $\pi(0)$. Additionally, if $\pi(0)$ is a path, we call $b_\eta$ (resp. $r_\eta$)  an \textit{anchor} point if its degree in $H_1$ is at least 2. Similarly, by Lemma \ref{lem:b-anchor-path-exists-u_l} it follows that, either there is a $b^{\lambda+1}$-hanging cycle for $u_l$ in $G^*$ $0$-$1$-edge-disjoint from $\pi^*$, or a $b^{\lambda+1}$-path starting from $u_l$ in $G^*$ with special properties. This is true, as $(u_{l-1},u_l)$ is a $b^{\lambda+1}$-edge. We denote this structure by $\pi(\lambda+1)$. Additionally, if $\pi(\lambda+1)$ is a path, we call $b_\eta$ (resp. $r_\eta$)  an \textit{anchor} point if its degree in $H_1$ is at least 2. 

Now, if $\mathbf{1-b^0\ne b^{\lambda+1}}$, then $\pi(0)$ and $\pi(\lambda+1)$ must be vertex-disjoint. If they are not vertex-disjoint, there exists a $b^{\lambda+1}$-path from $u_l$ to $u_1$ in $G^*$ that is $0$-$1$-edge-disjoint from $\pi^*$: take the edges on $\pi(\lambda+1)$ from $u_l$ to a common vertex and the reverse of $\pi(0)$, from the common vertex to $u_1$. These reverse edges have parity opposite of $1-b^0$, i.e., the same as $b^{\lambda+1}$. But, by our assumption, such a $b^{\lambda+1}$-path does not exist. Hence, $\pi(0)$ and $\pi(\lambda+1)$ are vertex-disjoint. 

In the other case, $\mathbf{1-b^0= b^{\lambda+1}}$. We note that if $\pi^*$ has no switch, $b^0=b^{\lambda+1}$. Thus, if $1-b^0= b^{\lambda+1}$, then we can safely assume that $\pi^*$ has at least one switch. In this case, suppose there are a $(1-b^0)$-hanging cycle for $u_1$ and a $b^{\lambda+1}$-hanging cycle for $u_l$ in $G^*$, such that both are $0$-$1$-edge-disjoint, each of the hanging cycles is $0$-$1$-edge-disjoint from $\pi^*$, and either the special vertices of both are distinct or the special points are the same and the degree of that point in $H_1$ is at least 3. Then, we take the hanging cycle for $u_1$ as $\pi(0)$ and the one for $u_l$ as $\pi(\lambda+1)$. Otherwise, if there is a $(1-b^0)$-hanging cycle for $u_1$ in $G^*$ $0$-$1$-edge-disjoint from $\pi^*$ or a $b^{\lambda+1}$-hanging cycle for $u_l$ in $G^*$ $0$-$1$-edge-disjoint from $\pi^*$, we consider one of those. Assume that the former holds. The other case is symmetric. We take such a hanging cycle for $u_1$ as $\pi(0)$. Then, by Lemma \ref{lem:$1$-path-1-hanging-cycle-exist}, with $v_i=u_l$ and $v_{i'}=u_1$, there exists a $b^{\lambda+1}$-path from $u_l$ in $G^*$ with special properties. This is true, as the parity of $(u_{l-1},u_l)$ is  $b^{\lambda+1}$, the parity of $(u_1,u_2)$ is $b^0=(1-b^{\lambda+1})$, there is no $b^{\lambda+1}$-hanging cycle for $u_l$ in $G^*$ $0$-$1$-edge-disjoint from $\pi^*$ and $\pi(0)$ such that its special point is distinct from that of $\pi(0)$ or if they are the same point, the degree of that point in $H_1$ is at least 3. We take this $b^{\lambda+1}$-path as $\pi(\lambda+1)$. Otherwise, there is neither a $(1-b^0)$-hanging cycle for $u_1$ in $G^*$ $0$-$1$-edge-disjoint from $\pi^*$ nor a $b^{\lambda+1}$-hanging cycle for $u_l$ in $G^*$ $0$-$1$-edge-disjoint from $\pi^*$. Then, by Lemma \ref{lem:2-paths-no-hanging-cycle-exists}, with $v_i=u_l,v_{i'}=u_1$, there exist two $0$-$1$-edge-disjoint paths with parity $1-b^0= b^{\lambda+1}$, from $u_1$ and $u_l$, respectively, such that both are also $0$-$1$-edge-disjoint from $\pi^*$. This is true, as $(u_{l-1},u_l)$ is a $b^{\lambda+1}$-edge and $(u_1,u_2)$ is a $b^0$-edge on $\pi^*$, and $1-b^0= b^{\lambda+1}$. In this case, we take the path from $u_1$ as $\pi(0)$ and the path from $u_l$ as $\pi(\lambda+1)$.  

For all $1\le h\le \lambda$, if there are two $0$-$1$-edge-disjoint $b^h$-hanging cycles for $u_{j^h}$ in $G^*$ that are $0$-$1$-edge-disjoint from $\pi^*$ and have distinct special points or the same special point of degree at least 3 in $H_1$, denote them by $\pi_1(h)$ and $\pi_2(h)$. Otherwise, if there is one $b^h$-hanging cycle for $u_{j^h}$ in $G^*$ that is $0$-$1$-edge-disjoint from $\pi^*$, denote it by $\pi_1(h)$. Now, $(u_{j^h-1},u_{j^h})$ is a $b^h$-edge and $(u_{j^h},u_{j^h+1})$ is a $(1-b^h)$-edge. Then, by Lemma \ref{lem:$1$-path-1-hanging-cycle-exist}, there is a $b^h$-path starting from $u_{j^h}$ in $G^*$ with special properties. Denote this path by $\pi_2(h)$. Otherwise, there is no $b^h$-hanging cycle for $u_{j^h}$ in $G^*$ that is $0$-$1$-edge-disjoint from $\pi^*$. In this case, by Lemma \ref{lem:2-paths-no-hanging-cycle-exists}, there are two $b^h$-paths starting from $u_{j^h}$ in $G^*$ with special properties. Denote them by $\pi_1(h)$ and $\pi_2(h)$. Note that in all the cases, $\pi()$, $\pi_1()$ or $\pi_2()$ can either be a path or a hanging cycle. 

 
\begin{lemma}\label{lem:disjoint-anchor-paths}    
    For any two indexes $1\le h_1\ne h_2\le \lambda$, $\pi_i(h_1)$ and $\pi_j(h_2)$ are vertex-disjoint for $i,j \in \{1,2\}$. For $1\le h_1\le \lambda$ and $i \in \{1,2\}$, $\pi(0)$ and $\pi_i(h_1)$ are vertex-disjoint, and  $\pi(\lambda+1)$ and $\pi_i(h_1)$ are vertex-disjoint. For any $1\le h\le \lambda$, $\pi_1(h)$ and $\pi_2(h)$ are $0$-$1$-edge-disjoint. Moreover, $\pi(0)$ and $\pi(\lambda+1)$ are $0$-$1$-edge-disjoint. 
\end{lemma}

\begin{proof}
    Consider any two paths $\pi_i(h_1)$ and $\pi_j(h_2)$ for $1\le h_1\ne h_2\le \lambda$ and $i,j \in \{1,2\}$. Wlog, assume that $h_1 < h_2$. Suppose $\pi_i(h_1)$ and $\pi_j(h_2)$ have a common vertex $v_j$. If $h_2=h_1+1$, then the parity $b^{h_1}\ne b^{h_2}$. But, this implies that there is a direct $b^{h_1}$-path from $u_{j^{h_1}}$ to $u_{j^{h_2}}$ via $v_j$, and the switch at $u_{j^{h_1}}$ from $b^{h_1}$ to $b^{h_2}$ was not needed. But, this contradicts the optimality of $\pi^*$. Otherwise, $h_2\ge h_1+2$, and there is a path from $u_{j^{h_1}}$ to $u_{j^{h_2}}$ via $v_j$ having at most 1 switch, and the switch at $u_{j^{h_1}}$ or $u_{j^{h_1+1}}$ was not needed. This again contradicts the optimality of $\pi^*$.

    Similarly, suppose $\pi(0)$ and $\pi_i(h_1)$ have a common vertex $v_j$. If $h_1=1$, then there is a direct $(1-b^{0})$-path from $u_{1}$ to $u_{j^{1}}$ via $v_j$, and the switch at $u_{j^{1}}$ from $b^0$ to $1-b^0$ was not needed. But, this contradicts the optimality of $\pi^*$. If $h_1\ge 2$, then there is a path from $u_{1}$ to $u_{j^{h_1}}$ via $v_j$ that has at most 1 switch, and the switch at $u_{j^{h_1}}$ or $u_{j^{h_1-1}}$ was not needed. Similarly, one can prove that $\pi(\lambda+1)$ and $\pi_i(h_1)$ are vertex-disjoint.

    Lastly, the $0$-$1$-edge-disjointness of $\pi_1(h)$ and $\pi_2(h)$, or $\pi(0)$ and $\pi(\lambda+1)$ follows from our construction in Lemma \ref{lem:b-anchor-path-exists-u_l}, \ref{lem:b-anchor-path-exists-u_1}, \ref{lem:$1$-path-1-hanging-cycle-exist}, and \ref{lem:2-paths-no-hanging-cycle-exists}. 
\end{proof}

\medskip
\noindent
{\bf Construction of $E_1'$.} Let $U\subseteq E^*$ be the subset of edges that lie on the structures in $\mathcal{S}=\cup_{i=1}^{\lambda} (\{\pi_1(i)\}\cup \{\pi_2(i)\})\cup \{\pi(0),\pi(\lambda+1),\pi^*\}$. Next, we define the subset $E_1'$ of $E'$ that has a one-to-one mapping with $U$. In particular, consider any edge $(v_i,v_j)$ in $U$. Note that if it is a $0$-edge, it was added due to an edge  $\{p,q\}$ in $E_1$ such that $p\in P_1\cap C_i^*$ and $q\in P_2\cap C_j^*$. We add the edge $\{p,q\}$ to $E_1'$. Otherwise, if $(v_i,v_j)$ is a $1$-edge, it was added due to an edge  $\{p,q\}$ in $E_1$ such that $p\in P_2\cap C_i^*$ and $q\in P_1\cap C_j^*$. In this case, we add the edge $\{p,q\}$ to $E_1'$. 

Next, we show the construction of $E_2'$, and in particular, prove the following lemma. 

\begin{lemma}\label{lem:new-graph-DCS}
 There is a subset of edges $E_2'\subset E$, such that the set of edges $(E'\setminus E_1')\cup E_2'$ form a valid degree-constrained subgraph of $G$ on the set of vertices $P_1\cup P_2$.    
\end{lemma}

In the following, we prove Lemma \ref{lem:new-graph-DCS}. We define the desired subset $E_2'$ of $E$ in the following way, which depends on the edges in $U$. Note that a structure in $\mathcal{S}$ can be a hanging cycle or a path. Consider any such structure $S$ that is a hanging cycle for $v_i$ with $v_j$ as the join vertex. The edges of $S$ correspond to 1 point of $C_{i}^*$, 3 points of $C_j^*$ if $v_j\ne v_i$, otherwise 2 points, and lastly, two points of the cluster corresponding to any other vertex on $S$. Two points in each such cluster of the last type are of different colors and we add an edge between them in $E_2'$. Removal of the edges of $E_1'$ corresponding to $S$ and the addition of these edges to $E_2'$ do not change the degree of these points. Now, we have two cases. First, suppose $v_j\ne v_i$. Among the three points of $C_j^*$, two have degree at least 2 in $H_1$ and they are of the same color. One of them (the first point) corresponds to the incoming edge of $v_j$ on the path portion of $S$ and the second corresponds to the incoming edge of $v_j$ on the cycle portion of $S$, i.e., the special point of $S$. The third point has the opposite color to these two points. We connect this third point with the first point in $E_2'$. Removal of the edges of $E_1'$ corresponding to $S$ and addition of this edge, do not change the degree of the two endpoints. The degree of the other point is reduced by 1 but still is at least 1. This is true, as it is the special point of $S$, and so its degree in $H_1$ was at least 2 before. Now, due to the deletion of the edges on $\pi^*$ incident to $v_i$, the degree of two points, say $p$ and $q$, of $C_{i}^*$ get reduced by 1 if $v_i$ is a switching vertex $u_{j^h}$. If $v_i=u_1$ or $v_i=u_l$, the degree of only one point, say $p$, of $C_{i}^*$ gets reduced by 1. In $E_2'$, we connect the point, say $q'$, of $C_{i}^*$ corresponding to $S$ to either $p$ or one of $p$ and $q$ in the former case, preferably to the one that is not an endpoint of the edges of $E_2'$ so far. If both points are not such an endpoint, choose the point corresponding to $(u_{j^h},u_{j^h+1})$. See Figure \ref{fig:E_2'forHangingCycle} for an illustration. 
If $v_j=v_i$, then also connect $q'$ to $p$ or one of $p$ and $q$ in the same way as above. Removal of the edges of $E_1'$ corresponding to $q'$ and $\pi^*$ and the addition of this edge with $q'$ as an endpoint do not change the degree of the two endpoints. 

\begin{figure}[h]
    \centering
    \includegraphics[width=0.4\linewidth]{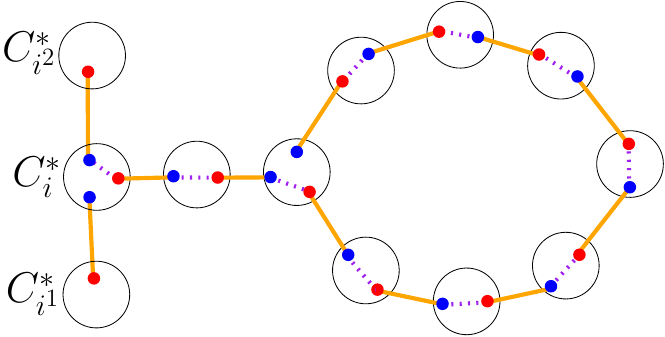}
    \caption{Figure illustrating the construction of $E_2'$ for a hanging cycle. The bold (orange) edges are in $E_1'$ and the dashed (purple) edges are in $E_2'$.}
    \label{fig:E_2'forHangingCycle}
\end{figure}

Now, suppose there is a path $S'$ in $\mathcal{S}$ from a vertex $v_i$ to a vertex $v_j$ that is constructed by Lemma \ref{lem:b-anchor-path-exists-u_1} or \ref{lem:b-anchor-path-exists-u_l}. The edges of $S'$ correspond to 1 point of both $C_{i}^*$ and $C_j^*$, and two points of the cluster corresponding to any intermediate vertex on $S'$. Two points in each such cluster of the last type are of different colors and we add an edge between them in $E_2'$. Removal of the edges of $E_1'$ corresponding to $S'$ and the addition of these edges to $E_2'$ do not change the degree of these points. Now, due to the deletion of the edges on $\pi^*$ incident to $v_i$, the degree of one point, say $p$, of $C_{i}^*$ gets reduced by 1. In $E_2'$, we connect the point, say $q'$, of $C_{i}^*$ corresponding to $S'$ to $p$. Now, if the point in $C_j^*$ corresponding to $S'$ is an {anchor} point, i.e., its degree in $H_1$ is at least 2, we do not add any additional edge. Otherwise, if its degree is 1, then by Lemma \ref{lem:b-anchor-path-exists-u_1} and  \ref{lem:b-anchor-path-exists-u_l}, there is a point of opposite color whose degree in $H_1$ is at most $t-1$. We add an edge to $E_2'$ between these two points in $C_j^*$. See Figure \ref{fig:E_2'forPath} for an illustration. 

\begin{figure}
    \centering
    \includegraphics[width=0.4\linewidth]{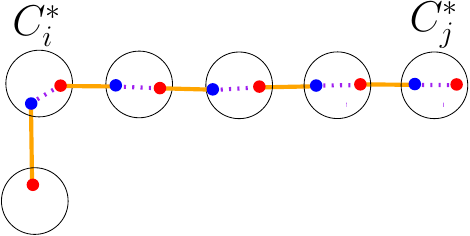}
    \caption{Figure illustrating the construction of $E_2'$ for a path. The bold (orange) edges are in $E_1'$ and the dashed (purple) edges are in $E_2'$. $C_j^*$ has a red point of degree in $H_1$ at most $t-1$.}
    \label{fig:E_2'forPath}
\end{figure}


If $1-b^0\ne b^{\lambda+1}$, then $\pi(0)$ and $\pi(\lambda+1)$ must be vertex-disjoint. In this case, we add the replacement edges by applying the above scheme to them separately. This ensures that the removal of the edges of $E_1'$ corresponding to $\pi(0)\cup \pi(\lambda+1)$ and the addition of the respective edges, do not change the degree of the points except for the following points. For an anchor point or a special point, the degree is decreased by 1. But, its degree remains at least 1, as it was at least two before. For a point in $C_j^*$ corresponding to the last vertex $v_j$ on such a path, its degree might be increased by 1. But, the degree remains at most $t$, as it was at most $t-1$ before.  

Next, we consider the case $1-b^0= b^{\lambda+1}$. As we argued before, in this case, $\pi^*$ has at least one switch. For the following exposition, we consider a general scenario with two vertices $v_i, v_{i'}$ such that either $v_i=v_{i'}=u_{j^h}$ for some $1\le h\le \lambda$ or $v_i=u_l,v_{i'}=u_1$. In the first case, let $\pi_1=\pi_1(h)$ and $\pi_2=\pi_2(h)$. In the second case, let $\pi_1=\pi(0)$ and $\pi_2=\pi(\lambda+1)$. Wlog, we are going to assume that the edge $(v_{i'},v_{i^2})$ on $\pi^*$ is a $1$-edge. The other situation when it is a $0$-edge is symmetric. By our assumption, in the first case, $v_i=v_{i'}$ is a switching vertex, and thus the edge $(v_{i^1},v_i)$ on $\pi^*$ is a $0$-edge. In the second case, $b^0=1$, and thus $b^{\lambda+1}=0$. Hence, in this case also, the edge $(v_{i^1},v_i)$ on $\pi^*$ is a $0$-edge. 

Now, if both $\pi_1$ and $\pi_2$ are $0$-$1$-edge-disjoint hanging cycles with distinct special vertices or same special vertices having degree in $H_1$ at least 3, then we add the replacement edges by applying the above scheme separately on them. This ensures that the removal of the edges of $E_1'$ corresponding to $\pi_1\cup \pi_2\cup \pi^*$ and the addition of the respective edges, do not change the degree of the points except the ones corresponding to the two special vertices. Also, the degree of these latter two points is at least 1. If the two special vertices are the same, its degree is at least 1 afterwards, as its degree was at least 3 before.   

Otherwise, assume that only one of $\pi_1$ or $\pi_2$ is a hanging cycle. Wlog, suppose $\pi_1$ is a hanging cycle for $v_{i'}$ and $\pi_2$ is a path from $v_i$ to a certain vertex $v_j$. We first add the replacement edges for $\pi_1$ using the above procedure. Next, we describe the process of adding the replacement edges for the path $\pi_2$. Consider any intermediate vertex (if any) $v_{\iota}$ on this path. Then, there are exactly two points in $C_{\iota}^*$ corresponding to the edges on $\pi_2$, which are of opposite colors. We add an edge between these two points in $E_2'$. Removal of the edges of $E_1'$ corresponding to $\pi_2$ and the addition of this edge do not change the degree of the two points. Next, we add edges corresponding to $v_j$. Let $e_\eta$ be the last edge on $\pi_2$ if it has an edge or the edge $(v_{i^1},v_i)$ on $\pi^*$ otherwise. By our earlier assumption, $(v_{i^1},v_i)$ is a $0$-edge. Now, if $v_i=u_l$, then $z=j^\lambda+1, z'=l$, and if $v_i=u_{j^h}$, then $z=2$ if $h=1$ and $z=j^{h-1}+1$ otherwise, and $z'=l-1$ if $h=\lambda$ and $z'=j^{h+1}-1$ otherwise.  Let $q_1$ be the blue point in $C_{j}^*$ corresponding to $\pi^*$ in case $v_{j}$ is one of $u_z,u_{z+1},\ldots,u_{z'}$. Let $v_{j'}$ be the join vertex of $\pi_1$. Also, let $q_2$ be the blue point in $C_{j'}^*$ corresponding to the incoming edge to $v_{j'}$ that lies on the cycle of $\pi_1$. Let $q_3$ be the other blue point (if any) in $C_{j'}^*$ corresponding to $\pi_1$. 
Let $q_4$ and $q_5$ be the two blue points in $C_{i}^*$ and $C_{i'}^*$ corresponding to $(v_{i^1},v_i)$ and $(v_{i'},v_{i^2})$ on $\pi^*$. By our scheme of adding replacement edges for hanging cycles, $q_5$ is already connected by an edge in $E_2'$. Now, by Lemma \ref{lem:$1$-path-1-hanging-cycle-exist}, the following cases can occur. 

\medskip
\noindent
(i) $b_{\eta}= q_2$ and the degree of $q_2$ in $H_1$ is at least 3: while adding the edges corresponding to $\pi_1$, the degree of $q_2$ is reduced by 1, but is still at least 2. Once the edge in $E_1'$ corresponding to $e_\eta$ is removed, the degree of $b_\eta$ is at least 1. In this case, we do not add any edge to $E_2'$ corresponding to $b_\eta$.  

\medskip
\noindent
(ii) $b_{\eta}= q_2$, the degree of $q_2$ in $H_1$ is 2, $v_{j}\ne v_i$, and $q_2 = q_3$: note that $j=j'$, as $b_{\eta}= q_2$. By the way of processing $\pi_1$, $q_3$ is already connected in $E_2'$. So, we do not add any more edges corresponding to it. 

\medskip
\noindent
(iii) $b_{\eta}= q_2$, the degree of $q_2$ in $H_1$ is 2, $v_{j}\ne v_i$, $q_2\ne q_3$, and the degree of $q_3$  in $H_1$ is at least 2: note that while processing $\pi_1$, $q_3$ is already connected in $E_2'$ with a red point $p_3$. We remove this edge from $E_2'$ and connect $b_\eta$ with $p_3$ in $E_2'$. The degree of $p_3$ remains the same. The degree of $q_3$ is reduced by 1, but is still at least 1 and not changed further due to Lemma \ref{lem:disjoint-anchor-paths}. Removal of the edges of $E_1'$ corresponding to $\pi_2\cup \{(v_{i^1},v_i),(v_{i'},v_{i^2})\}$ and the addition of the edge $\{b_\eta,p_3\}$, do not change the degree of $b_\eta$. 

\medskip
\noindent
(iv) $b_{\eta}= q_2$, the degree of $q_2$ in $H_1$ is 2, $v_{j}$ is one of $u_z,u_{z+1},\ldots,u_{z'}$, and $q_2 = q_1$: in this case, we do not add any additional edge. However, $q_1$ will be connected in $E_2'$ when the non-switching vertex $v_j$ will be processed. 

\medskip
\noindent
(v) $b_{\eta}= q_2$, the degree of $q_2$ in $H_1$ is 2, $v_{j}$ is one of $u_z,u_{z+1},\ldots,u_{z'}$, $q_1\ne q_2$, and the degree of $q_1$  in $H_1$ is at least 2: note that $v_{j}$ is a non-switching vertex on $\pi^*$. In this case, we connect $b_{\eta}=q_2$ to the red point in $C_{j}^*$ corresponding to the outgoing $0$-edge of $v_j$ on $\pi^*$. By Lemma \ref{lem:disjoint-anchor-paths}, at most 1 such edge is added to $E_2'$ for this red point. Additionally, our construction will ensure that no other edge is added to $E_2'$ for this red point. Thus, the removal of the edges of $E_1'$ and the addition of the edges of $E_2'$, do not change the degree of this red point. 

\medskip
\noindent
(viii) $b_{\eta}= q_2$, the degree of $q_2$ in $H_1$ is 2, $v_{j}=v_i\ne u_l$, and the degree of $q_5$  in $H_1$ is at least 2: note that while processing $\pi_1$, $q_5$ is already connected in $E_2'$ with the red point, say $q'$, of $C_{i}^*$ corresponding to $\pi_1$. We remove this edge from $E_2'$ and connect $b_\eta$ with $q'$ in $E_2'$. The degree of $q'$ remains the same. The degree of $q_5$ is reduced by 1, but is still at least 1 and not changed further due to Lemma \ref{lem:disjoint-anchor-paths}. Removal of the edges of $E_1'$ corresponding to $\pi_2\cup \{(v_{i^1},v_i),(v_{i'},v_{i^2})\}$ and the addition of the edge $\{b_\eta,q'\}$, do not change the degree of $b_\eta$. 

\medskip
\noindent
(x) $b_{\eta}\ne q_2$ and the degree of $b_{\eta}$ in $H_1$ is at least 2: in this case, removal of the edges of $E_1'$ corresponding to $\pi_2\cup \{(v_{i^1},v_i),(v_{i'},v_{i^2})\}$ reduces the degree of $b_\eta$ by 1. As $\pi_2$ is vertex-disjoint from any structure in $\mathcal{S}$ except $\pi_1$ by Lemma \ref{lem:disjoint-anchor-paths}, and $b_{\eta}\ne q_2$, the degree of $b_\eta$ remains at least 1.

\medskip
\noindent
(xi) The degree of $b_{\eta}$  in $H_1$ is 1, $v_{j}$ is one of $u_z,u_{z+1},\ldots,u_{z'}$, and the degree of $q_1$  in $H_1$ is at least 2: note that $v_{j}$ is a non-switching vertex on $\pi^*$. In this case, we connect $b_{\eta}$ to the red point in $C_{j}^*$ corresponding to the outgoing $0$-edge of $v_j$ on $\pi^*$. By Lemma \ref{lem:disjoint-anchor-paths}, at most 1 such edge is added to $E_2'$ for this red point. Additionally, our construction will ensure that no other edge is added to $E_2'$ for this red point. Thus, the removal of the edges of $E_1'$ and the addition of the edges of $E_2'$, do not change the degree of this red point. Also, the degree of $q_1$ will be at least 1.  

\medskip
\noindent
(xii) The degree of $b_{\eta}$ in $H_1$ is 1, $v_{j}$ is one of $u_z,u_{z+1},\ldots,u_{z'}$, and the degree of $q_3$ in $H_1$  is at least 2 and it is in $C_j^*$: note that while processing $\pi_1$, $q_3$ is already connected in $E_2'$ with a red point $p_3$. We remove this edge from $E_2'$ and connect $b_\eta$ with $p_3$ in $E_2'$. The degree of $p_3$ remains the same. The degree of $q_3$ is reduced by 1, but is still at least 1 and not changed further due to Lemma \ref{lem:disjoint-anchor-paths}. Removal of the edges of $E_1'$ corresponding to $\pi_2\cup \{(v_{i^1},v_i),(v_{i'},v_{i^2})\}$ and the addition of the edge $\{b_\eta,p_3\}$, do not change the degree of $b_\eta$.

\medskip
\noindent
(xiv) The degree of $b_{\eta}$  in $H_1$ is 1, $v_{j}$ is not one of $u_z,u_{z+1},\ldots,u_{z'+1}$, and the degree of $q_3$  in $H_1$ 
 is at least 2 and it is in $C_j^*$: same as (xii), because the degree of $q_3\ge 2$ and it is in $C_j^*$. 

\medskip
\noindent
(xvi) The degree of $b_{\eta}$ in $H_1$  is 1, $v_{j}=v_i\ne u_l$, and the degree of $q_5$  in $H_1$ is at least 2: note that while processing $\pi_1$, $q_5$ is already connected in $E_2'$ with the red point, say $q'$, of $C_{i}^*$ corresponding to $\pi_1$. We remove this edge from $E_2'$ and connect $b_\eta$ with $q'$ in $E_2'$. The degree of $q'$ remains the same. The degree of $q_5$ is reduced by 1, but is still at least 1 and not changed further due to Lemma \ref{lem:disjoint-anchor-paths}. Removal of the edges of $E_1'$ corresponding to $\pi_2\cup \{(v_{i^1},v_i),(v_{i'},v_{i^2})\}$ and the addition of the edge $\{b_\eta,q'\}$, do not change the degree of $b_\eta$. 

\medskip
In all the remaining cases, there is a red point in $C_{j}^*$ whose degree in $H_1$ is at most $t-1$. We connect $b_\eta$ to this red point in $E_2'$. By Lemma \ref{lem:disjoint-anchor-paths}, its degree remains at most $t$, as the addition and removal of the edges corresponding to $\pi_1$ doesn't change the degree of any red point. 

\medskip
Next, we add edges corresponding to $v_i$. In all the cases, if $v_j\ne v_i$, we also connect the red point in $C_{i}^*$ corresponding to $\pi_2$ to $q_4$ in $E_2'$. Removal of the edges of $E_1'$ corresponding to $\pi_2\cup \{(v_{i^1},v_i)\}$ and the addition of this edge, do not change the degree of the two endpoints. 


\medskip 
Otherwise, $\pi_1$ is a path from $v_i$ to a vertex $v_{j^1}$, and $\pi_2$ is a path from $v_{i'}$ to a vertex $v_{j^2}$. We describe the process of adding the replacement edges for these two paths. Consider any intermediate vertex (if any) $v_{j'}$ on such a path. Then, there are exactly two points in $C_{j'}^*$ corresponding to the edges on this path, which are of opposite colors. We add an edge between these two points in $E_2'$. Removal of the edges of $E_1'$ corresponding to this path and the addition of the edge does not change the degree of the two points. Next, we add edges corresponding to $v_{j^1}$ and $v_{j^2}$. Wlog, suppose the edge $(v_{i^1},v_i)$ on $\pi^*$ is a $0$-edge, and thus $(v_{i'},v_{i^2})$ on $\pi^*$ is a $1$-edge, as we assumed that their parity are different.  
Denote by $e_{\eta^1}$ the last edge on $\pi_1$ if it has an edge or $(v_{i^1},v_i)$ otherwise. Also, denote by $e_{\eta^2}$ the last edge on $\pi_2$ if it has an edge or $(v_{i^2},v_{i'})$ otherwise. Lastly, if $v_{i'}=u_1$, then $y=1, y'=j^1-1$, and if $v_{i'}=u_{j^h}$, then $y=2$ if $h=1$ and $y=j^{h-1}+1$ otherwise, and $y'=l-1$ if $h=\lambda$ and $y'=j^{h+1}-1$ otherwise.  Also, if $v_i=u_l$, then $z=j^\lambda+1, z'=l$, and if $v_i=u_{j^h}$, then $z=2$ if $h=1$ and $z=j^{h-1}+1$ otherwise, and $z'=l-1$ if $h=\lambda$ and $z'=j^{h+1}-1$ otherwise.  By Lemma \ref{lem:2-paths-no-hanging-cycle-exists}, there are several cases.

\medskip
\noindent
(i) The degree of both $b_{\eta^1}$ and $b_{\eta^2}$ in $H_1$ is at least 2 and $b_{\eta^1}\ne  b_{\eta^2}$: in this case, the removal of the edges of $E_1'$ corresponding to $\pi_1\cup \pi_2\cup \{(v_{i^1},v_i),(v_{i'},v_{i^2})\}$, reduces the degree of both $b_{\eta^1}$ and $b_{\eta^2}$ by 1. As $\pi_1$ and $\pi_2$ are vertex-disjoint from any other structure in $\mathcal{S}$ by Lemma \ref{lem:disjoint-anchor-paths}, the degree of both $b_{\eta^1}$ and $b_{\eta^2}$ remain at least 1. 

\medskip
\noindent
(ii) The degree of both $b_{\eta^1}$ and $b_{\eta^2}$ in $H_1$ is at least 2, $b_{\eta^1}= b_{\eta^2}$, $v_{j^1}$ is one of $u_z,u_{z+1},\ldots,u_{z'}$, and the degree of the blue point in $C_{j^1}^*$ corresponding to $\pi^*$, say $p_3$, is at least 2 in $H_1$: thus after the removal of the edge of $E_1'$ corresponding to $p_3$, the degree of $p_3$ is at least 1. Note that $v_{j^1}$ is a non-switching vertex on $\pi^*$. In this case, we connect $b_{\eta^1}=b_{\eta^2}$ to the red point in $C_{j^1}^*$ corresponding to the outgoing $0$-edge of $v_{j^1}$ on $\pi^*$. By Lemma \ref{lem:disjoint-anchor-paths}, at most 1 such edge is added to $E_2'$ for this red point. Additionally, our construction will ensure that no other edge is added to $E_2'$ for this red point. Thus, the removal of the edges of $E_1'$ and the addition of the edges of $E_2'$, do not change the degree of this red point. 
        
\medskip
\noindent
(iii) The degree of both $b_{\eta^1}$ and $b_{\eta^2}$ in $H_1$ are at least 2, $b_{\eta^1}= b_{\eta^2}$, $v_{j^1}$ is one of $u_z,u_{z+1},\ldots,u_{z'}$, and the degree of the blue point in $C_{j^1}^*$ corresponding to $\pi^*$ is 1 in $H_1$: in this case, we know that there is a red point in $C_{j^1}^*$ whose degree in $H_1$ is at most $t-1$. We connect $b_{\eta^1}$ to this red point in $E_2'$. By Lemma \ref{lem:disjoint-anchor-paths}, at most 1 such edge is added to $E_2'$ for this red point. Thus, its degree remains at most $t$.

\medskip
\noindent
(iv)  The degree of both $b_{\eta^1}$ and $b_{\eta^2}$ in $H_1$ are at least 2, $b_{\eta^1}= b_{\eta^2}$, $v_{j^1}$ is not one of $u_z,u_{z+1},\ldots,u_{z'}$: this case is same as (iii), as here also there is a red point in $C_{j^1}^*$ whose degree in $H_1$ is at most $t-1$.

\medskip
\noindent
(v) The degree of $b_{\eta^1}$ in $H_1$  is at least 2, $v_{j^2}$ is not one of $u_y,u_{y+1},\ldots,u_{y'-1}$, and the degree of $b_{\eta^2}$ in $H_1$ is 1: we know that there is a red point in $C_{j^2}^*$ whose degree in $H_1$ is at most $t-1$. We connect $b_{\eta^2}$ to this red point in $E_2'$. By Lemma \ref{lem:disjoint-anchor-paths}, at most 1 such edge is added to $E_2'$ for this red point. We do not add any additional edge for $b_{\eta^1}$, as its degree is at least 2.  

\medskip
\noindent
(vi) The degree of $b_{\eta^1}$ in $H_1$  is at least 2, $v_{j^2}$ is one of $u_y,u_{y+1},\ldots,u_{y'-1}$, the degree of $b_{\eta^2}$ in $H_1$  is 1, and the degree of the blue point in $C_{j^2}^*$ corresponding to $\pi^*$, say $p_4$, is at least 2 in $H_1$: thus after the removal of the edge of $E_1'$ corresponding to $p_4$, the degree of $p_4$ is at least 1. Note that $v_{j^2}$ is a non-switching vertex on $\pi^*$. In this case, we connect $b_{\eta^2}$ to the red point in $C_{j^2}^*$ corresponding to the outgoing $0$-edge of $v_{j^2}$ on $\pi^*$. By Lemma \ref{lem:disjoint-anchor-paths}, at most 1 such edge is added to $E_2'$ for this red point. Additionally, our construction will ensure that no other edge is added to $E_2'$ for this red point. Thus, the removal of the edges of $E_1'$ and the addition of the edges of $E_2'$, do not change the degree of this red point. We do not add any additional edge for $b_{\eta^1}$, as its degree is at least 2.  

\medskip
\noindent
(vii) The degree of $b_{\eta^1}$  in $H_1$ is at least 2, $v_{j^2}$ is one of $u_y,u_{y+1},\ldots,u_{y'-1}$, the degree of $b_{\eta^2}$  in $H_1$ is 1, and the degree of the blue point in $C_{j^2}^*$ corresponding to $\pi^*$ is 1: this case is the same as (v), as we know that there is a red point in $C_{j^2}^*$ whose degree in $H_1$ is at most $t-1$.

\medskip
\noindent
(viii) The degree of $b_{\eta^1}$ in $H_1$ is 1 and the degree of $b_{\eta^2}$ in $H_1$ is at least 2: we know that there is a red point in $C_{j^1}^*$ whose degree in $H_1$ is at most $t-1$. This case is also the same as (iii), as we do not add any additional edge for $b_{\eta^2}$. 

\medskip
\noindent
(ix) $j^1\ne j^2$, the degree of both $b_{\eta^1}$ and $b_{\eta^2}$ are 1 in $H_1$, and $v_{j^2}$ is not one of $u_y,u_{y+1},\ldots,u_{y'-1}$: we know that there is a red point $q_1$ in $C_{j^1}^*$ whose degree in $H_1$ is at most $t-1$ and there is a red point $q_2$ in $C_{j^2}^*$ whose degree in $H_1$ is at most $t-1$. We connect $b_{\eta^1}$ to $q_1$ and $b_{\eta^2}$ to $q_2$ in $E_2'$. By Lemma \ref{lem:disjoint-anchor-paths}, at most 1 such edge is added to $E_2'$ for both $q_1$ and $q_2$.  Thus, their degrees remain at most $t$. 

\medskip
\noindent
(x) $j^1\ne j^2$, the degree of both $b_{\eta^1}$ and $b_{\eta^2}$ are 1 in $H_1$, $v_{j^2}$ is one of $u_y,u_{y+1},\ldots,u_{y'-1}$, and the degree of the blue point in $C_{j^2}^*$ corresponding to $\pi^*$, say $p_4$, is at least 2: thus after the removal of the edge of $E_1'$ corresponding to $p_4$, the degree of $p_4$ is at least 1. Note that $v_{j^2}$ is a non-switching vertex on $\pi^*$. In this case, we connect $b_{\eta^2}$ to the red point in $C_{j^2}^*$ corresponding to the outgoing $0$-edge of $v_{j^2}$ on $\pi^*$. By Lemma \ref{lem:disjoint-anchor-paths}, at most 1 such edge is added to $E_2'$ for this red point. Additionally, our construction will ensure that no other edge is added to $E_2'$ for this red point. Thus, the removal of the edges of $E_1'$ and the addition of the edges of $E_2'$, do not change the degree of this red point. Now, we know that there is a red point in $C_{j^1}^*$ whose degree in $H_1$ is at most $t-1$. We connect $b_{\eta^1}$ to this red point. By Lemma \ref{lem:disjoint-anchor-paths}, at most 1 such edge is added to $E_2'$ for this point.  Thus, its degree remains at most $t$. 

\medskip
\noindent
(xi) $j^1\ne j^2$, the degree of both $b_{\eta^1}$ and $b_{\eta^2}$ are 1 in $H_1$, $v_{j^2}$ is one of $u_y,u_{y+1},\ldots,u_{y'-1}$, and the degree of the blue point in $C_{j^2}^*$ corresponding to $\pi^*$ is 1: we know that there is a red point $q_1$ in $C_{j^1}^*$ whose degree in $H_1$ is at most $t-1$ and there is a red point $q_2$ in $C_{j^2}^*$ whose degree in $H_1$ is at most $t-1$. So, this case is the same as (ix). 

\medskip
\noindent
(xii) $j^1= j^2$, the degree of both $b_{\eta^1}$ and $b_{\eta^2}$ are 1 in $H_1$: first of all, $b_{\eta^1}$ and $b_{\eta^2}$ are not the same point as their degree are 1 and $\pi_1$, $\pi_2$ are $0$-$1$-edge-disjoint. In this case, $C_{j^1}^*=C_{j^2}^*$ has a red point $q_3$ whose degree in $H_1$ is at most $t-2$ or two red points $q_1$ and $q_2$ whose degree in $H_1$ are at most $t-1$. If such a point $q_3$ exists, we connect both $b_{\eta^1}$ and $b_{\eta^2}$ to $q_3$ in $E_2'$. By Lemma \ref{lem:disjoint-anchor-paths}, at most 2 such edges are added to $E_2'$ for $q_3$.  Thus, its degree remains at most $t$. Otherwise, we connect $b_{\eta^1}$ to $q_1$ and $b_{\eta^2}$ to $q_2$ in $E_2'$. By Lemma \ref{lem:disjoint-anchor-paths}, at most 1 such edge is added to $E_2'$ for both $q_1$ and $q_2$. Thus, their degrees remain at most $t$. 

\medskip
\noindent
Next, we add edges corresponding to $v_i$ and $v_{i'}$. Let $p_1$ and $p_2$ be the two blue points in $C_{i}^*$ and $C_{i'}^*$ corresponding to $(v_{i^1},v_i)$ and $(v_{i'},v_{i^2})$ on $\pi^*$. If $v_{j^1}\ne v_i$, we connect the red point in $C_{i}^*$ corresponding to $\pi_1$ to $p_1$ in $E_2'$. If $v_{j^2}\ne v_{i'}$, we connect the red point in $C_{i'}^*$ corresponding to $\pi_2$ to $p_2$ in $E_2'$. The removal of the edges of $E_1'$ corresponding to $\pi_1\cup \pi_2\cup \{(v_{i^1},v_i),(v_{i'},v_{i^2})\}$ and the addition of these two edges, do not change the degree of the endpoints. 

In the above procedure, we have described how to add replacement edges for all the structures in $\mathcal{S}=(\cup_{i=1}^{\lambda} \pi_1(i)\cup \pi_2(i))\cup \{\pi(0),\pi(\lambda+1),\pi^*\}$, except $\pi^*$. In particular, we have added edges corresponding to the vertices $u_1,u_{j^1},\ldots,u_{j^\lambda},u_l$ on $\pi^*$. Next, we consider a non-switching vertex $v_s$ on $\pi^*$. If this vertex is already processed in the above procedure, we do not add any edges to $E_2'$ for $v_s$. In particular, two points in $C_s^*$ corresponding to the edges of $v_s$ on $\pi^*$ have degrees at least 1 after deletion of the edges of $E_1'$ and the addition of the edges of $E_2'$ constructed so far. Otherwise, we connect these two points with an edge in $E_2'$. As $v_s$ is a switching vertex, these two points have opposite colors. Also, the removal of the edges of $E_1'$ and the addition of the edges of $E_2'$, do not change the degree of the two endpoints of the added edge. 

This finishes the definition of the set $E_2'$. Then, Lemma \ref{lem:new-graph-DCS} follows by the above construction. 
Next, we analyze the total weight of the edges in $E_2'$. 

\begin{lemma}\label{lem:cost-of-E_2'}
    $w(E_2')\le 6\cdot \sum_{i=1}^\tau r(C_i^*)$. Moreover, if $\pi^*$ does not have a switch, $w(E_2')\le 4\cdot \sum_{i=1}^\tau r(C_i^*)$.
\end{lemma}

\begin{proof}
    The way we add the edges to $E_2'$, both endpoints of each edge lie in a cluster $C_j^*$ such that the vertex corresponding to the cluster lies on a structure in $\mathcal{S}$. Now, by Lemma \ref{lem:disjoint-anchor-paths} it follows that if a vertex is shared by two paths in $(\cup_{h=1}^{\lambda} \pi_1(h)\cup \pi_2(h))\cup \{\pi(0),\pi(\lambda+1)\}$, then it is either on $\pi_1(h)$ and $\pi_2(h)$, on $\pi(0)$ and $\pi_i(h)$, or on $\pi(\lambda+1)$ and $\pi_i(h)$, where $1\le h\le \lambda$ and $i\in \{1,2\}$. Thus, additionally considering $\pi^*$, each such vertex can be on at most three structures in $\mathcal{S}$. Fix such a vertex $v_j$ and its corresponding cluster $C_j^*$. Now, consider the processing of any such path that contains $v_j$. During this processing, we add at most one edge to $E_2'$ corresponding to $v_j$, whose endpoints are in $C_j^*$. Thus, $E_2'$ contains at most three edges such that their endpoints lie in $C$. The sum of the weights of these three edges is at most 3 times the diameter of $C_j^*$. Summing over all clusters, we obtain the lemma.  

    We note that if $\pi^*$ has no switch, $b^0=b^{\lambda+1}$. As we argued before, in this case, $\pi(0)$ and $\pi(\lambda+1)$ are vertex-disjoint. Thus, each vertex $v_j$ can appear in at most two paths in $\{\pi(0),\pi(\lambda+1),\pi^*\}$. Hence, the improved 4-factor follows in this case. 
\end{proof}

\subsection{Proof of Lemma \ref{lem:b-anchor-path-exists-u_l}}
\anchorPathUl*

\begin{proof}
    We prove the existence of the $0$-path with the desired properties. The proof for $1$-path is similar. 
    
    Consider the graph $G_1$ constructed in the following way from $G^*$. First, remove all the $1$-edges and the edges on $\pi^*$ from $G^*$. While there is a cycle in $G^*$, remove all the edges of this cycle from $G^*$ and repeat this step. When the above procedure ends, we are left with a graph without any cycle. Let us denote this graph by $G_1$. 

    Consider any maximal path $\pi_1$ in $G_1$ starting from $u_l$. Let $v_j$ be the last vertex on this path. As $\pi_1$ is maximal, there is no outgoing edge of $v_j$ in $G_1$. Note that $u_l$ can be same as $v_j$. 
    If $u_l=v_j$, set $e_\eta=(u_{l-1},u_l)$. Otherwise, $e_\eta$ is the last edge on $\pi_1$. We claim that $\pi_1$ is the desired 0-path that satisfies (i) or (ii). 

    Let $\pi_1^*$ be the $u_l$ to $u_1$ path consists of the edges that are reverse of the edges on $\pi^*$, i.e., a $b$-edge on $\pi^*$ gives rise to a $(1-b)$-edge of $\pi_1^*$, where $b\in \{0,1\}$. Now, suppose $\pi_1$ and $\pi^*$ has a common vertex $v_c$. Then by Observation \ref{obs:range-of-b^h-path}, $v_c$ is one of $u_{z},\ldots,u_{l}$ where $z=j^\lambda+1$ if $\pi^*$ has at least one switch, and $z=1$ otherwise. Now, if $\pi^*$ has at least one switch, $G_1$ does not contain any edge on $\pi_1^*$ between $u_l$ and $u_{j^\lambda}$, as each such edge is a 1-edge ($(u_{l-1},u_l)$ is a 0-edge). Also, $G_1$ does not contain the edges of $\pi^*$. Thus, in this case, $\pi_1$ is 0-1-edge-disjoint from $\pi^*$. Otherwise, $\pi^*$ does not have a switch, and thus $\pi_1^*$ is a 1-path. So, its edges are not in $G_1$, and in this case also, $\pi_1$ is 0-1-edge-disjoint from $\pi^*$. 
    
    Now, if the degree of $b_\eta$ in $H_1$ is at least 2, then $\pi_1$ satisfies (i). Suppose its degree is 1. We know that there is no outgoing 0-edge in $G_1$ from $v_j$. However, there can be outgoing 0-edges from $v_j$ in $G^*$. Each such edge was deleted from $G^*$, and thus either it is on $\pi^*$ or it must have been part of a cycle in $G^*$. Let $n_1$ be the number of outgoing 0-edges from $v_j$ in $G^*$. Suppose such an edge is on $\pi^*$, i.e., $v_j$ is on $\pi^*$ and $v_j\ne u_l$. As argued before, $v_j$ is one of $u_{z},\ldots,u_{l-1}$ where $z=j^\lambda+1$ if $\pi^*$ has at least one switch, and $z=1$ otherwise. If $v_j=u_1$, then $\pi_1$ is a 0-path from $u_l$ to $u_1$ in $G_1$ that is 0-1-edge-disjoint from $\pi^*$. But, this is a contradiction, and hence $v_j\ne u_1$. Thus, in both the cases, if there is an outgoing 0-edge of $v_j$ on $\pi^*$, there is also an incoming 0-edge of $v_j$ on $\pi^*$. For any other outgoing 0-edge of $v_j$, it is part of a cycle in $G^*$ that was removed. Hence, there are also $n_1$ incoming 0-edges to $v_j$ in $G^*$. Now, consider the $0$-edge $e_\eta$. If $e_\eta$ is the last edge on $\pi_1$, it is in $G_1$ and thus is not on $\pi^*$ or on a cycle that was removed. Otherwise, $u_l=v_j$, and $e_\eta=(u_{l-1},u_l)$, which was on $\pi^*$ and thus was not on a cycle that was removed. Also, $u_l$ does not have an outgoing 0-edge on $\pi^*$. Thus, additionally considering $e_\eta$, there are at least $n_1+1$ incoming 0-edges to $v_j$ in $G^*$. 
    
    Let $n_r$ and $n_b$ be the number of red and blue points in $C_j^*$, respectively. Also, let $n_2$ be the number of edges in $E_1$ across red and blue points in $C_j^*$. Now, we know that the degree of $b_\eta$ is 1 in $H_1$. If $t=1$, the edges of $H_1$ form a perfect matching, i.e., the degree of each vertex in $H_1$ is 1. Then, $n_2 =n_r-n_1$, and the number of edges incident on the blue vertices in $C_j^*$ is at least $n_2+n_1+1=n_r-n_1+n_1+1=n_r+1$. But, this implies that $n_b\ge n_r+1$, which contradicts $1$-balance in $C_j^*$. So, assume that $t\ge 2$. We claim that there is a red point in $C_j^*$ whose degree in $H_1$ is at most $t-1$, which satisfies (ii). 

    For the sake of contradiction, assume that all the red points in $C_j^*$ have degree $t$ in $H_1$. Then, $n_2 =t\cdot n_r-n_1$, and the number of edges incident on the blue vertices in $C_j^*$ is at least $n_2+n_1+1=t\cdot n_r-n_1+n_1+1=t\cdot n_r+1$. 
    
    Note that there are at least $n_1-1$ outgoing 0-edges from $v_j$, as $\pi^*$ contains at most one outgoing 0-edge of $v_j$. Each such edge is part of a different cycle, and thus there is also an incoming edge to $v_j$ that is part of the same cycle. However, each such incoming edge must correspond to a degree 1 blue vertex in $C_j^*$. Otherwise, the union of $\pi_1$ and such a cycle form a special 0-hanging cycle for $u_l$. Moreover, similar to $\pi_1$, the vertices of $\pi^*$ that this hanging cycle can contain are $u_{z},\ldots,u_{l}$ where $z=j^\lambda+1$ if $\pi^*$ has at least one switch, and $z=2$ otherwise. Hence, this hanging cycle is 0-1-edge-disjoint from $\pi^*$, which is a contradiction. Thus, there must be $n_1$ such blue points with degree 1. Now, the red points in $C_j^*$ are connected to the blue points in $C_j^*$ using $n_2 =t\cdot n_r-n_1$ edges. As $H_1$ is a collection of stars, and each red point in $C_j^*$ has degree $t\ge 2$, the degree of these blue points must be 1. Now, we know that the incoming edge $e_\eta$ to $v_j$ is not part of a cycle that was removed. Also, the degree of $b_\eta$ is 1, and so it is not connected to a red point of $C_j^*$ in $H_1$. Thus, the number of blue points in $C_j^*$ with degree 1 in $H_1$ is $(n_1-1)+n_2+1=t\cdot n_r$. As $C_j^*$ can contain at most $t\cdot n_r$ blue points due to its $t$-balance, the degree of all blue points of $C_j^*$ in $H_1$ is 1. Thus, the sum of their degree is $t\cdot n_r$. But, this contradicts the fact that they are incident to at least $t\cdot n_r+1$ edges. Thus our claim must be true, which completes the proof of the lemma.     
\end{proof}

\subsection{Proof of Lemma \ref{lem:b-anchor-path-exists-u_1}}
\anchorPathUone*

\begin{proof}
    We prove the existence of the $0$-path with the desired properties. The proof for $1$-path is similar. 
    
    Let $\pi_1^*$ be the $u_l$ to $u_1$ path consisting of the edges that are reverse of the edges on $\pi^*$, i.e., a $b$-edge on $\pi^*$ gives rise to a $(1-b)$-edge of $\pi_1^*$, where $b\in \{0,1\}$. Consider the graph $G_1$ constructed in the following way from $G^*$. First, remove all the $1$-edges and the edges on $\pi_1^*$ from $G^*$. While there is a cycle in $G^*$, remove all the edges of this cycle from $G^*$ and repeat this step. When the above procedure ends, we are left with a graph without any cycle. Let us denote this graph by $G_1$. 

    Consider any maximal path $\pi_1$ in $G_1$ starting from $u_1$. Let $v_j$ be the last vertex on this path. As $\pi_1$ is maximal, there is no outgoing edge of $v_j$ in $G_1$. Note that $u_1$ can be same as $v_j$. 
    If $u_1=v_j$, set $e_\eta=(u_{2},u_1)$. Otherwise, $e_\eta$ is the last edge on $\pi_1$. We claim that $\pi_1$ is the desired 0-path that satisfies (i) or (ii). 

    Now, suppose $\pi_1$ and $\pi^*$ has a common vertex $v_c$. Then by Observation \ref{obs:range-of-b^h-path}, $v_c$ is one of $u_{1},\ldots,u_{z}$ where $z=j^1-1$ if $\pi^*$ has at least one switch, and $z=l$ otherwise. Now, if $\pi^*$ has at least one switch, $G_1$ does not contain any edge on $\pi^*$ between $u_1$ and $u_{j^1}$, as each such edge is a 1-edge ($(u_1,u_2)$ is a 1-edge). Also, $G_1$ does not contain the edges of $\pi_1^*$. Thus, in this case, $\pi_1$ is 0-1-edge-disjoint from $\pi^*$. Otherwise, $\pi^*$ does not have a switch, and it is a 1-path. So, its edges are not in $G_1$. So, in this case also, $\pi_1$ is 0-1-edge-disjoint from $\pi^*$. 
    
    Now, if the degree of $b_\eta$ in $H_1$ is at least 2, then $\pi_1$ satisfies (i). So, suppose its degree is 1. We know that there is no outgoing 0-edge in $G_1$ from $v_j$. However, there can be outgoing 0-edges from $v_j$ in $G^*$. Each such edge was deleted from $G^*$, and thus either it is on $\pi_1^*$ or it must have been part of a cycle in $G^*$. Let $n_1$ be the number of outgoing 0-edges of $v_j$ in $G^*$. Suppose such an edge is on $\pi_1^*$, i.e., $v_j$ is on $\pi_1^*$ and $v_j\ne u_1$. As argued before, $v_j$ is one of $u_{2},\ldots,u_{z}$ where $z=j^1-1$ if $\pi^*$ has at least one switch, and $z=l$ otherwise. If $v_j=u_l$, then $\pi_1$ is a 0-path from $u_1$ to $u_l$ in $G_1$ that is 0-1-edge-disjoint from $\pi^*$. By taking the reverse of the edges on $\pi_1$, there is a 1-path from $u_l$ to $u_1$ in $G^*$ that is 0-1-edge-disjoint from $\pi^*$. But, this is a contradiction, and hence $v_j\ne u_l$. Thus, in both the cases, if there is an outgoing 0-edge of $v_j$ on $\pi_1^*$, there is also an incoming 0-edge of $v_j$ on $\pi_1^*$. For any other outgoing 0-edge of $v_j$, it is part of a cycle in $G^*$ that was removed. Hence, there are also $n_1$ incoming 0-edges to $v_j$ in $G^*$. 
    
    Now, consider the $0$-edge $e_\eta$. If $e_\eta$ is the last edge on $\pi_1$, it is in $G_1$ and thus is not on $\pi_1^*$ or on a cycle that was removed. Otherwise, $v_j=u_1$, and $e_\eta=(u_{2},u_1)$, which was on $\pi_1^*$ and thus was not on a cycle that was removed. Also, $u_1$ does not have an outgoing 0-edge on $\pi_1^*$. Thus, additionally considering $e_\eta$, there are at least $n_1+1$ incoming 0-edges to $v_j$ in $G^*$. 
    
    Let $n_r$ and $n_b$ be the number of red and blue points in $C_j^*$, respectively. Also, let $n_2$ be the number of edges in $E_1$ across red and blue points in $C_j^*$. 
    Now, we know that the degree of $b_\eta$ is 1 in $H_1$. 
    If $t=1$, the edges of $H_1$ form a perfect matching, i.e., the degree of each vertex in $H_1$ is 1. Then, $n_2 =n_r-n_1$, and the number of edges incident on the blue vertices in $C_j^*$ is at least $n_2+n_1+1=n_r-n_1+n_1+1=n_r+1$. But, this implies that $n_b\ge n_r+1$, which contradicts $1$-balance in $C_j^*$. So, assume that $t\ge 2$. We claim that there is a red point in $C_j^*$ whose degree in $H_1$ is at most $t-1$, which satisfies (ii). 

    For the sake of contradiction, assume that all the red points in $C_j^*$ have degree $t$ in $H_1$. Then, $n_2 =t\cdot n_r-n_1$, and the number of edges incident on the blue vertices in $C_j^*$ is at least $n_2+n_1+1=t\cdot n_r-n_1+n_1+1=t\cdot n_r+1$. 
    
    Note that there are $n_1-1$ outgoing 0-edges of $v_j$ each of which is part of a different cycle, and thus there is also an incoming edge to $v_j$ that is part of the same cycle. However, each such incoming edge must correspond to a degree 1 blue vertex in $C_j^*$. Otherwise, the union of $\pi_1$ and such a cycle forms a special 0-hanging cycle for $u_1$. Moreover, similar to $\pi_1$, the vertices of $\pi^*$ that this hanging cycle can contain are $u_{1},\ldots,u_{z}$ where $z=j^1-1$ if $\pi^*$ has at least one switch, and $z=l-1$ otherwise. Hence, this hanging cycle is 0-1-edge-disjoint from $\pi^*$, which is a contradiction. Thus, there must be $n_1$ such blue points with degree 1. Now, the red points in $C_j^*$ are connected to the blue points in $C_j^*$ using $n_2 =t\cdot n_r-n_1$ edges. As $H_1$ is a collection of stars, and each red point in $C_j^*$ has degree $t\ge 2$, the degree of these blue points must be 1. Now, we know that the incoming edge $e_\eta$ to $v_j$ is not part of a cycle that was removed. Also, the degree of $b_\eta$ is 1, and so it is not connected to a red point of $C_j^*$ in $H_1$. Thus, the number of blue points in $C_j^*$ with degree 1 in $H_1$ is $(n_1-1)+n_2+1=t\cdot n_r$. As $C_j^*$ can contain at most $t\cdot n_r$ blue points due to its $t$-balance, the degree of all blue points of $C_j^*$ in $H_1$ is 1. Thus, the sum of their degree is $t\cdot n_r$. But, this contradicts the fact that they are incident to at least $t\cdot n_r+1$ edges. Thus our claim must be true, which completes the proof of the lemma.        
\end{proof}
\subsection{Proof of Lemma \ref{lem:$1$-path-1-hanging-cycle-exist}}

\onePathOneHangingCycle*

\begin{proof}
    We prove the existence of the $0$-path with the desired properties. The proof for $1$-path is similar. 
    
    First, we remove the edges of $O$ from $G^*$ to obtain a new graph $G_1'$. By our assumption, there is no special 0-hanging cycle in $G_1'$ for $v_i$ that is 0-1-edge-disjoint from $\pi^*$ and its special vertex is either distinct from $q_2$ or the same as $q_2$ and has degree in $H_1$ at least 3.  
    
    Let $\pi_1^*$ be the $u_l$ to $u_1$ path consists of the edges that are reverse of the edges on $\pi^*$, i.e., a $b$-edge on $\pi^*$ gives rise to a $(1-b)$-edge of $\pi_1^*$, where $b\in \{0,1\}$. Consider the graph $G_1$ constructed in the following way from $G_1'$. First, remove all the $1$-edges and the edges on $\pi^*$ and $\pi_1^*$ from $G_1'$. While there is a cycle in $G_1'$, remove all the edges of this cycle from $G_1'$ and repeat this step. When the above procedure ends, we are left with a graph without any cycle. Let us denote this graph by $G_1$. 

    Consider any maximal path $\pi_1$ in $G_1$ starting from $v_i$. Let $v_j$ be the last vertex on this path. As $\pi_1$ is maximal, there is no outgoing 0-edge of $v_j$ in $G_1$. Note that $v_i$ can be same as $v_j$. As $G_1$ does not contain any edge of $\pi^*$ or $\pi_1^*$, $\pi_1$ is 0-1-edge-disjoint from $\pi^*$. Now, as $\pi_1$ is a 0-path, it is also 0-1-edge-disjoint from $O$ whose edges were removed. We claim that $\pi_1$ is the desired 0-path. 
    
    Consider the cluster $C_j^*$ corresponding to $v_j$. 
    If $v_i=v_j$, set $e_\eta=(v_{i^1},v_i)$. Otherwise, $e_\eta$ is the last edge on $\pi_1$. The edge in $E_1$ corresponding to $e_{\eta}$ is $\{r_{\eta},b_{\eta}\}$, where $r_{\eta}$ is the red point and $b_{\eta}$ is the blue point. As $e_\eta$ is a 0-edge, $b_{\eta}$ belongs to $C_j^*$.      

    
    Let $n_1$ be the number of outgoing 0-edges of $v_j$ in $G^*$, and $n_r$ and $n_b$ be the number of red and blue points in $C_j^*$, respectively. Also, let $n_2$ be the number of edges in $E_1$ across red and blue points in $C_j^*$. 

    \begin{claim}\label{cl:incoming-cluster-j2-1}
        The number of incoming 0-edges of $v_{j}$ in $G^*$ is at least $n_1+1$. 
    \end{claim}    
    
    \begin{proof} 
    We know that there is no outgoing 0-edge in $G_1$ from $v_j$; otherwise, $\pi_1$ would not be maximal. However, there might exist outgoing 0-edges of $v_j$ in $G^*$. Each such edge $e'$ was removed from $G^*$. Thus, it is part of either $O$, $\pi^*$, $\pi_1^*$, or a cycle in $G_1'$ that was removed. Suppose $e'$ is part of $O$. Then, if it is on the path part of $O$, {there is an incoming 0-edge of $v_j$ that is on $O$}, unless $v_j$ is the first vertex on this path. In the latter case, $v_j=v_{i'}$. Then, $v_{i}\ne u_l$, as otherwise there is a zero switch path from $u_l$ to $v_{i'}=u_1$, which is a contradiction. Thus, $v_i=u_{j^h}$, and there is an incoming 0-edge $(v_{i^2},v_i)$, which is on $\pi_1^*$ and not on a cycle in $G_1'$ that was removed. Otherwise, $e'$ is on the cycle in the hanging cycle $O$, and there is at least one (and at most 2) incoming 0-edge of $v_j$ that is on $O$. 
    Suppose $v_{j}$ is on $\pi^*$ or equivalently on $\pi_1^*$.         
    Then by Observation \ref{obs:range-of-b^h-path}, $v_j$ is one of $u_z,u_{z+1},\ldots,u_{z'}$, where if $v_i=u_l$, then $z=j^\lambda+1, z'=l$, and if $v_i=u_{j^h}$, then $z=2$ if $h=1$ and $z=j^{h-1}+1$ otherwise, and $z'=l-1$ if $h=\lambda$ and $z'=j^{h+1}-1$ otherwise. 
    First, suppose $v_i=u_l$. As $(v_{i^1},v_i)$ is a 0-edge, all the edges on $\pi^*$ between $u_{z-1}$ and $u_{z'+1}$ are 0-edges, and hence the edges on $\pi_1^*$ between them are 1-edges. Thus, if $v_{j}$ is on $\pi^*$ and it has an outgoing 0-edge $e'$, it is on $\pi^*$, but not on $\pi_1^*$. Also, $v_{j}\ne v_i$, as $v_i$ does not have an outgoing 0-edge on $\pi^*$. It follows that $v_{j}\in \{u_z,u_{z+1},\ldots,u_{z'-1}\}$, and hence there is also an incoming 0-edge of $v_{j}$ on $\pi^*$. 
    Now, suppose $v_i=u_{j^h}$. As $(v_{i^1},v_i)$ is a 0-edge, all the edges on $\pi^*$ between $u_{z-1}$ and $u_{j^h}$ are 0-edges, and hence the edges on $\pi_1^*$ between them are 1-edges. Also, the edges on $\pi^*$ between $u_{j^h}$ and $u_{z'+1}$ are 1-edges, and hence the edges on $\pi_1^*$ between them are 0-edges. Thus, if $v_{j}$ is on $\pi^*$ and it has an outgoing 0-edge $e'$, it is either on $\pi^*$ or on $\pi_1^*$. Also, $v_{j}\ne v_i$, as $v_i$ does not have an outgoing 0-edge on $\pi^*$ or $\pi_1^*$. It follows that $v_{j}\in \{u_z,u_{z+1},\ldots,u_{z'-1}\}\setminus \{v_i\}$, and hence there is also a corresponding incoming 0-edge of $v_{j}$ on $\pi^*$ or $\pi_1^*$.  
    Now, suppose $e'$ is part of a cycle in $G_1'$, which was removed. Then, there is also an incoming 0-edge of $v_j$ in that cycle. Moreover, each such edge is part of a different cycle in $G_1'$. By the above argument, there are also at least $n_1$ incoming 0-edges of $v_j$ in $G^*$. 

    Now, consider the $0$-edge $e_\eta$. If $e_\eta$ is the last edge on $\pi_1$, it is in $G_1$ and thus is not on $O$, $\pi^*$, $\pi_1^*$ or on a cycle that was removed. Otherwise, $v_j=v_i$, and $e_\eta$ is the 0-edge $(v_{i^1},v_i)$, which was on $\pi^*$ and thus was not on $O$ or on a cycle that was removed. Also, $v_{i}$ does not have an outgoing 0-edge on $\pi^*$ or $\pi_1^*$. Thus, additionally considering $e_\eta$, there are at least $n_1+1$ incoming 0-edges of $v_j$ in $G^*$. 
    \end{proof}
    
    If $t=1$, the edges of $H_1$ form a perfect matching, i.e., the degree of each vertex in $H_1$ is 1. Then, $n_2 =n_r-n_1$, and the number of edges incident on the blue vertices in $C_j^*$ is at least $n_2+n_1+1=n_r-n_1+n_1+1=n_r+1$. But, this implies that $n_b\ge n_r+1$, which contradicts $1$-balance in $C_j^*$. Henceforth, we assume that $t\ge 2$.

    \begin{claim}\label{cl:edges-incident-blue-j2-1}
       Suppose all the red points in $C_{j}^*$ have degree $t$ in $H_1$ and one of the following holds: (a) $v_{j}$ is not one of $u_z,u_{z+1},\ldots,u_{z'+1}$, the degree of $q_3$ is 1 or it is not in $C_j^*$ or it does not exist, and either $b_\eta=q_2$ and the degree of $q_2$ is 2, or the degree of $b_\eta$ is 1, (b) $v_{j}$ is one of $u_z,u_{z+1},\ldots,u_{z'}$, the degree of $q_1$ is 1, the degree of $q_3$ is 1 or it is not in $C_j^*$ or it does not exist, and either $b_\eta=q_2$ and the degree of $q_2$ is 2, or the degree of $b_\eta$ is 1, and (c) $v_{j}=v_i\ne u_l$ and the degree of $q_5$ is 1, or $v_j=v_i=u_l\ne v_{j'}$, and either $b_\eta=q_2$ and the degree of $q_2$ is 2, or the degree of $b_\eta$ is 1. Then the blue points of $C_{j}^*$ can incident to at most $t\cdot n_r$ edges in $H_1$. 
    \end{claim}
    
    \begin{proof}
    First, note that if $b_\eta=q_2$ and the degree of $q_2$ is 2, there cannot be a special 0-hanging cycle for $v_i$ in $G^*$, with $q_2$ being the special vertex, that is 0-1-edge-disjoint from $\pi^*$ and $O$. This is true, as $q_2$ is already an endpoint of two edges corresponding to two edges on $O$ and $\pi_1$, which are 0-1-edge-disjoint. 
    
    For the sake of contradiction, suppose the blue points of $C_{j}^*$ are incident to at least $t\cdot n_r+1$ edges. Note that for each outgoing 0-edge of $v_j$ that is part of a cycle in $G_1'$, there is also an incoming 0-edge of $v_j$ that is part of the same cycle. Also, each such incoming 0-edge must correspond to a blue vertex in $C_j^*$. If this blue point is $q_2$, the union of $\pi_1$ and such a cycle form a special 0-hanging cycle for $v_i$ in $G^*$, with $q_2$ being the special vertex. Now, as this is a 0-hanging cycle, it is 0-1-edge-disjoint from $O$ whose edges were removed. Moreover, as $G_1$ does not contain any edge of $\pi^*$ or $\pi_1^*$, this hanging cycle is 0-1-edge-disjoint from $\pi^*$. 
    Then by our assumption, the degree of $q_2$ is 2. Thus, at most one such incoming 0-edge can correspond to $q_2$. Now, suppose the blue point is not $q_2$. Then, its degree must be 1, as otherwise the union of $\pi_1$ and such a cycle form a special 0-hanging cycle for $v_i$ in $G^*$, with the special vertex being the blue point. Now, as this is a 0-hanging cycle, it is 0-1-edge-disjoint from $O$ whose edges were removed. 
    Moreover, as $G_1$ does not contain any edge of $\pi^*$ or $\pi_1^*$, this hanging cycle is 0-1-edge-disjoint from $\pi^*$. But, this leads to a contradiction. So, the degree of such a blue point must be 1. Now, by excluding the possible outgoing 0-edge corresponding to $q_2$ and at most two outgoing 0-edges on $O$ and $\pi^*$,  there are at least $n_1-3$ other outgoing 0-edges of $v_j$ that is part of a cycle in $G_1'$. We refer to the latter outgoing 0-edges of $v_j$ as \textit{key} edges. By the above discussion, there is a unique blue point of degree 1 in $C_j^*$ corresponding to each key edge. Thus, there are at least $n_1-3$ such blue points.   
    Now, the red points in $C_j^*$ are connected to the blue points in $C_j^*$ using $n_2 =t\cdot n_r-n_1$ edges. As $H_1$ is a collection of stars, and each red point in $C_j^*$ has degree $t\ge 2$, the degree of these blue points must be 1. 

    Now, suppose (a) is true. Here, $v_{j}$ is not one of $u_z,u_{z+1},\ldots,u_{z'+1}$, and so the number of key edges is at least $n_1-2$, as no outgoing 0-edge of $v_j$ can be on $\pi^*$ in this case. In case $b_\eta=q_2$  and the degree of $q_2$ is 2, there is no special 0-hanging cycle for $v_i$ in $G^*$ as argued before. Thus, the excluded edge corresponding to $q_2$ doesn't exist. So, in this case, there are at least $n_1-1$ key edges. As $b_\eta=q_2$, $v_j=v_{j'}$. Then, $v_{j'}$ is not on $\pi^*$, as $v_j$ is not on $\pi^*$ by our assumption in this case. Thus, $v_{j'}\ne v_{i'}$, and hence $q_3$ exists and it is in $C_j^*$. Then $q_3$ has degree 1, and so the total number of blue points in $C_j^*$ of degree 1, including $q_3$ and the $n_2$ blue points connected to the red points in $C_j^*$, is at least 
    $(n_1-1)+1+n_2=t\cdot n_r$.  
    Next, suppose the degree of $b_\eta$ is 1. Thus, $b_\eta\ne q_2$. If $q_3$ exists and is in $C_j^*$, then the number of blue points in $C_j^*$ with degree 1 in $H_1$ including $b_\eta$ and $q_3$ is again at least $(n_1-2)+2+n_2=t\cdot n_r$. Otherwise, if $q_3$ exists and is not in $C_j^*$, or if $q_3$ doesn't exist, then all the outgoing 0-edges of $v_j$ except 1 are on different removed cycles. The exceptional edge is on $O$. Thus, the excluded edge corresponding to $q_2$ doesn't exist, and the number of key edges  is at least $n_1-1$. Hence, the total number of blue points with degree 1 including $b_\eta$ is again at least $(n_1-1)+1+n_2=t\cdot n_r$. As $C_j^*$ can contain at most $t\cdot n_r$ blue points due to its $t$-balance, the degree of all blue points of $C_j^*$ in $H_1$ is 1. Thus, the sum of their degree is $t\cdot n_r$. But, this contradicts the fact that they are incident to at least $t\cdot n_r+1$ edges. 
    
    Next, suppose (b) is true. Thus, $v_{j}$ is one of $u_z,u_{z+1},\ldots,u_{z'}$. As we argued before, the number of key edges is at least $n_1-3$. 
    In case $b_\eta=q_2$ { and the degree of $q_2$ is 2}, there is no special 0-hanging cycle for $v_i$ in $G^*$ as argued before. Thus, the excluded edge corresponding to $q_2$ doesn't exist. So, in this case, there are at least $n_1-2$ key edges. 
    Now, the degree of $q_1$ is 1. 
    Note that as $b_\eta=q_2$, $v_j=v_{j'}$. Thus, if $q_3$ exists, it is in $C_j^*$, and by our assumption, it has degree 1. Also, $q_1$ is distinct from $q_3$, as their degrees are 1, and the two corresponding edges are distinct. Thus the number of blue points with degree 1 including $q_1$ and $q_3$ is at least $(n_1-2)+2+n_2=t\cdot n_r$. Otherwise, if $q_3$ doesn't exist, then $v_j=v_{j'}=v_{i'}$, and so all the outgoing 0-edges of $v_j$ except 1 are on different removed cycles. The exceptional edge is on $O$. So, the number of key edges is $n_1-1$. 
    Thus, the total number of blue points of degree 1 including $q_1$ is at least 
    $(n_1-1)+1+n_2=t\cdot n_r$.  
     Next, suppose the degree of $b_\eta$ is 1. Thus, $b_\eta\ne q_2$. If $q_3$ exists and is in $C_j^*$, then the number of blue points in $C_j^*$ with degree 1 in $H_1$ including $b_\eta$, $q_1$ and $q_3$ is again at least $(n_1-3)+3+n_2=t\cdot n_r$. Otherwise, if $q_3$ exists and is not in $C_j^*$, or if $q_3$ doesn't exist, then the outgoing 0-edges of $v_j$ except 1 are on different removed cycles or on $\pi^*$. The exceptional edge is on $O$. 
     Thus, the excluded edge corresponding to $q_2$ doesn't exist, and the number of key edges is at least $n_1-2$. Hence, the total number of blue points with degree 1 including $b_\eta$ and $q_1$ is again at least $(n_1-2)+2+n_2=t\cdot n_r$. As $C_j^*$ can contain at most $t\cdot n_r$ blue points due to its $t$-balance, the degree of all blue points of $C_j^*$ in $H_1$ is 1. Thus, the sum of their degree is $t\cdot n_r$. But, this contradicts the fact that they are incident to at least $t\cdot n_r+1$ edges.

     Lastly, suppose (c) is true. Now, as $v_j=v_i$ is either a switching vertex or $u_l$, it does not have an outgoing 0-edge on $\pi^*$. If $b_\eta=q_2$ and the degree of $q_2$ is 2, then the excluded edge corresponding to $q_2$ doesn't exist as we argued before. In this case, $v_i=v_j=v_{j'}$, and thus $v_j\ne u_l$. Then an outgoing 0-edge of $v_j$ cannot be on $\pi^*$, and so the number of key edges is at least $n_1-1$. Thus, the number of blue points in $C_j^*$ with degree 1 in $H_1$ including $q_5$ is at least $(n_1-1)+1+n_2=t\cdot n_r$. 
     Now, suppose the degree of $b_\eta$ is 1. Thus, $b_\eta\ne q_2$, as $q_2$ is the special point of $O$. Now, if $v_j=v_i\ne u_l$, then again an outgoing 0-edge of $v_j$ cannot be on $\pi^*$. So, the number of key edges is at least $n_1-2$. It follows that the number of blue points in $C_j^*$ with degree 1 in $H_1$ including $q_5$ and $b_\eta$ is again at least $(n_1-2)+2+n_2=t\cdot n_r$. Lastly, if $v_j=v_i=u_l$, then an outgoing 0-edge of $v_j$ cannot be on $\pi^*$. In {this case}, $v_j\ne v_{j'}$, as $v_j=u_l$, and thus the excluded edge corresponding to $q_2$ doesn't exist. So, the number of key edges is at least $n_1-1$. Hence, the number of blue points in $C_j^*$ with degree 1 in $H_1$ including $b_\eta$ is at least $(n_1-1)+1+n_2=t\cdot n_r$. This completes the proof of the claim. 
    \end{proof}

    Next, we do an exhaustive case analysis and show that one of the cases is true. 
    First, we consider the case when $b_{\eta}= q_2$. If \textbf{$b_{\eta}= q_2$ and the degree of $q_2$ in $H_1$ is at least 3}, then (i) is true. Otherwise, if \textbf{$b_{\eta}= q_2$, the degree of $q_2$ in $H_1$ is 2, $v_{j}\ne v_i$, and $q_2 = q_3$}, then (ii) is true. Otherwise, if \textbf{$b_{\eta}= q_2$, the degree of $q_2$ in $H_1$ is 2, $v_{j}\ne v_i$, $q_2\ne q_3$, and the degree of $q_3$ is at least 2}, then (iii) is true. 
    
    Next, we consider the subcases when $v_{j}\ne v_i$ is one of $u_z,u_{z+1},\ldots,u_{z'}$. If \textbf{$b_{\eta}= q_2$, the degree of $q_2$ in $H_1$ is 2, $v_{j}$ is one of $u_z,u_{z+1},\ldots,u_{z'}$, and $q_2 = q_1$}, then (iv) is true. Otherwise, if \textbf{$b_{\eta}= q_2$, the degree of $q_2$ in $H_1$ is 2, $v_{j}$ is one of $u_z,u_{z+1},\ldots,u_{z'}$, $q_1\ne q_2$, and the degree of $q_1$ is at least 2}, then (v) is true.

    \medskip
    \noindent
    {\bf (vi) $b_{\eta}= q_2$, the degree of $q_2$ in $H_1$ is 2, $v_{j}$ is one of $u_z,u_{z+1},\ldots,u_{z'}$, $q_2\ne q_1$, $q_2\ne q_3$, and the degree of $q_1$ and $q_3$ are 1.} We claim that there is a red point in $C_{j}^*$ whose degree in $H_1$ is at most $t-1$. For the sake of contradiction to our claim, assume that all the red points in $C_{j}^*$ have degree $t$ in $H_1$. Then, $n_2 = t\cdot n_r-n_1$, and by Claim \ref{cl:incoming-cluster-j2-1}, the number of edges incident on the blue vertices in $C_{j}^*$ is at least $n_2+(n_1+1)=t\cdot n_r-n_1+n_1+1=t\cdot n_r+1$. But, by the Case (b) of Claim \ref{cl:edges-incident-blue-j2-1}, we obtain a contradiction. Hence, our assumption must be false.   

Next, we consider the subcase when $v_{j}\ne v_i=u_{z'+1}$ is not one of $u_z,u_{z+1},\ldots,u_{z'}$.  

    \medskip
    \noindent
    {\bf (vii) $b_{\eta}= q_2$, the degree of $q_2$ in $H_1$ is 2, $v_{j}$ is not one of $u_z,u_{z+1},\ldots,u_{z'+1}$, $q_2\ne q_3$, and the degree of $q_3$ is 1.} We claim that there is a red point in $C_{j}^*$ whose degree in $H_1$ is at most $t-1$. For the sake of contradiction to our claim, assume that all the red points in $C_{j}^*$ have degree $t$ in $H_1$. Then, $n_2 = t\cdot n_r-n_1$, and by Claim \ref{cl:incoming-cluster-j2-1}, the number of edges incident on the blue vertices in $C_{j}^*$ is at least $n_2+(n_1+1)=t\cdot n_r-n_1+n_1+1=t\cdot n_r+1$. But, by the Case (a) of Claim \ref{cl:edges-incident-blue-j2-1}, we obtain a contradiction. Hence, our assumption must be false.    

    Next, we consider the subcase when $v_{j}=v_i$. As $b_{\eta}= q_2$, $v_{j'}=v_j$. Thus, $v_{j'}=v_i$. Now, if $v_i=u_l$, this implies that there is a 0-path from $u_1$ to $u_l$. But, this contradicts the fact that $\pi^*$ has a switch. Thus, in this subcase, $v_i\ne u_l$. Now, if \textbf{$b_{\eta}= q_2$, the degree of $q_2$ in $H_1$ is 2, $v_{j}=v_i\ne u_l$, and the degree of $q_5$ is at least 2}, then (viii) holds. 

    \medskip
    \noindent
    {\bf (ix) $b_{\eta}= q_2$, the degree of $q_2$ in $H_1$ is 2, $v_{j}=v_i\ne u_l$, and the degree of $q_5$ is 1.} We claim that there is a red point in $C_{j}^*$ whose degree in $H_1$ is at most $t-1$. For the sake of contradiction to our claim, assume that all the red points in $C_{j}^*$ have degree $t$ in $H_1$. Then, $n_2 = t\cdot n_r-n_1$, and by Claim \ref{cl:incoming-cluster-j2-1}, the number of edges incident on the blue vertices in $C_{j}^*$ is at least $n_2+(n_1+1)=t\cdot n_r-n_1+n_1+1=t\cdot n_r+1$. But, by the Case (c) of Claim \ref{cl:edges-incident-blue-j2-1}, we obtain a contradiction. Hence, our assumption must be false.    

    Next, we consider the case when $b_{\eta}\ne q_2$. If \textbf{the degree of $b_{\eta}$ in $H_1$ is at least 2}, then (x) is true. Now, we consider the subcases when $v_{j}$ is one of $u_z,u_{z+1},\ldots,u_{z'}$. If \textbf{the degree of $b_{\eta}$ is 1, $v_{j}$ is one of $u_z,u_{z+1},\ldots,u_{z'}$, and the degree of $q_1$ is at least 2}, then (xi) is true. Otherwise, if \textbf{the degree of $b_{\eta}$ is 1, $v_{j}$ is one of $u_z,u_{z+1},\ldots,u_{z'}$, and the degree of $q_3$ is at least 2 and it is in $C_j^*$}, then (xii) is true.

    
    \medskip
    \noindent
    {\bf (xiii) The degree of $b_{\eta}$ in $H_1$ is 1, $v_{j}$ is one of $u_z,u_{z+1},\ldots,u_{z'}$, the degree of $q_1$ is 1, and the degree of $q_3$ is 1 or it is not in $C_j^*$ or it doesn't exist.} We claim that there is a red point in $C_{j}^*$ whose degree in $H_1$ is at most $t-1$. For the sake of contradiction to our claim, assume that all the red points in $C_{j}^*$ have degree $t$ in $H_1$. Then, $n_2 = t\cdot n_r-n_1$, and by Claim \ref{cl:incoming-cluster-j2-1}, the number of edges incident on the blue vertices in $C_{j}^*$ is at least $n_2+(n_1+1)=t\cdot n_r-n_1+n_1+1=t\cdot n_r+1$. But, by the Case (b) of Claim \ref{cl:edges-incident-blue-j2-1}, we obtain a contradiction. Hence, our assumption must be false.   

    Next, we consider the subcase when $v_{j}\ne v_i=u_{z'+1}$ is not one of $u_z,u_{z+1},\ldots,u_{z'}$. If \textbf{the degree of $b_{\eta}$ is 1, $v_{j}$ is not one of $u_z,u_{z+1},\ldots,u_{z'+1}$, and the degree of $q_3$ is at least 2 and it is in $C_j^*$}, then (xiv) is true. 

    \medskip
    \noindent
    {\bf (xv) The degree of $b_{\eta}$ is 1, $v_{j}$ is not one of $u_z,u_{z+1},\ldots,u_{z'+1}$, and the degree of $q_3$ is 1 or it is not in $C_j^*$ or it doesn't exist.} We claim that there is a red point in $C_{j}^*$ whose degree in $H_1$ is at most $t-1$. For the sake of contradiction to our claim, assume that all the red points in $C_{j}^*$ have degree $t$ in $H_1$. Then, $n_2 = t\cdot n_r-n_1$, and by Claim \ref{cl:incoming-cluster-j2-1}, the number of edges incident on the blue vertices in $C_{j}^*$ is at least $n_2+(n_1+1)=t\cdot n_r-n_1+n_1+1=t\cdot n_r+1$. But, by the Case (a) of Claim \ref{cl:edges-incident-blue-j2-1}, we obtain a contradiction. Hence, our assumption must be false.    

    Next, we consider the subcase when $v_{j}=v_i$. If \textbf{the degree of $b_{\eta}$ is 1, $v_{j}=v_i\ne u_l$, and the degree of $q_5$ is at least 2}, then (xvi) is true. 
    
    \medskip
    \noindent
    {\bf (xvii) The degree of $b_{\eta}$ is 1, $v_{j}=v_i$, and the degree of $q_5$ is 1 if $v_i\ne u_l$.} We claim that there is a red point in $C_{j}^*$ whose degree in $H_1$ is at most $t-1$. For the sake of contradiction to our claim, assume that all the red points in $C_{j}^*$ have degree $t$ in $H_1$. Then, $n_2 = t\cdot n_r-n_1$, and by Claim \ref{cl:incoming-cluster-j2-1}, the number of edges incident on the blue vertices in $C_{j}^*$ is at least $n_2+(n_1+1)=t\cdot n_r-n_1+n_1+1=t\cdot n_r+1$. Now, if $v_{j'}=v_j$, then $v_i\ne u_l$. Thus, either $v_{j}=v_i\ne u_l$ or $v_j=v_i=u_l\ne v_{j'}$. But, then by the Case (c) of Claim \ref{cl:edges-incident-blue-j2-1}, we obtain a contradiction. Hence, our assumption must be false.

    This completes the proof of the lemma.      
\end{proof}

\subsection{Proof of Lemma \ref{lem:2-paths-no-hanging-cycle-exists}} 
\TwoPathsNoHangingCycle*
\begin{proof}
    We prove the existence of the $0$-paths with the desired properties. The proof for $1$-paths is similar. 
    
    Let $\pi_1^*$ be the $u_l$ to $u_1$ path consists of the edges that are reverse of the edges on $\pi^*$, i.e., a $b$-edge on $\pi^*$ gives rise to a $(1-b)$-edge of $\pi_1^*$, where $b\in \{0,1\}$. Consider the graph $G_1$ constructed in the following way from $G^*$. First, remove all the $1$-edges and the edges on $\pi^*$ and $\pi_1^*$ from $G^*$. While there is a cycle in $G^*$, remove all the edges of this cycle from $G^*$ and repeat this step. When the above procedure ends, we are left with a graph without any cycle. Let us denote this graph by $G_1$. 

    Consider any maximal path $\pi_1$ in $G_1$ starting from $v_i$. Let $\pi_2$ in $G_1$ be a maximal path starting from $v_{i'}$ that is 0-1-edge-disjoint from $\pi_1$. Note that both of these paths exist, as a path may consist of a single vertex if there is no outgoing 0-edge of  $v_i$ or $v_{i'}$. Let $v_{j^1}$ be the last vertex on $\pi_1$ and $v_{j^2}$ be the last vertex on $\pi_2$. Note that $v_i$ may be same as $v_{j^1}$ and $v_{i'}$ may be same as $v_{j^2}$. As $G_1$ does not contain any edge of $\pi^*$ or $\pi_1^*$, $\pi_1$ and $\pi_2$ are 0-1-edge-disjoint from $\pi^*$. We claim that $\pi_1$ and $\pi_2$ are the desired 0-paths. 
    
    Consider the cluster $C_{j^1}^*$ corresponding to $v_{j^1}$. If $v_i=v_{j^1}$, $e_{\eta^1}=(v_{i^1},v_i)$. Otherwise, $e_{\eta^1}$ is the last edge on $\pi_1$. The edge in $E_1$ corresponding to $e_{\eta^1}$ is $\{r_{\eta^1},b_{\eta^1}\}$, where $r_{\eta^1}$ is the red point and $b_{\eta^1}$ is the blue point. As $e_{\eta^1}$ is a 0-edge, $b_{\eta^1}$ belongs to $C_{j^1}^*$. Also, consider the cluster $C_{j^2}^*$ corresponding to $v_{j^2}$. If $v_{i'}=v_{j^2}$, $e_{\eta^2}=(v_{i^2},v_{i'})$. Otherwise, $e_{\eta^2}$ is the last edge on $\pi_2$. The edge in $E_1$ corresponding to $e_{\eta^2}$ is $\{r_{\eta^2},b_{\eta^2}\}$, where $r_{\eta^2}$ is the red point and $b_{\eta^2}$ is the blue point. As $e_{\eta^2}$ is a 0-edge, $b_{\eta^2}$ belongs to $C_{j^2}^*$. 
    


    Let $n_1$ and $n_1'$ be the respective number of outgoing 0-edges in $G^*$ from $v_{j^1}$ and $v_{j^2}$. Also, let $n_r$ and $n_b$ be the number of red and blue points in $C_{j^1}^*$, respectively. Let $n_2$ be the number of edges in $E_1$ across red and blue points in $C_{j^1}^*$. Similarly, let $n_r'$ and $n_b'$ be the number of red and blue points in $C_{j^2}^*$, respectively. Also, let $n_2'$ be the number of edges in $E_1$ across red and blue points in $C_{j^2}^*$. 
    

    \begin{claim}\label{cl:incoming-cluster-j1}
        The number of incoming 0-edges of $v_{j^1}$ in $G^*$ is at least $n_1+1$. Additionally, if $v_{j^1}=v_{j^2}
        $,
        the number of such edges is at least $n_1+2$.  
    \end{claim}
    
    \begin{proof}
        We know that there is no outgoing 0-edge in $G_1$ from $v_{j^1}$; otherwise, $\pi_1$ would not be maximal. However, there may be outgoing 0-edges of $v_{j^1}$ in $G^*$. Each such edge $e'$ was removed from $G_1$. Thus, it either must have been part of a cycle in $G^*$ or on $\pi^*$ or $\pi_1^*$. Suppose $v_{j^1}$ is on $\pi^*$ or equivalently on $\pi_1^*$.         
        Then by Observation \ref{obs:range-of-b^h-path}, $v_{j^1}$ is 
         one of $u_z,u_{z+1},\ldots,u_{z'}$, where if $v_i=u_l$, then $z=j^\lambda+1, z'=l$, and if $v_i=u_{j^h}$, then $z=2$ if $h=1$ and $z=j^{h-1}+1$ otherwise, and $z'=l-1$ if $h=\lambda$ and $z'=j^{h+1}-1$ otherwise. 
          First, suppose $v_i=u_l$. As $(v_{i^1},v_i)$ is a 0-edge, all the edges on $\pi^*$ between $u_{z-1}$ and $u_{z'+1}$ are 0-edges, and hence the edges on $\pi_1^*$ between them are 1-edges. Thus, if $v_{j^1}$ is on $\pi^*$ and it has an outgoing 0-edge $e'$, it is on $\pi^*$, but not on $\pi_1^*$. Also, $v_{j^1}\ne v_i$, as $v_i$ does not have an outgoing 0-edge on $\pi^*$. It follows that $v_{j^1}\in \{u_z,u_{z+1},\ldots,u_{z'-1}\}$, and hence there is also an incoming 0-edge of $v_{j^1}$ on $\pi^*$. Now, suppose $v_i=u_{j^h}$. As $(v_{i^1},v_i)$ is a 0-edge, all the edges on $\pi^*$ between $u_{z-1}$ and $u_{j^h}$ are 0-edges, and hence the edges on $\pi_1^*$ between them are 1-edges. Also, the edges on $\pi^*$ between $u_{j^h}$ and $u_{z'+1}$ are 1-edges, and hence the edges on $\pi_1^*$ between them are 0-edges. Thus, if $v_{j^1}$ is on $\pi^*$ and it has an outgoing 0-edge $e'$, it is either on $\pi^*$ or on $\pi_1^*$. Also, $v_{j^1}\ne v_i$, as $v_i$ does not have an outgoing 0-edge on $\pi^*$ or $\pi_1^*$. It follows that $v_{j^1}\in \{u_z,u_{z+1},\ldots,u_{z'-1}\}\setminus \{v_i\}$, and hence there is also a corresponding incoming 0-edge of $v_{j^1}$ on $\pi^*$ or $\pi_1^*$.     
          
          Now, suppose $e'$ is part of a cycle in $G^*$ that was removed. Then, there is also an incoming 0-edge of $v_{j^1}$ in that cycle. Moreover, each such edge is part of a different cycle in $G^*$. By the above argument, there are also at least $n_1$ incoming 0-edges of $v_{j^1}$ in $G^*$. Now, consider the $0$-edge $e_{\eta^1}$. If $e_{\eta^1}$ is the last edge on $\pi_1$, it is in $G_1$ and thus is not on $\pi^*$, $\pi_1^*$ or on a cycle that was removed. Otherwise, $v_{j^1}=v_i$, and $e_{\eta^1}=(v_{j^1},v_i)$, which was on $\pi^*$. Thus, $e_{\eta^1}$ was not on a cycle that was removed. Also, $v_{i}$ does not have an outgoing 0-edge on $\pi^*$ or $\pi_1^*$. Thus, additionally considering $e_{\eta^1}$, there are at least $n_1+1$ incoming 0-edges of $v_{j^1}$ in $G^*$.

         Lastly, suppose $v_{j^1}=v_{j^2}$ and consider the edge $e_{\eta^2}$. 
       If $e_{\eta^2}$ is the last edge on $\pi_1$, it is in $G_1$ and thus is not on $\pi^*$, $\pi_1^*$ or on a cycle that was removed. Otherwise, $v_{j^2}=v_{i'}$, and $e_{\eta^2}=(v_{j^2},v_{i'})$, which is on $\pi_1^*$. Thus, $e_{\eta^2}$ was not on a cycle that was removed. Also, $v_{i'}$ does not have an outgoing 0-edge on $\pi^*$ or $\pi_1^*$.  Thus, additionally considering $e_{\eta^2}$, there are at least $n_1+2$ incoming 0-edges of $v_{j^1}$ in $G^*$.
    \end{proof}

    \begin{claim}\label{cl:incoming-cluster-j2}
        The number of incoming 0-edges of $v_{j^2}$ in $G^*$ is at least $n_1'+1$. 
    \end{claim}
    \begin{proof}
       We know that there is no outgoing 0-edge in $G_1$ from $v_{j^2}$; otherwise, $\pi_2$ would not be maximal. However, there may be outgoing 0-edges of $v_{j^2}$ in $G^*$. Each such edge $e'$ was removed from $G_1$. Thus, it either must have been part of a cycle in $G^*$, or on $\pi_1$, $\pi^*$, or $\pi_1^*$. Suppose $e'$ is part of $\pi_1$. Then, {there is an incoming 0-edge of $v_{j^2}$ that is on $\pi_1$}, unless $v_{j^2}$ is the first vertex on this path. In the latter case, $v_{j^2}=v_{i}$. Then, $v_{i}\ne u_l$, as otherwise there is a zero switch path from $v_{i'}=u_1$ to $u_l$ in $G^*$, which is a contradiction. Thus, $v_i=u_{j^h}$, and there is an incoming 0-edge $(v_{i^1},v_i)$, which is on $\pi^*$ and not on a cycle in $G_1$ that was removed.       

       Suppose $v_{j^2}$ is on $\pi^*$ or equivalently on $\pi_1^*$, and it has an outgoing 0-edge $e'$.  Then by Observation \ref{obs:range-of-b^h-path}, $v_{j^2}$ is one of $u_y,u_{y+1},\ldots,u_{y'}$ where if $v_{i'}=u_1$, then $y=1, y'=j^1-1$, and if $v_{i'}=u_{j^h}$, then $y=2$ if $h=1$ and $y=j^{h-1}+1$ otherwise, and $y'=l-1$ if $h=\lambda$ and $y'=j^{h+1}-1$ otherwise. 
       Now, if $v_{i'}=u_1$, all the edges on $\pi^*$ between $u_{1}$ and $u_{j^1}$ are 1-edges by our assumption, and hence the edges on $\pi_1^*$ between them are 0-edges. Thus, if $v_{j^2}$ is on $\pi^*$ and it has an outgoing 0-edge $e'$, $e'$ must be on $\pi_1^*$, but not on $\pi^*$. It follows that $v_{j^2}\in \{u_2,\ldots,u_{j^1-1}\}$, and hence there is also an incoming 0-edge of $v_{j^2}$ on $\pi_1^*$. 
       Next suppose $v_{i'}=u_{j^h}$ for $h\ge 1$. As $(v_{i^1},v_i)$ is a 0-edge, all the edges on $\pi^*$ between $u_{y-1}$ and $u_{j^h}$ are 0-edges, and hence the edges on $\pi_1^*$ between them are 1-edges. Also, the edges on $\pi^*$ between $u_{j^h}$ and $u_{y'+1}$ are 1-edges, and hence the edges on $\pi_1^*$ between them are 0-edges. Thus, if $v_{j^2}$ is on $\pi^*$ and it has an outgoing 0-edge $e'$, it is either on $\pi^*$ or on $\pi_1^*$. Also, $v_{j^2}\ne v_i$, as $v_i$ does not have an outgoing 0-edge on $\pi^*$ or $\pi_1^*$. It follows that $v_{j^2}\in \{u_y,u_{y+1},\ldots,u_{y'-1}\}\setminus \{v_i\}$, and hence there is also a corresponding incoming 0-edge of $v_{j^2}$ on $\pi^*$ or $\pi_1^*$.

       Next, suppose the outgoing 0-edge $e'$ of $v_{j^1}$ is part of a cycle in $G^*$ that was removed. Then, there is also an incoming 0-edge of $v_{j^2}$ in that cycle. Moreover, each such edge is part of a different cycle in $G^*$. By the above argument, there are also at least $n_1'$ incoming 0-edges of $v_{j^2}$ in $G^*$. 

       Now, consider the $0$-edge $e_{\eta^2}$. If $e_{\eta^2}$ is the last edge on $\pi_2$, it is in $G_1$ and thus is not on $\pi^*$, $\pi_1^*$ or on a cycle that was removed. Otherwise, $v_{j^2}=v_{i'}$, and $e_{\eta^2}=(v_{j^2},v_{i'})$, which is on $\pi_1^*$. Thus, $e_{\eta^2}$ was not on a cycle that was removed. Also, $v_{i'}$ does not have an outgoing 0-edge on $\pi^*$ or $\pi_1^*$.  Thus, additionally considering $e_{\eta^2}$, there are at least $n_1'+1$ incoming 0-edges of $v_{j^2}$ in $G^*$.       
    \end{proof}

    If $t=1$, the edges of $H_1$ form a perfect matching, i.e., the degree of each vertex in $H_1$ is 1. Then, $n_2' =n_r'-n_1'$, and the number of edges incident on the blue vertices in $C_{j^2}^*$ is at least $n_2'+n_1'+1=n_r'-n_1'+n_1'+1=n_r'+1$. But, this implies that $n_b'\ge n_r'+1$, which contradicts $1$-balance in $C_{j^2}^*$. Henceforth, we assume that $t\ge 2$. 

    Recall that if $v_i=u_l$, then $z=j^\lambda+1, z'=l$, and if $v_i=u_{j^h}$, then $z=2$ if $h=1$ and $z=j^{h-1}+1$ otherwise, and $z'=l-1$ if $h=\lambda$ and $z'=j^{h+1}-1$ otherwise. 

    \begin{claim}\label{cl:edges-incident-blue-j1}
        Suppose all the red points in $C_{j^1}^*$ have degree $t$ in $H_1$ and one of the following holds: (a) $v_{j^1}$ is not one of $u_z,u_{z+1},\ldots,u_{z'-1}$, and (b) $v_{j^1}$ is one of $u_z,u_{z+1},\ldots,u_{z'-1}$, and the degree of $b_{\eta^1}$ is 1 or the degree of the blue point in $C_{j^1}^*$ corresponding to $\pi^*$ is 1. Then the blue points of $C_{j^1}^*$ can incident to at most $t\cdot n_r$ edges. 
    \end{claim}
    \begin{proof}
     For the sake of contradiction, suppose the blue points of $C_{j^1}^*$ are incident to at least $t\cdot n_r+1$ edges. Note that for each outgoing 0-edge of  $v_{j^1}$ that is part of a cycle in $G^*$ and was removed, there is also an incoming 0-edge of $v_{j^1}$ that is part of the same cycle. However, each such incoming 0-edge must correspond to a degree 1 blue vertex in $C_{j^1}^*$. Otherwise, the union of $\pi_1$ and such a cycle form a special 0-hanging cycle for $v_{i}$ that is 0-1-edge-disjoint from $\pi^*$, which is a contradiction. Note that there are at least $n_1-1$ such outgoing 0-edges of $v_{j^1}$, as one of $\pi^*$ and $\pi_1^*$ contains at most one outgoing 0-edge of $v_{j^1}$. Thus, there are $n_1-1$ such blue points with degree 1. Additionally, if $v_{j^1}$ is not one of $u_z,u_{z+1},\ldots,u_{z'-1}$, all outgoing 0-edges of $v_{j^1}$ are on removed cycles, and thus there are $n_1$ such blue points with degree 1. Now, the red points in $C_{j^1}^*$ are connected to the blue points in $C_{j^1}^*$ using $n_2 =t\cdot n_r-n_1$ edges. As $H_1$ is a collection of stars, and each red point in $C_{j^1}^*$ has degree $t\ge 2$, the degree of these blue points must be 1. 
    
    If $v_{j^1}$ is not one of $u_z,u_{z+1},\ldots,u_{z'-1}$, the number of blue points in $C_{j^1}^*$ with degree 1 in $H_1$ is at least $n_1+n_2=t\cdot n_r$. As $C_{j^1}^*$ can contain at most $t\cdot n_r$ blue points due to its $t$-balance, the degree of all blue points of $C_{j^1}^*$ in $H_1$ is 1. Thus, the sum of their degree is $t\cdot n_r$. But, this contradicts the fact that they are incident to at least $t\cdot n_r+1$ edges. Thus our claim must be true. 

    Otherwise, $v_{j^1}$ is one of $u_z,u_{z+1},\ldots,u_{z'-1}$. In this case, $e_{\eta^1}$ is the last edge of $\pi_1$ or on $\pi^*$, and thus is not part of a cycle in $G^*$ that was deleted. If the degree of $b_{\eta^1}$ is 1, it is not connected to a red point of $C_{j^1}^*$ in $H_1$. Thus, the number of blue points in $C_{j^1}^*$ with degree 1 in $H_1$ is at least $(n_1-1)+n_2+1=t\cdot n_r$. Otherwise, the degree of the blue point in $C_{j^1}^*$ corresponding to $\pi^*$ is 1. Thus, the number of blue points in $C_{j^1}^*$ with degree 1 in $H_1$ is again at least $(n_1-1)+n_2+1=t\cdot n_r$. Thus our claim must be true in this case as well.     
    \end{proof}

    Recall that if $v_{i'}=u_1$, then $y=1, y'=j^1-1$, and if $v_{i'}=u_{j^h}$, then $y=2$ if $h=1$ and $y=j^{h-1}+1$ otherwise, and $y'=l-1$ if $h=\lambda$ and $y'=j^{h+1}-1$ otherwise. 

    \begin{claim}\label{cl:edges-incident-blue-j2}
       Suppose all the red points in $C_{j^2}^*$ have degree $t$ in $H_1$ and one of the following holds: (a) $v_{j^2}$ is not one of $u_y,u_{y+1},\ldots,u_{y'-1}$ and the degree of $b_{\eta^2}$ is 1, and (b) $v_{j^2}$ is one of $u_y,u_{y+1},\ldots,u_{y'-1}$, the degree of $b_{\eta^2}$ is 1 and the degree of the blue point in $C_{j^2}^*$ corresponding to $\pi^*$ is 1. Then the blue points of $C_{j^2}^*$ can incident to at most $t\cdot n_r'$ edges. 
    \end{claim}
    
    \begin{proof}
    For the sake of contradiction, suppose the blue points of $C_{j^2}^*$ are incident to at least $t\cdot n_r'+1$ edges. Note that for each outgoing 0-edge of  $v_{j^2}$ that is part of a cycle in $G^*$ that was removed, there is also an incoming 0-edge of $v_{j^2}$ that is part of the same cycle. However, each such incoming 0-edge must correspond to a degree 1 blue vertex in $C_{j^2}^*$. Otherwise, the union of $\pi_2$ and such a cycle form a special 0-hanging cycle for $v_{i'}$ that is 0-1-edge-disjoint from $\pi^*$, which is a contradiction. Note that there are at least $n_1'-2$ such outgoing 0-edges of $v_{j^2}$, as one of $\pi^*$ and $\pi_1^*$,  and $\pi_1$ may contain at most two outgoing 0-edge of $v_{j^2}$. Thus, there are at least $n_1'-2$ such blue points with degree 1. 
    Additionally, if $v_{j^2}$ is not one of $u_y,u_{y+1},\ldots,u_{y'-1}$, no outgoing 0-edges of $v_{j^2}$ are on $\pi^*$ and $\pi_1^*$, and thus there are $n_1'-1$ such blue points with degree 1. Now, the red points in $C_{j^2}^*$ are connected to the blue points in $C_{j^2}^*$ using $n_2' =t\cdot n_r'-n_1'$ edges. As $H_1$ is a collection of stars, and each red point in $C_{j^2}^*$ has degree $t\ge 2$, the degree of these blue points must be 1. 
    
    If $v_{j^2}$ is not one of $u_y,u_{y+1},\ldots,u_{y'-1}$, the number of blue points in $C_{j^2}^*$ with degree 1 in $H_1$ including $b_{\eta^2}$ is at least $(n_1'-1)+1+n_2'=t\cdot n_r'$. As $C_{j^2}^*$ can contain at most $t\cdot n_r'$ blue points due to its $t$-balance, the degree of all blue points of $C_{j^2}^*$ in $H_1$ is 1. Thus, the sum of their degrees is $t\cdot n_r'$. But, this contradicts the fact that they are incident to at least $t\cdot n_r'+1$ edges. Thus our claim must be true. 

    Otherwise, $v_{j^2}$ is one of $u_y,u_{y+1},\ldots,u_{y'-1}$. In this case, $e_{\eta^2}$ is the last edge of $\pi_2$ or on $\pi_1^*$, and thus is not part of a cycle in $G^*$ that was deleted. If the degree of $b_{\eta^2}$ is 1, it is not connected to a red point of $C_{j^2}^*$ in $H_1$. Thus, the number of blue points in $C_{j^2}^*$ with degree 1 in $H_1$ including $b_{\eta^2}$ and the blue point in $C_{j^2}^*$ corresponding to $\pi^*$ is at least $(n_1'-2)+2+n_2'=t\cdot n_r'$. 
    \end{proof}

    With the above claims in hand, we do a case-by-case analysis to complete the proof. 
    Now, if \textbf{the degree of both $b_{\eta^1}$ and $b_{\eta^2}$ in $H_1$ is at least 2 and $b_{\eta^1}\ne b_{\eta^2}$},
    then (i) holds.  
    Also, if \textbf{the degree of both $b_{\eta^1}$ and $b_{\eta^2}$ in $H_1$ is at least 2, $b_{\eta^1}= b_{\eta^2}$, $v_{j^1}$ is one of $u_z,u_{z+1},\ldots,u_{z'-1}$, and the degree of the blue point in $C_{j^1}^*$ corresponding to $\pi^*$ is at least 2}, then (ii) holds. So, assume that (i) and (ii) are not true.

    \medskip
    \noindent
    {\bf (iii) The degree of both $b_{\eta^1}$ and $b_{\eta^2}$ in $H_1$ are at least 2, $b_{\eta^1}= b_{\eta^2}$, $v_{j^1}$ is one of $u_z,u_{z+1},\ldots,u_{z'-1}$, and the degree of the blue point in $C_{j^1}^*$ corresponding to $\pi^*$ is 1.} We claim that there is a red point in $C_{j^1}^*$ whose degree in $H_1$ is at most $t-1$.  For the sake of contradiction to our claim, assume that all the red points in $C_{j^1}^*$ have degree $t$ in $H_1$. Then, $n_2 = t\cdot n_r-n_1$, and by Claim \ref{cl:incoming-cluster-j1}, the number of edges incident on the blue vertices in $C_{j^1}^*$ is at least $n_2+(n_1+1)=t\cdot n_r-n_1+n_1+1=t\cdot n_r+1$. But, as the degree of the blue point in $C_{j^1}^*$ corresponding to $\pi^*$ is 1, due to Item (b), this contradicts Claim \ref{cl:edges-incident-blue-j1}.

    \medskip
    \noindent
    {\bf (iv) The degree of both $b_{\eta^1}$ and $b_{\eta^2}$ in $H_1$ are at least 2, $b_{\eta^1}= b_{\eta^2}$, $v_{j^1}$ is not one of $u_z,u_{z+1},\ldots,u_{z'-1}$.} We claim that there is a red point in $C_{j^1}^*$ whose degree in $H_1$ is at most $t-1$.  For the sake of contradiction to our claim, assume that all the red points in $C_{j^1}^*$ have degree $t$ in $H_1$. Then, $n_2 = t\cdot n_r-n_1$, and by Claim \ref{cl:incoming-cluster-j1}, the number of edges incident on the blue vertices in $C_{j^1}^*$ is at least $n_2+(n_1+1)=t\cdot n_r-n_1+n_1+1=t\cdot n_r+1$. But, as $v_{j^1}$ is not one of $u_z,u_{z+1},\ldots,u_{z'-1}$, due to Item (a), this contradicts Claim \ref{cl:edges-incident-blue-j1}.

    \medskip
    \noindent
    {\bf (v) The degree of $b_{\eta^1}$ is at least 2, $v_{j^2}$ is not one of $u_y,u_{y+1},\ldots,u_{y'-1}$, and the degree of $b_{\eta^2}$ in $H_1$ is 1.} We claim that there is a red point in $C_{j^2}^*$ whose degree in $H_1$ is at most $t-1$.  For the sake of contradiction to our claim, assume that all the red points in $C_{j^2}^*$ have degree $t$ in $H_1$. Then, $n_2' = t\cdot n_r'-n_1'$, and by Claim \ref{cl:incoming-cluster-j2}, the number of edges incident on the blue vertices in $C_{j^2}^*$ is at least $n_2'+(n_1'+1)=t\cdot n_r'-n_1'+n_1'+1=t\cdot n_r'+1$. But, as the degree of $b_{\eta^2}$ is 1, due to Item (a), this contradicts Claim \ref{cl:edges-incident-blue-j2}.    

    Now, if \textbf{the degree of $b_{\eta^1}$ is at least 2, $v_{j^2}$ is one of $u_y,u_{y+1},\ldots,u_{y'-1}$, the degree of $b_{\eta^2}$ is 1, and the degree of the blue point in $C_{j^2}^*$ corresponding to $\pi^*$ is at least 2}, then (vi) holds. 

    \medskip
    \noindent
    {\bf (vii) The degree of $b_{\eta^1}$ is at least 2, $v_{j^2}$ is one of $u_y,u_{y+1},\ldots,u_{y'-1}$, the degree of $b_{\eta^2}$ is 1, and the degree of the blue point in $C_{j^2}^*$ corresponding to $\pi^*$ is
 1.} We claim that there is a red point in $C_{j^2}^*$ whose degree in $H_1$ is at most $t-1$.  For the sake of contradiction to our claim, assume that all the red points in $C_{j^2}^*$ have degree $t$ in $H_1$. Then, $n_2' = t\cdot n_r'-n_1'$, and by Claim \ref{cl:incoming-cluster-j2}, the number of edges incident on the blue vertices in $C_{j^2}^*$ is at least $n_2'+(n_1'+1)=t\cdot n_r'-n_1'+n_1'+1=t\cdot n_r'+1$. But, as the degree of $b_{\eta^2}$ is 1, due to Item (b), this contradicts Claim \ref{cl:edges-incident-blue-j2}.

    \medskip
    \noindent
    {\bf (viii) The degree of $b_{\eta^1}$ in $H_1$ is 1 and the degree of $b_{\eta^2}$ in $H_1$ is at least 2.} We claim that there is a red point in $C_{j^1}^*$ whose degree in $H_1$ is at most $t-1$. For the sake of contradiction to our claim, assume that all the red points in $C_{j^1}^*$ have degree $t$ in $H_1$. Then, $n_2 = t\cdot n_r-n_1$, and by Claim \ref{cl:incoming-cluster-j1}, the number of edges incident on the blue vertices in $C_{j^1}^*$ is at least $n_2+(n_1+1)=t\cdot n_r-n_1+n_1+1=t\cdot n_r+1$. But, as the degree of $b_{\eta^1}$ is 1, this contradicts Claim \ref{cl:edges-incident-blue-j1}.

    \medskip
    \noindent
    {\bf (ix) $j^1\ne j^2$, the degree of both $b_{\eta^1}$ and $b_{\eta^2}$ are 1 in $H_1$, and $v_{j^2}$ is not one of $u_y,u_{y+1},\ldots,u_{y'-1}$.} First, we claim that there is a red point in $C_{j^1}^*$ whose degree in $H_1$ is at most $t-1$. For the sake of contradiction to our claim, assume that all the red points in $C_{j^1}^*$ have degree $t$ in $H_1$. Then, $n_2 = t\cdot n_r-n_1$, and by Claim \ref{cl:incoming-cluster-j1}, the number of edges incident on the blue vertices in $C_{j^1}^*$ is at least $n_2+(n_1+1)=t\cdot n_r-n_1+n_1+1=t\cdot n_r+1$. But, as the degree of $b_{\eta^1}$ is 1, this contradicts Claim \ref{cl:edges-incident-blue-j1}.

    Next, we claim that there is a red point in $C_{j^2}^*$ whose degree in $H_1$ is at most $t-1$.  For the sake of contradiction to our claim, assume that all the red points in $C_{j^2}^*$ have degree $t$ in $H_1$. Then, $n_2' = t\cdot n_r'-n_1'$, and by Claim \ref{cl:incoming-cluster-j2}, the number of edges incident on the blue vertices in $C_{j^2}^*$ is at least $n_2'+(n_1'+1)=t\cdot n_r'-n_1'+n_1'+1=t\cdot n_r'+1$. But, as the degree of $b_{\eta^2}$ is 1, due to Item (a), this contradicts Claim \ref{cl:edges-incident-blue-j2}.  

    Now, if \textbf{$j^1\ne j^2$, the degree of both $b_{\eta^1}$ and $b_{\eta^2}$ are 1 in $H_1$, $v_{j^2}$ is one of $u_y,u_{y+1},\ldots,u_{y'-1}$, and the degree of the blue point in $C_{j^2}^*$ corresponding to $\pi^*$ is at least 2}, then (x) holds. 
    
    \medskip
    \noindent
    {\bf (xi) $j^1\ne j^2$, the degree of both $b_{\eta^1}$ and $b_{\eta^2}$ are 1 in $H_1$, $v_{j^2}$ is one of $u_y,u_{y+1},\ldots,u_{y'-1}$, and the degree of the blue point in $C_{j^2}^*$ corresponding to $\pi^*$ is 1.} First, we claim that there is a red point in $C_{j^1}^*$ whose degree in $H_1$ is at most $t-1$. For the sake of contradiction to our claim, assume that all the red points in $C_{j^1}^*$ have degree $t$ in $H_1$. Then, $n_2 = t\cdot n_r-n_1$, and by Claim \ref{cl:incoming-cluster-j1}, the number of edges incident on the blue vertices in $C_{j^1}^*$ is at least $n_2+(n_1+1)=t\cdot n_r-n_1+n_1+1=t\cdot n_r+1$. But, as the degree of $b_{\eta^1}$ is 1, this contradicts Claim \ref{cl:edges-incident-blue-j1}.

    Next, we claim that there is a red point in $C_{j^2}^*$ whose degree in $H_1$ is at most $t-1$.  For the sake of contradiction to our claim, assume that all the red points in $C_{j^2}^*$ have degree $t$ in $H_1$. Then, $n_2' = t\cdot n_r'-n_1'$, and by Claim \ref{cl:incoming-cluster-j2}, the number of edges incident on the blue vertices in $C_{j^2}^*$ is at least $n_2'+(n_1'+1)=t\cdot n_r'-n_1'+n_1'+1=t\cdot n_r'+1$. But, as the degree of $b_{\eta^2}$ is 1, due to Item (b), this contradicts Claim \ref{cl:edges-incident-blue-j2}. 



    \medskip
    \noindent    
    {\bf (xii) $j^1= j^2$, the degree of both $b_{\eta^1}$ and $b_{\eta^2}$ are 1 in $H_1$.} We claim that the sum of the degrees of the red points of $C_{j^2}^*=C_{j^1}^*$ in $H_1$ is at most $t\cdot n_r-2$. For the sake of contradiction, assume that this sum is at least $t\cdot n_r-1$. Then, $n_2 = (t\cdot n_r-1)-n_1$, and by Claim \ref{cl:incoming-cluster-j1}, the number of edges incident on the blue vertices in $C_{j^1}^*$ is at least $n_2+(n_1+2)=t\cdot n_r-1-n_1+n_1+2=t\cdot n_r+1$. But, as the degree of $b_{\eta^1}$ is 1, this contradicts Claim \ref{cl:edges-incident-blue-j1}. 

    This completes the proof of the lemma.  

\end{proof}

\section{The Algorithm for Balanced Sum-of-Radii Clustering}
\label{sec:balanced}

In this section, we prove Theorem \ref{thm:ell-groups}. Recall that we are given $\ell$ disjoint groups $P_1,\ldots, P_{\ell}$ having $n$ points in total in a metric space $(\Omega=\cup_{i=1}^\ell P_i,d)$, such that $|P_1|=|P_2|=\ldots =|P_\ell|$. 

Our algorithm is as follows.  
\medskip
\noindent
\begin{tcolorbox}
\noindent
{\bf The Algorithm.}\\
\textbf{1.} For each $2\le i\le \ell$, construct a graph $G_i=(V_i,E_i)$ where $V_i=P_1\cup P_i$ and $E_i=\{\{p,q\}\mid p\in P_1, q\in P_i\}$. Define the weight function $w_i$ such that for each edge $e=\{p,q\}$, $w_i(e)=d(p,q)$. Compute a minimum-weight (w.r.t. $w_i$) perfect matching $M_i$ of $G_i$. For each $p\in P_1$, let $S_p$ be the union of $\{p\}$ and the points from $P_2,\ldots,P_\ell$ that are matched to $p$ in $M=\cup_{i=2}^\ell M_i$.  

\medskip
\noindent
\textbf{2.} 
Construct an edge-weighted graph $G'$ in the following way: For each $p\in \Omega$, add a vertex to $G'$; For each $p\in P_1$, add a vertex corresponding to $S_p$ to $G'$, which we also call by $S_p$; For each $p,q \in \Omega$, add the edge $\{p,q\}$ to $G'$ with weight $d(p,q)$; For all $p'\in \Omega$ and $p\in P_1$, add the edge $\{p',S_p\}$ to $G'$ with weight $\max_{q\in S_p} d(p',q)$. Let $d'$ be the shortest path metric in $G'$. Construct the metric space $(\Omega',d')$ where $\Omega'$ is the subset of vertices $\{S_p\mid p\in P_1\}$ in $G'$. 


\medskip
\noindent
\textbf{3.} Compute a sum of radii clustering $X=\{X_1,\ldots,X_k\}$ of the points in $\Omega'$ using the Algorithm of Buchem et al.~\cite{buchem20243+} (with $\Omega'$ also being the candidate set of centers). 

\medskip
\noindent
\textbf{4.} Compute a clustering $X'$ of the points in $\cup_{i=1}^\ell P_i$ using $X$ in the following way. For each cluster $X_i$, add the cluster  $\cup_{q\in S_p\mid S_p\in X_i} \{q\}$ to $X'$. Return $X'$.  
\end{tcolorbox}

Next, we analyze the algorithm. First, we have the following observation. 

\begin{observation}
    $X'$ is a balanced clustering of $\cup_{i=1}^\ell P_i$. 
\end{observation}

\begin{proof}
    Note that any cluster of $X'$ is a union of sets $S_p$ such that $p\in P_1$. As each such $S_p$ contains exactly one point from $P_i$ for $1\le i\le \ell$, this cluster is 1-balanced. Hence, $X'$ is a balanced clustering of $\cup_{i=1}^\ell P_i$. 
\end{proof}

Next, we analyze the approximation factor. 
Let $\mathcal{C}^*=\{C_1^*, C_2^*, \ldots, C_k^*\}$ be a fixed optimal balanced clustering. 
We have the following lemma whose proof is given later. 

\begin{lemma}\label{lem:approx-cluster-of-star-points-32-factor-1}
    Consider the clustering $X$ of $\Omega'$ constructed in Step 3 of the algorithm. Then cost$_{(\Omega',d')}(X)\le \mathbf{(60+\epsilon)}\cdot \sum_{i=1}^k r_{(\Omega,d)}(C_i^*)$. 
\end{lemma}

\begin{corollary}
     Consider the clustering $X'$ of $\cup_{i=1}^\ell P_i$ constructed in Step 3 of the algorithm. Then cost$_{(\Omega,d)}(X')\le \mathbf{(180+\epsilon)}\cdot \sum_{i=1}^k r_{(\Omega,d)}(C_i^*)$. Thus, our algorithm is a $\mathbf{(180+\epsilon)}$-approximation algorithm. 
\end{corollary}

\begin{proof}
    We claim that cost$_{(\Omega,d)}(X')\le 3\cdot $cost$_{(\Omega',d')}(X)$. Then the corollary follows by Lemma \ref{lem:approx-cluster-of-star-points-32-factor-1}. Consider any cluster $X_i$ of $X$ and the cluster $X_i'$ in $X'$ constructed from it. Let $S_p$ in $\Omega'$ be the center of $X_i$. Now, for any $S_q\in X_i$, $d'(S_p,S_q)$ is the weight of a shortest path in $G'$ between $S_p$ and $S_q$. Let $p'\in \Omega$ be the successor of $S_p$ on such a shortest path. So, $d'(S_p,S_q)\ge d'(S_p,p')=\max_{y\in S_p} d(y,p')\ge d(p,p')$. The equality follows by the definition of $d'$ and the fact that $d$ is a metric. Then, $d'(p,S_q)\le d'(p,p')+d'(p',S_p)+d'(S_p,S_q)\le 3\cdot d'(S_p,S_q)$. The first inequality is due to triangle inequality. Now, $d'(p,S_q)=\max_{q'\in S_q} d(p,q')$. It follows that the ball in $(\Omega,d)$ centered at $p\in \Omega$ and having radius $3\cdot r_{(\Omega',d')}(X_i)$ contains all the points in $X_i'$. Hence, $r_{(\Omega,d)}(X_i')\le 3\cdot r_{(\Omega',d')}(X_i)$ and the claim follows. 
\end{proof}

\subsection{Proof of Lemma \ref{lem:approx-cluster-of-star-points-32-factor-1}}
In the following, we are going to prove Lemma \ref{lem:approx-cluster-of-star-points-32-factor-1}. Consider the union of matchings $M=\cup_{i=2}^\ell M_i$ computed in Step 1. Also, consider the optimal clusters in $\mathcal{C}^*$. We construct a new clustering $\hat{\mathcal{C}}=\{\hat{C}_1, \hat{C}_2, \ldots, \hat{C}_\kappa\}$ by merging some clusters in $\mathcal{C}^*$, where $1\le \kappa\le k$. Initially, we set $\hat{\mathcal{C}}$ to $\mathcal{C}^*$. For each edge $\{p,q\}$ of $M$ such that $p\in \hat{C}_i, q\in \hat{C}_j$ and $i\ne j$, replace $\hat{C}_i, \hat{C}_j$ in $\hat{\mathcal{C}}$ by their union and denote it by $\hat{C}_i$ as well. 

When the above merging procedure ends, by renaming the indexes, let $\hat{\mathcal{C}}=\{\hat{C}_1, \hat{C}_2, \ldots, \hat{C}_\kappa\}$ be the new clustering. Then, we have the following lemma. 

\begin{observation}\label{obs:S-is-in-some-C-i-1}
    For any $p\in P_1$, all the points of $S_p$ are contained in a set $\hat{C}_i$ for some $1\le i\le \kappa$.  
\end{observation}

\begin{observation}\label{obs:star-points-in-C_i-balanced}
    Consider any set $S_p$ and the cluster $\hat{C}_i \supseteq S_p$ with center $c\in \Omega$. Then, for any $q \in S_p$, $d(q,c)\le r_{(\Omega,d)}(\hat{C}_i)$. 
\end{observation}

\begin{proof}
Because the points of $S_p$ are in $\hat{C}_i$, the farthest any point in $S_p$ can be from $c$ is not more than the farthest any point in $\hat{C}_i$ is from $c$. So, for any point $q$ in $S_p$, the distance between $q$ and $c$ is less than or equal to the maximum distance between any point in $\hat{C}_i$ and $c$, which we denote as $r_{(\Omega,d)}(\hat{C}_i)$.   
\end{proof}
    
\begin{observation}\label{obs:d'-p-to-c-is-at-most-radius-1}
    Consider the point $S_p$ in $\Omega'$ corresponding to a set $S_p$ and the cluster $\hat{C}_i \supseteq S_p$ with center $c\in \Omega$. Then, $d'(S_p,c)\le r_{(\Omega,d)}(\hat{C}_i)$. 
\end{observation}

\begin{proof}
    By definition, $d'(S_p,c) = \max_{q \in S_p} d(q,c)$. By Observation \ref{obs:star-points-in-C_i-balanced}, $d(q,c)\le r_{(\Omega,d)}(\hat{C}_i)$. It follows that $d'(S_p,c) \le r_{(\Omega,d)}(\hat{C}_i)$.       
\end{proof}

Consider the clustering $\mathcal{C}'=\{C_1',\ldots,C'_\kappa\}$ of $\Omega'$ defined in the following way. For each set $S_p$, identify the cluster $\hat{C}_i$ in $\hat{\mathcal{C}}$ that contains all the points of $S_p$. By Observation \ref{obs:S-is-in-some-C-i-1}, such an index $i$ exists. Assign the point $S_p$ in $\Omega'$ corresponding to the set $S_p$ to $C_i'$.  

\begin{lemma}\label{lem:costofC'}
    cost$_{(\Omega',d')}(\mathcal{C}')\le 2\cdot $ cost$_{(\Omega,d)}(\hat{\mathcal{C}})$. 
\end{lemma}

\begin{proof}
    First, we claim that $r_{(\Omega,d')}(C_i')\le r_{(\Omega,d)}(\hat{C}_i)$ for all $1\le i\le \kappa$. Let $c$ in $\Omega$ be the center of $\hat{C}_i$. Consider any set $S_p$ such that its corresponding point in $\Omega'$ is in $C_i'$. Then, by Observation \ref{obs:d'-p-to-c-is-at-most-radius-1}, $d'(S_p,c)\le r_{(\Omega,d)}(\hat{C}_i)$. As $c$ is in $\Omega$, it follows that, $r_{(\Omega,d')}(C_i')$ is at most $r_{(\Omega,d)}(\hat{C}_i)$. 

    Now, consider any two $S_p,S_q \in C_i'$ for $p,q \in P_1$. By the above claim, $d'(S_p,S_q)\le 2\cdot r_{(\Omega,d)}(\hat{C}_i)$. Thus, for each such cluster $C_i'$, we can set a point $S_p\in C_i'$ as the center. As $S_p\in \Omega'$, $r_{(\Omega',d')}\le 2\cdot r_{(\Omega,d)}(\hat{C}_i)$. Summing over all clusters $C_i'$, we obtain the lemma.   
\end{proof}

We will prove the following {lemma}. 

\begin{lemma}\label{lem:optimal-cluster-of-stars-10-factor}
    cost$_{(\Omega,d)}(\hat{\mathcal{C}})\le 10\cdot \sum_{i=1}^k r_{(\Omega,d)}(C_i^*)$. 
\end{lemma}

Then, Lemma \ref{lem:approx-cluster-of-star-points-32-factor-1} follows by Lemma \ref{lem:optimal-cluster-of-stars-10-factor} and \ref{lem:costofC'} noting that the Algorithm of Buchem et al.~\cite{buchem20243+} yields a 3-factor approximation of the optimal clustering. In the rest of this section, we prove Lemma \ref{lem:optimal-cluster-of-stars-10-factor}. 

\subsection{Proof of Lemma \ref{lem:optimal-cluster-of-stars-10-factor}}
For simplicity of notation, we drop $(\Omega,d)$ from $r_{(\Omega,d)}(.)$, as henceforth centers are always assumed to be in $\Omega$. Let us consider any fixed $\hat{C}_i$, and suppose it is constructed by merging the clusters $C_{i_1}^*,C_{i_2}^*,\ldots,C_{i_\tau}^*$. It is sufficient to {prove that} $r(\hat{C}_i)\le 10\cdot \sum_{j=1}^\tau r(C_{i_j}^*)$. For simplicity of notation, we rename $\hat{C}_i$ to $\hat{C}$, and $C_{i_1}^*,C_{i_2}^*,\ldots,C_{i_\tau}^*$ to $C_{1}^*,C_{2}^*,\ldots,C_{\tau}^*$.

In the following, we construct an edge-weighted, directed multi-graph $G_1^*=(V_1^*,E_1^*)$ in the following manner. $G_1^*$ has a vertex $v_j$ corresponding to each cluster $C_j^*$, where $1\le j\le \tau$. There is an edge $e=(v_i,v_j)$ from $v_i$ to $v_j$, $i\ne j$, for each $p\in P_1\cap C_i^*$ and $q\in P_z\cap C_j^*$ such that $\{p,q\}$ is in $M$ and $2\le z\le \ell$. We refer to such an edge as a $0$-edge of color $z$. The weight $\omega_e$ of the edge $e$ is $d(p,q)$. Similarly, there is a $1$-edge of color $z$, $e=(v_i,v_j)$, from $v_i$ to $v_j$ for each $p\in P_z\cap C_i^*$ and $q\in P_1\cap C_j^*$ such that $\{p,q\}$ is in $M$ and $2\le z\le \ell$. The weight $\omega_e$ of the edge $e$ is $d(p,q)$. 
For each $2\le z\le \ell$, let $G_z^*=(V_z^*,E_z^*)$ be the subgraph of $G_1^*$ induced by the color $z$ edges. The 0-in-degree of a vertex $v\in V_z^*$ is the number of incoming $0$-edges to $v$ in $E_z^*$. The 0-out-degree of a vertex $v\in V_z^*$ is the number of outgoing $0$-edges from $v$ in $E_z^*$. 
 
A directed path (or simply a path) $\pi =\{u_1,\ldots,u_l\}$ from $u_1$ to $u_l$ in $G_1^*$ is a sequence of distinct vertices such that $(u_i,u_{i+1})$ is in $G_1^*$ for all $1\le i\le l-1$. Two consecutive edges $e_1=(u_i,u_{i+1}),e_2=(u_{i+1},u_{i+2})$ on $\pi$ are said to form a \textit{color-switch} if they have different colors. 
We say that the color-switch happens at $u_{i+1}$ and it is the corresponding color-switching vertex. 
A directed cycle is formed from $\pi$ by adding the edge $(u_l,u_1)$ (if any) with it.  

\begin{observation}
    For any two vertices $v_i,v_j\in V^*$, there is a directed path from $v_i$ to $v_j$ in $G_1^*$. 
\end{observation}

The above observation shows that $G_1^*$ is a connected graph. However, $G_z^*$ is not-necessarily connected. For a vertex $v\in V_z^*$, denote the connected component in $G_z^*$ that it is in by comp$(v,z)$. Note that there is no edge in $E_z^*$ across any two components.     

Consider any two vertices $v_\alpha$ and $v_\beta$ of $G_1^*$. Let $\pi_1^*=\{v_\alpha=u_1,\ldots,u_l=v_\beta\}$ be a directed path from $v_\alpha$ to $v_\beta$ having the minimum number of color-switches. 
We prove the following lemma. 

\begin{lemma}\label{lem:cost-of-pi-1}
    $\sum_{e\in \pi_1^*} \omega_{e}\le 8\cdot \sum_{i=1}^\tau r(C_i^*)$. 
\end{lemma}

Then, similar to the $t$-balanced case, Lemma \ref{lem:optimal-cluster-of-stars-10-factor} follows. 

Let $j^1 < j^2 <\ldots < j^\lambda$ be the indexes of the vertices on $\pi_1^*=\{u_1,\ldots,u_l\}$ where the color-switches occur. Note that $j^1 > 1, j^\lambda < l$. For our convenience, we denote $u_1$ by $u_{j^0}$ and $u_l$ by $u_{j^{\lambda+1}}$. As the color of the edges do not change between two consecutive color-switching vertices, we have the following observation. 

\begin{observation}
    For $0 \le h\le \lambda$, suppose $u_{j^h}$ is in comp$(u_{j^h},z)$ of $G_z^*$ for some $2\le z\le \ell$. Then, $u_{j^{h+1}}$ is also in the same component comp$(u_{j^h},z)$ of $G_z^*$. 
\end{observation}

Fix any $0\le h\le \lambda$. Let $z^h$ be the color of the edges on $\pi_1^*$ between $u_{j^h}$ and $u_{j^{h+1}}$. Also, let $\pi(h)$ be a path in $G_{z^h}^*$ from $u_{j^h}$ to $u_{j^{h+1}}$ having the minimum number of switches as defined in Section \ref{sec:t-balanced}. 

\begin{lemma}\label{lem:0-switch}
    For any $0\le h\le \lambda$, $\pi(h)$ is either a $0$-path or a $1$-path, i.e., it does not have any switch.  
\end{lemma}

\begin{proof}
    We prove that there is a path in $G_{z^h}^*$ from $u_{j^h}$ to $u_{j^{h+1}}$ that does not have any switch. First, we claim that the 0-in-degree and 0-out-degree of any vertex $v_j$ are equal in $G_{z^h}^*$. Let $n_1$ be the number of points in $P_1\cap C_j^*$ and $n_2$ be the number of points among these that are matched in $M_{z^h}$ with the points in $P_{z^h}\cap C_j^*$. Thus, there are $n_1-n_2$ outgoing $0$-edges of $v_j$ in $G_{z^h}^*$. As $C_j^*$ is 1-balanced, it also has $n_1-n_2$ points from $P_{z^h}$ that are matched with points of $P_1$ outside of $C_j^*$. Thus, $v_j$ must also have exactly $n_1-n_2$ incoming $0$-edges. 
    
    Consider any directed path $\pi$ in $G_{z^h}^*$ from $u_{j^h}$ to $u_{j^{h+1}}$. For the sake of contradiction, suppose $\pi$ has at least one switch and let $v_s$ be the first switching vertex. Wlog, let the first edge on $\pi$ be a $0$-edge. We prove that there is a $1$-path in $G_{z^h}^*$ from $u_{j^h}$ to $v_s$. This path along with the portion of $\pi$ from $v_s$ to $u_{j^{h+1}}$ shows the existence of a path with a strictly lesser number of switches than that of $\pi$. But, this is a contradiction, and so $\pi$ cannot have a switch. 

    Consider the graph $G_1$ constructed in the following way from $G_{z^h}^*$. First, remove all the $1$-edges and the edges of $\pi$ between $u_{j^h}$ and $v_s$ from $G_{z^h}^*$. This decreases the 0-in-degree of $v_s$ by 1 and the 0-out-degree of $u_{j^h}$ by 1. But, the difference between the 0-in-degree and 0-out-degree of all other vertices remain the same. Now, while there is a cycle in $G_{z^h}^*$, remove all the edges of this cycle from $G_{z^h}^*$ and repeat this step. When the above procedure ends, we are left with a graph without any cycle. Let us denote this graph by $G_1$. 
    
    Note that after the removal of a cycle from $G_{z^h}^*$, the difference between the 0-in-degree and 0-out-degree of any vertex does not change. In particular, the 0-in-degree of $v_s$ is one less than its 0-out-degree in $G_1$. The 0-out-degree of $u_{j^h}$ is one lesser than its 0-in-degree in $G_1$. Moreover, the 0-in-degree and 0-out-degree of any other vertex are equal in $G_1$. It follows that $v_s$ has at least one outgoing $0$-edge $(v_s,v_j)$ in $G_1$. If $v_j$ is $u_{j^h}$, we have found a $0$-path from $v_s$ to $u_{j^h}$. Otherwise, $v_j$ has the same 0-in-degree and 0-out-degree. We consider an outgoing $0$-edge of $v_j$, and repeat this process of visiting a new vertex. As all the vertices except $v_s$ and $u_{j^h}$ have same 0-in-degree and 0-out-degree, and $G_1$ does not have a cycle, this process stops when we reach to $u_{j^h}$. Also, each vertex is visited at most once, and thus the process must stop after a finite number of iterations. Now, once the process stops, we obtain a $0$-path from $v_s$ to $u_{j^h}$. Taking the reverse $1$-edges of this path in $G_{z^h}^*$, we obtain the desired $1$-path from $u_{j^h}$ to $v_s$. This completes the proof of the lemma.     
\end{proof}

Now, note that $\{C_1^*\cap (P_1\cup P_{z^h}), C_2^*\cap (P_1\cup P_{z^h}), \ldots, C_k^*\cap (P_1\cup P_{z^h})\}$ is a clustering of the points of $P_1\cup P_{z^h}$ such that for each $1\le i\le k$, $r(C_i^*\cap (P_1\cup P_{z^h}))\le r(C_i^*)$. Hence, by Lemma \ref{lem:cost-of-pi} and \ref{lem:0-switch}, we have the following observation. 

\begin{observation}\label{obs:cost-of-pi(h)}
    $\sum_{e\in \pi(h)} \omega_{e}\le 4\cdot \sum_{v_j\in \text{comp}(u_{j^h},z^h)} r(C_j^*)$. 
\end{observation}

\begin{lemma}\label{lem:disjoint-comp}
    For any $0\le x,y\le \lambda$ such that $|x-y|\ge 2$, there is no common vertex between \text{comp}$(u_{j^x},z^x)$ and \text{comp}$(u_{j^y},z^y)$. 
\end{lemma}

\begin{proof}
    Wlog, assume $x < y$. As $|x-y|\ge 2$, there is a color-switching vertex between $u_{j^x}$ and $u_{j^y}$ in $\pi_1^*$. Then, by the definition of $\pi_1^*$, for all paths in $G_1^*$ from $u_{j^x}$ to $u_{j^{y+1}}$, there are at least 2 color-switches. Suppose there is a common vertex $v_j$ between \text{comp}$(u_{j^x},z^x)$ and \text{comp}$(u_{j^y},z^y)$. Then, there is a directed path $\pi_1$ in $G_{z^x}^*$ from $u_{j^x}$ to $v_j$ and a directed path $\pi_2$ in $G_{z^y}^*$ from $v_j$ to $u_{j^{y+1}}$. It follows that the path obtained by the concatenation of $\pi_1$ and $\pi_2$ is a path in $G_1^*$ from $u_{j^x}$ to $u_{j^{y+1}}$ with exactly one switch. But, this is a contradiction, and hence the lemma follows. 
\end{proof}

By the above lemma, \text{comp}$(u_{j^x},z^x)$ and \text{comp}$(u_{j^y},z^y)$ can have a common vertex only if $|x-y|\le 1$. Hence, 

\begin{align*}
    \sum_{e\in \pi_1^*} \omega_{e} &\le \sum_{h=0}^{\lambda} \sum_{e\in \pi(h)} \omega_{e}\\
    &\le \sum_{h=0}^{\lambda} 4\cdot \sum_{v_j\in \text{comp}(u_{j^h},z^h)} r(C_j^*) \tag{By Observation \ref{obs:cost-of-pi(h)}}\\
    &\le 8 \cdot \sum_{j=1}^\tau r(C_j^*). \tag{By Lemma \ref{lem:disjoint-comp}}
\end{align*}

This completes the proof of Lemma \ref{lem:cost-of-pi-1}. 

\section{Conclusion and Open Questions}
\label{sec:conclusion}
In this work, we designed poly-time constant-approximations for both $(t,k)$-fair sum-of-radii with two groups and $(1,k)$-fair sum-of-radii with $\ell \ge 2$ groups. One of our main contributions is a novel cluster-merging-based analysis technique that might be of independent interest. We have not paid any particular attention to optimizing the approximation factors. Indeed, there are large gaps between the achieved factors and the best-known approximation bound for vanilla sum-of-radii in polynomial time. Achieving small constant factors is an interesting question. Moreover, obtaining a poly-time constant-approximation for $(t,k)$-fair median/means remains open. One promising direction is to investigate whether our cluster-merging-based analysis technique helps in this case. 

Note that it is natural to study generalizations of $(t,k)$-fair clustering with an arbitrary $\ell \ge 2$ number of groups. One such generalization is fair representational clustering as defined in the introduction. Obtaining any poly-time approximation for such a generalization remains an open question. It might be helpful to consider this problem in restricted settings, such as unweighted graph metrics. 

One might also be interested in $(t,k)$-fair clustering with non-integer $t$. Our understanding in that case is limited. Some pathological examples arise in this case due to the balance parameter not being an integer. 
For example, if $t=1+(1/I)$ for an integer $I\ge 2$, and we have $I$ red points and $I+1$ blue points, then there is only one possible fair clustering which has a single cluster containing the whole point set. Indeed, the tuple $(I,I+1)$ cannot be divided into $(1+(1/I))$-balanced non-zero integer tuples $(I_1,I_1'),\ldots,(I_k,I_k')$ for $k\ge 2$ such that $\sum_{i=1}^k I_i=I$ and $\sum_{i=1}^k I_i'=I+1$. This is true, as if $I_i < I$, to maintain $t$-balance $I_i'\le (1+(1/I))\cdot I_i=I_i+(I_i/I)$, so $I_i'\le I_i$ as it is an integer. Similarly $I_i'\ge (I/(I+1))\cdot I_i=(1-(1/(I+1))\cdot I_i = I_i-(I_i/(I+1))$, so $I_i'\ge I_i$ as it is an integer. But, then $I_i'=I_i$ and so $\sum_{i=1}^k I_i= \sum_{i=1}^k I_i'$, which is a contradiction. In general, setting the parameter $t$ not to be an integer heavily restricts the number of possible fair clusterings. For this reason, it is not clear whether such a model has any practical advantage over the one when $t$ is an integer. We note that the above example is also a pathological case for the \textit{representation-preserving} model where one wants to preserve the red-blue ratio of the dataset exactly in every cluster \cite{bercea2019cost}.




\bibliographystyle{plainurl}
\bibliography{lipics-v2021-sample-article}

\begin{thebibliography}{10}

\bibitem{aamand2023constant}
Anders Aamand, Justin~Y. Chen, Allen Liu, Sandeep Silwal, Pattara Sukprasert, Ali Vakilian, and Fred Zhang.
\newblock Constant approximation for individual preference stable clustering.
\newblock In {\em Advances in Neural Information Processing Systems (NeurIPS)}, 2023.

\bibitem{abbasi2020fair}
Mohsen Abbasi, Aditya Bhaskara, and Suresh Venkatasubramanian.
\newblock Fair clustering via equitable group representations.
\newblock In {\em Proceedings of the Conference on Fairness, Accountability, and Transparency (FAccT)}, page 504–514, 2021.

\bibitem{ahmadian_et_al:LIPIcs.ICALP.2016.69}
Sara Ahmadian and Chaitanya Swamy.
\newblock {Approximation Algorithms for Clustering Problems with Lower Bounds and Outliers}.
\newblock In Ioannis Chatzigiannakis, Michael Mitzenmacher, Yuval Rabani, and Davide Sangiorgi, editors, {\em 43rd International Colloquium on Automata, Languages, and Programming (ICALP 2016)}, volume~55 of {\em Leibniz International Proceedings in Informatics (LIPIcs)}, pages 69:1--69:15, Dagstuhl, Germany, 2016. Schloss Dagstuhl -- Leibniz-Zentrum f{\"u}r Informatik.
\newblock URL: \url{https://drops.dagstuhl.de/entities/document/10.4230/LIPIcs.ICALP.2016.69}, \href {https://doi.org/10.4230/LIPIcs.ICALP.2016.69} {\path{doi:10.4230/LIPIcs.ICALP.2016.69}}.

\bibitem{anegg2020technique}
Georg Anegg, Haris Angelidakis, Adam Kurpisz, and Rico Zenklusen.
\newblock A technique for obtaining true approximations for {$k$}-center with covering constraints.
\newblock In {\em International conference on integer programming and combinatorial optimization}, pages 52--65. Springer, 2020.

\bibitem{AryaGKMMP-SIAMJ04}
Vijay Arya, Naveen Garg, Rohit Khandekar, Adam Meyerson, Kamesh Munagala, and Vinayaka Pandit.
\newblock Local search heuristics for k-median and facility location problems.
\newblock {\em {SIAM} J. Comput.}, 33(3):544--562, 2004.

\bibitem{backurs2019scalable}
Arturs Backurs, Piotr Indyk, Krzysztof Onak, Baruch Schieber, Ali Vakilian, and Tal Wagner.
\newblock Scalable fair clustering.
\newblock In {\em International Conference on Machine Learning}, pages 405--413, 2019.

\bibitem{bandyapadhyay2024polynomial}
Sayan Bandyapadhyay, Eden Chlamt{\'a}{\v{c}}, Yury Makarychev, and Ali Vakilian.
\newblock A polynomial-time approximation for pairwise fair $ k $-median clustering.
\newblock {\em arXiv preprint arXiv:2405.10378}, 2024.

\bibitem{bandyapadhyay2019constant}
Sayan Bandyapadhyay, Tanmay Inamdar, Shreyas Pai, and Kasturi Varadarajan.
\newblock A constant approximation for colorful {$k$}-center.
\newblock In {\em 27th Annual European Symposium on Algorithms (ESA 2019)}. Schloss Dagstuhl-Leibniz-Zentrum fuer Informatik, 2019.

\bibitem{DBLP:conf/compgeom/BandyapadhyayL023a}
Sayan Bandyapadhyay, William Lochet, and Saket Saurabh.
\newblock {FPT} constant-approximations for capacitated clustering to minimize the sum of cluster radii.
\newblock In Erin~W. Chambers and Joachim Gudmundsson, editors, {\em 39th International Symposium on Computational Geometry, SoCG 2023, June 12-15, 2023, Dallas, Texas, {USA}}, volume 258 of {\em LIPIcs}, pages 12:1--12:14. Schloss Dagstuhl - Leibniz-Zentrum f{\"{u}}r Informatik, 2023.

\bibitem{DBLP:conf/isaac/BandyapadhyayV16}
Sayan Bandyapadhyay and Kasturi~R. Varadarajan.
\newblock Approximate clustering via metric partitioning.
\newblock In Seok{-}Hee Hong, editor, {\em 27th International Symposium on Algorithms and Computation, {ISAAC} 2016, December 12-14, 2016, Sydney, Australia}, volume~64 of {\em LIPIcs}, pages 15:1--15:13. Schloss Dagstuhl - Leibniz-Zentrum f{\"{u}}r Informatik, 2016.
\newblock \href {https://doi.org/10.4230/LIPIcs.ISAAC.2016.15} {\path{doi:10.4230/LIPIcs.ISAAC.2016.15}}.

\bibitem{banerjee2024novel}
Sandip Banerjee, Yair Bartal, Lee-Ad Gottlieb, and Alon Hovav.
\newblock Novel properties of hierarchical probabilistic partitions and their algorithmic applications.
\newblock In {\em 2024 IEEE 65th Annual Symposium on Foundations of Computer Science (FOCS)}, pages 1724--1767. IEEE, 2024.

\bibitem{banerjee2025improved}
Sandip Banerjee, Yair Bartal, Lee-Ad Gottlieb, and Alon Hovav.
\newblock Improved fixed-parameter bounds for min-sum-radii and diameters k-clustering and their fair variants.
\newblock In {\em Proceedings of the AAAI Conference on Artificial Intelligence}, volume~39, pages 15481--15488, 2025.

\bibitem{DBLP:journals/algorithmica/BehsazS15}
Babak Behsaz and Mohammad~R. Salavatipour.
\newblock On minimum sum of radii and diameters clustering.
\newblock {\em Algorithmica}, 73(1):143--165, 2015.
\newblock \href {https://doi.org/10.1007/s00453-014-9907-3} {\path{doi:10.1007/s00453-014-9907-3}}.

\bibitem{bera2019fair}
Suman Bera, Deeparnab Chakrabarty, Nicolas Flores, and Maryam Negahbani.
\newblock Fair algorithms for clustering.
\newblock In {\em Advances in Neural Information Processing Systems}, pages 4954--4965, 2019.

\bibitem{bercea2019cost}
Ioana~O Bercea, Martin Gro{\ss}, Samir Khuller, Aounon Kumar, Clemens R{\"o}sner, Daniel~R Schmidt, and Melanie Schmidt.
\newblock On the cost of essentially fair clusterings.
\newblock In {\em Approximation, Randomization, and Combinatorial Optimization. Algorithms and Techniques (APPROX/RANDOM 2019)}. Schloss Dagstuhl-Leibniz-Zentrum fuer Informatik, 2019.

\bibitem{bohm2020fair}
Matteo B{\"o}hm, Adriano Fazzone, Stefano Leonardi, and Chris Schwiegelshohn.
\newblock Fair clustering with multiple colors.
\newblock {\em arXiv preprint arXiv:2002.07892}, 2020.

\bibitem{brubach2021fairness}
B~Brubach, D~Chakrabarti, J~Dickerson, A~Srinivasan, and L~Tsepenekas.
\newblock Fairness, semi-supervised learning, and more: A general framework for clustering with stochastic pairwise constraints.
\newblock In {\em Proc. Thirty-Fifth AAAI Conference on Artificial Intelligence (AAAI)}, 2021.

\bibitem{buchem20243+}
Moritz Buchem, Katja Ettmayr, Hugo~KK Rosado, and Andreas Wiese.
\newblock A ($3+\epsilon$)-approximation algorithm for the minimum sum of radii problem with outliers and extensions for generalized lower bounds.
\newblock In {\em Proceedings of the 2024 Annual ACM-SIAM Symposium on Discrete Algorithms (SODA)}, pages 1738--1765. SIAM, 2024.

\bibitem{carta2024fpt}
Lena Carta, Lukas Drexler, Annika Hennes, Clemens R\"{o}sner, and Melanie Schmidt.
\newblock {FPT Approximations for Fair $k$-Min-Sum-Radii}.
\newblock In Juli\'{a}n Mestre and Anthony Wirth, editors, {\em 35th International Symposium on Algorithms and Computation (ISAAC 2024)}, volume 322 of {\em Leibniz International Proceedings in Informatics (LIPIcs)}, pages 16:1--16:18, Dagstuhl, Germany, 2024. Schloss Dagstuhl -- Leibniz-Zentrum f{\"u}r Informatik.
\newblock URL: \url{https://drops.dagstuhl.de/entities/document/10.4230/LIPIcs.ISAAC.2024.16}, \href {https://doi.org/10.4230/LIPIcs.ISAAC.2024.16} {\path{doi:10.4230/LIPIcs.ISAAC.2024.16}}.

\bibitem{CharikarP04}
Moses Charikar and Rina Panigrahy.
\newblock Clustering to minimize the sum of cluster diameters.
\newblock {\em J. Comput. Syst. Sci.}, 68(2):417--441, 2004.
\newblock URL: \url{http://dx.doi.org/10.1016/j.jcss.2003.07.014}, \href {https://doi.org/10.1016/j.jcss.2003.07.014} {\path{doi:10.1016/j.jcss.2003.07.014}}.

\bibitem{chen2016matroid}
Danny~Z Chen, Jian Li, Hongyu Liang, and Haitao Wang.
\newblock Matroid and knapsack center problems.
\newblock {\em Algorithmica}, 75(1):27--52, 2016.

\bibitem{chen2024parameterized}
Xianrun Chen, Dachuan Xu, Yicheng Xu, and Yong Zhang.
\newblock Parameterized approximation algorithms for sum of radii clustering and variants.
\newblock In {\em Proceedings of the AAAI Conference on Artificial Intelligence}, volume~38, pages 20666--20673, 2024.

\bibitem{chen2019proportionally}
Xingyu Chen, Brandon Fain, Liang Lyu, and Kamesh Munagala.
\newblock Proportionally fair clustering.
\newblock In {\em International Conference on Machine Learning}, pages 1032--1041, 2019.

\bibitem{chierichetti2017fair}
Flavio Chierichetti, Ravi Kumar, Silvio Lattanzi, and Sergei Vassilvitskii.
\newblock Fair clustering through fairlets.
\newblock In {\em Advances in Neural Information Processing Systems}, pages 5029--5037, 2017.

\bibitem{chiplunkar2020solve}
Ashish Chiplunkar, Sagar Kale, and Sivaramakrishnan~Natarajan Ramamoorthy.
\newblock How to solve fair $k$-center in massive data models.
\newblock In {\em Proceedings of the International Conference on Machine Learning (ICML)}, pages 1877--1886, 2020.

\bibitem{chlamtavc2022approximating}
Eden Chlamt{\'a}{\v{c}}, Yury Makarychev, and Ali Vakilian.
\newblock Approximating fair clustering with cascaded norm objectives.
\newblock In {\em Proceedings of the 2022 Annual ACM-SIAM Symposium on Discrete Algorithms (SODA)}, pages 2664--2683, 2022.

\bibitem{DaiMV22}
Zhen Dai, Yury Makarychev, and Ali Vakilian.
\newblock Fair representation clustering with several protected classes.
\newblock In {\em FAccT '22: 2022 {ACM} Conference on Fairness, Accountability, and Transparency, Seoul, Republic of Korea, June 21 - 24, 2022}, pages 814--823. {ACM}, 2022.

\bibitem{DBLP:journals/njc/DoddiMRTW00}
Srinivas Doddi, Madhav~V. Marathe, S.~S. Ravi, David~Scot Taylor, and Peter Widmayer.
\newblock Approximation algorithms for clustering to minimize the sum of diameters.
\newblock {\em Nord. J. Comput.}, 7(3):185--203, 2000.

\bibitem{drexler2023approximating}
Lukas Drexler, Annika Hennes, Abhiruk Lahiri, Melanie Schmidt, and Julian Wargalla.
\newblock Approximating fair $k$-min-sum-radii in euclidean space.
\newblock In {\em International Workshop on Approximation and Online Algorithms}, pages 119--133. Springer, 2023.

\bibitem{filtser2024fpt}
Arnold Filtser and Ameet Gadekar.
\newblock Fpt approximations for capacitated sum of radii and diameters.
\newblock {\em arXiv preprint arXiv:2409.04984}, 2024.

\bibitem{friggstad_et_al:LIPIcs.ESA.2022.56}
Zachary Friggstad and Mahya Jamshidian.
\newblock {Improved Polynomial-Time Approximations for Clustering with Minimum Sum of Radii or Diameters}.
\newblock In Shiri Chechik, Gonzalo Navarro, Eva Rotenberg, and Grzegorz Herman, editors, {\em 30th Annual European Symposium on Algorithms (ESA 2022)}, volume 244 of {\em Leibniz International Proceedings in Informatics (LIPIcs)}, pages 56:1--56:14, Dagstuhl, Germany, 2022. Schloss Dagstuhl -- Leibniz-Zentrum f{\"u}r Informatik.
\newblock URL: \url{https://drops-dev.dagstuhl.de/entities/document/10.4230/LIPIcs.ESA.2022.56}, \href {https://doi.org/10.4230/LIPIcs.ESA.2022.56} {\path{doi:10.4230/LIPIcs.ESA.2022.56}}.

\bibitem{gabow1983efficient}
Harold~N Gabow.
\newblock An efficient reduction technique for degree-constrained subgraph and bidirected network flow problems.
\newblock In {\em Proceedings of the fifteenth annual ACM symposium on Theory of computing}, pages 448--456, 1983.

\bibitem{GhadiriSV21}
Mehrdad Ghadiri, Samira Samadi, and Santosh~S. Vempala.
\newblock Socially fair $k$-means clustering.
\newblock In Madeleine~Clare Elish, William Isaac, and Richard~S. Zemel, editors, {\em FAccT '21: 2021 {ACM} Conference on Fairness, Accountability, and Transparency, Virtual Event / Toronto, Canada, March 3-10, 2021}, pages 438--448. {ACM}, 2021.

\bibitem{ghadiri2022constant}
Mehrdad Ghadiri, Mohit Singh, and Santosh~S Vempala.
\newblock Constant-factor approximation algorithms for socially fair {$k$}-clustering.
\newblock {\em arXiv preprint arXiv:2206.11210}, 2022.

\bibitem{DBLP:journals/algorithmica/GibsonKKPV10}
Matt Gibson, Gaurav Kanade, Erik Krohn, Imran~A. Pirwani, and Kasturi~R. Varadarajan.
\newblock On metric clustering to minimize the sum of radii.
\newblock {\em Algorithmica}, 57(3):484--498, 2010.
\newblock \href {https://doi.org/10.1007/s00453-009-9282-7} {\path{doi:10.1007/s00453-009-9282-7}}.

\bibitem{DBLP:journals/siamcomp/GibsonKKPV12}
Matt Gibson, Gaurav Kanade, Erik Krohn, Imran~A. Pirwani, and Kasturi~R. Varadarajan.
\newblock On clustering to minimize the sum of radii.
\newblock {\em {SIAM} J. Comput.}, 41(1):47--60, 2012.
\newblock \href {https://doi.org/10.1137/100798144} {\path{doi:10.1137/100798144}}.

\bibitem{gonzalez1985clustering}
Teofilo~F Gonzalez.
\newblock Clustering to minimize the maximum intercluster distance.
\newblock {\em Theoretical Computer Science}, 38:293--306, 1985.

\bibitem{gupta2022lp}
Swati Gupta, Jai Moondra, and Mohit Singh.
\newblock Which $l_p$ norm is the fairest? approximations for fair facility location across all ``$p$'', 2022.
\newblock \href {https://arxiv.org/abs/2211.14873} {\path{arXiv:2211.14873}}.

\bibitem{DBLP:journals/dm/HeggernesL06}
Pinar Heggernes and Daniel Lokshtanov.
\newblock Optimal broadcast domination in polynomial time.
\newblock {\em Discret. Math.}, 306(24):3267--3280, 2006.
\newblock \href {https://doi.org/10.1016/j.disc.2006.06.013} {\path{doi:10.1016/j.disc.2006.06.013}}.

\bibitem{DBLP:journals/algorithmica/HenzingerLM20}
Monika Henzinger, Dariusz Leniowski, and Claire Mathieu.
\newblock Dynamic clustering to minimize the sum of radii.
\newblock {\em Algorithmica}, 82(11):3183--3194, 2020.
\newblock \href {https://doi.org/10.1007/s00453-020-00721-7} {\path{doi:10.1007/s00453-020-00721-7}}.

\bibitem{hotegni2023approximation}
Sedjro~Salomon Hotegni, Sepideh Mahabadi, and Ali Vakilian.
\newblock Approximation algorithms for fair range clustering.
\newblock In {\em International Conference on Machine Learning}, pages 13270--13284. PMLR, 2023.

\bibitem{DBLP:conf/esa/0002V20}
Tanmay Inamdar and Kasturi~R. Varadarajan.
\newblock Capacitated sum-of-radii clustering: An {FPT} approximation.
\newblock In Fabrizio Grandoni, Grzegorz Herman, and Peter Sanders, editors, {\em 28th Annual European Symposium on Algorithms, {ESA} 2020, September 7-9, 2020, Pisa, Italy (Virtual Conference)}, volume 173 of {\em LIPIcs}, pages 62:1--62:17. Schloss Dagstuhl - Leibniz-Zentrum f{\"{u}}r Informatik, 2020.
\newblock \href {https://doi.org/10.4230/LIPIcs.ESA.2020.62} {\path{doi:10.4230/LIPIcs.ESA.2020.62}}.

\bibitem{DBLP:conf/innovations/Jaiswal0Y24}
Ragesh Jaiswal, Amit Kumar, and Jatin Yadav.
\newblock {FPT} approximation for capacitated sum of radii.
\newblock In Venkatesan Guruswami, editor, {\em 15th Innovations in Theoretical Computer Science Conference, {ITCS} 2024, January 30 to February 2, 2024, Berkeley, CA, {USA}}, volume 287 of {\em LIPIcs}, pages 65:1--65:21. Schloss Dagstuhl - Leibniz-Zentrum f{\"{u}}r Informatik, 2024.

\bibitem{jia2020fair}
Xinrui Jia, Kshiteej Sheth, and Ola Svensson.
\newblock Fair colorful $k$-center clustering.
\newblock In {\em International Conference on Integer Programming and Combinatorial Optimization}, pages 209--222. Springer, 2020.

\bibitem{jung2019center}
Christopher Jung, Sampath Kannan, and Neil Lutz.
\newblock A center in your neighborhood: Fairness in facility location.
\newblock In {\em Proceedings of the Symposium on Foundations of Responsible Computing (FORC)}, page 5:1–5:15, 2020.

\bibitem{KanungoMNPSW04}
Tapas Kanungo, David~M. Mount, Nathan~S. Netanyahu, Christine~D. Piatko, Ruth Silverman, and Angela~Y. Wu.
\newblock A local search approximation algorithm for k-means clustering.
\newblock {\em Comput. Geom.}, 28(2-3):89--112, 2004.

\bibitem{kleindessner2019fair}
Matth{\"a}us Kleindessner, Pranjal Awasthi, and Jamie Morgenstern.
\newblock Fair $k$-center clustering for data summarization.
\newblock In {\em 36th International Conference on Machine Learning, ICML 2019}, pages 5984--6003. International Machine Learning Society (IMLS), 2019.

\bibitem{krishnaswamy2018constant}
Ravishankar Krishnaswamy, Shi Li, and Sai Sandeep.
\newblock Constant approximation for k-median and k-means with outliers via iterative rounding.
\newblock In {\em Proceedings of the 50th Annual ACM SIGACT Symposium on Theory of Computing}, pages 646--659, 2018.

\bibitem{makarychev2021approximation}
Yury Makarychev and Ali Vakilian.
\newblock Approximation algorithms for socially fair clustering.
\newblock In {\em Conference on Learning Theory (COLT)}, pages 3246--3264. PMLR, 2021.

\bibitem{micha2020proportionally}
Evi Micha and Nisarg Shah.
\newblock Proportionally fair clustering revisited.
\newblock In {\em International Colloquium on Automata, Languages, and Programming (ICALP)}, 2020.

\bibitem{negahbani2021better}
Maryam Negahbani and Deeparnab Chakrabarty.
\newblock Better algorithms for individually fair $k$-clustering.
\newblock {\em Advances in Neural Information Processing Systems (NeurIPS)}, 34:13340--13351, 2021.

\bibitem{rajabi2023computing}
Fatemeh Rajabi-Alni and Alireza Bagheri.
\newblock Computing a many-to-many matching with demands and capacities between two sets using the hungarian algorithm.
\newblock {\em Journal of mathematics}, 2023(1):7761902, 2023.

\bibitem{schmidt2019fair}
Melanie Schmidt, Chris Schwiegelshohn, and Christian Sohler.
\newblock Fair coresets and streaming algorithms for fair k-means.
\newblock In {\em International Workshop on Approximation and Online Algorithms}, pages 232--251. Springer, 2019.

\bibitem{vakilian2022improved}
Ali Vakilian and Mustafa Yal\c{c}{\i}ner.
\newblock Improved approximation algorithms for individually fair clustering.
\newblock In {\em International Conference on Artificial Intelligence and Statistics (AISTATS)}, pages 8758--8779. PMLR, 2022.

\end{thebibliography}

\end{document}